\documentclass[a4paper,11pt]{article}
\usepackage{jheppub_better}
\usepackage{ragged2e}
\usepackage[T1]{fontenc} 
\usepackage{float}
\usepackage{csquotes}
\usepackage{comment}
\MakeOuterQuote{"}
\setcitestyle{open={[},close={]}}  
\usepackage[skip=\baselineskip,parfill]{parskip}
\usepackage[table, usenames,dvipsnames]{xcolor}
\usepackage[coloredthm]{mathtam}
\usepackage{graphics}

\newtheorem{result}{Result}
\usepackage[capitalise]{cleveref}
\crefname{theorem}{Thm.}{Thms.}
\crefname{corollary}{Cor.}{Cors.}
\crefname{bclaim}{Claim}{Claims}
\crefname{proposition}{Prop.}{Props.}
\crefname{section}{Sec.}{Secs.}
\crefname{appendix}{App.}{Apps.}
\crefname{axiom}{Axiom}{Axioms}

\usepackage{titlesec}
\titlespacing*{\section}{0pt}{5.5ex plus 1ex minus .2ex}{4.3ex plus .2ex}
\titlespacing{\subsection}{0pt}{5.5ex plus 1ex minus .2ex}{4.3ex plus .2ex}
\titlespacing{\subsubsection}{0pt}{5.5ex plus 1ex minus .2ex}{4.3ex plus .2ex}

\usepackage{hyperref} 
 \hypersetup{
     colorlinks= true,
     linkcolor= RoyalBlue,
     filecolor= blue,
     citecolor = RoyalBlue, 
     urlcolor=cyan 
     }
\makeatletter
\gdef\@fpheader{}
\makeatother
\usetikzlibrary{fit}
\tikzset{highlight/.style={rectangle,
fill=red!15,
rounded corners = 0.5 mm,
inner sep=1pt,
fit=#1}}
\usetikzlibrary{calc,decorations,calligraphy}
\tikzset{whitefill/.style={rectangle,
fill=white,
rounded corners = 0.5 mm,
inner sep=1pt,
fit=#1}}
\usetikzlibrary{calc,decorations,calligraphy,positioning}

\usepackage{booktabs}
\usepackage{makecell} 
\usepackage{multirow}
\usepackage{tabularx}
\newcolumntype{C}{>{$\displaystyle}c<{$}}
\definecolor{tbcolor}{HTML}{E1EFF0} 
\definecolor{blizzardblue}{rgb}{0.67, 0.9, 0.93}
\usepackage[math]{cellspace}
\usepackage{nicematrix}
\NiceMatrixOptions{cell-space-limits = 3pt,exterior-arraycolsep}

\usepackage{xspace}

\usepackage[normalem]{ulem} 

\NiceMatrixOptions{code-for-first-row = {\scriptstyle},
                    code-for-first-col = {\scriptstyle},
                    code-for-last-row =  {\scriptstyle},
                    code-for-last-col =  {\scriptstyle}}

\definecolor{airforceblue}{rgb}{0.36, 0.54, 0.66}
\definecolor{brightmaroon}{rgb}{0.76, 0.13, 0.28}
\definecolor{brickred}{rgb}{0.8, 0.25, 0.33}
\definecolor{cadetblue}{rgb}{0.37, 0.62, 0.63}
\definecolor{cambridgeblue}{rgb}{0.64, 0.76, 0.68}
\definecolor{caribbeangreen}{rgb}{0.0, 0.8, 0.6}
\definecolor{munsellblue}{rgb}{0.0, 0.5, 0.69}
\definecolor{bondiblue}{rgb}{0.0, 0.58, 0.71}

\colorlet{labelcolor}{RedViolet}
\colorlet{shadingcolor}{NavyBlue!12}
\colorlet{shadingcolor2}{PineGreen!10}

\usepackage{enumitem}

\newcommand{\Ct}{\widetilde{C}}
\newcommand{\Ut}{\widetilde{U}}
\newcommand{\Vt}{\widetilde{V}}
\newcommand{\Qt}{\widetilde{Q}}
\newcommand{\Rt}{\widetilde{R}}
\newcommand{\zb}{{\scriptsize{\textsf{Z}}}}
\newcommand{\tb}{{\scriptsize{\textsf{T}}}}
\newcommand{\chib}{{\boldmath{\chi}}}

\title{\boldmath Bulk Reconstruction from Generalized Free Fields}
\author[a,b]{Tamra Nebabu}
\author[a,c]{Xiao-Liang Qi}
\affiliation[a]{Stanford Institute for Theoretical Physics, Stanford University, Stanford, California 94305, USA}
\affiliation[b]{Department of Applied Physics, Stanford University, Stanford, California 94305, USA}
\affiliation[c]{Department of Physics, Stanford University, Stanford, California 94305, USA}

\abstract{We propose a generalized protocol for constructing a dual free bulk theory from any boundary model of generalized free fields (GFFs). To construct the bulk operators, we employ a linear ansatz similar to the Hamilton-Kabat-Liftschytz-Lowe (HKLL) construction. However, unlike the HKLL construction, our protocol relies only on boundary data with no presupposed form for the bulk equations of motion, so our reconstructed bulk is fully emergent. For a (1+1)d bulk, imposing the bulk operator algebra as well as a causal structure is sufficient to determine the bulk operators and dynamics uniquely up to an unimportant local basis choice. We study the bulk construction for several two-sided SYK models with and without coupling between the two sides, and find good agreement with known results in the low-temperature conformal limit. In particular, we find bulk features consistent with the presence of a black hole horizon for the TFD state, and characterize the infalling fermion modes. We are also able to extract bulk quantities such as the curvature and bulk state correlators in terms of boundary quantities. In the presence of coupling between the two SYK models, we are able to observe evidence of the shockwave geometry and the traversable wormhole geometry using the two-sided mutual information between the reconstructed bulk operators. Our results show evidence that features of the geometric bulk can survive away from the low temperature conformal limit. Furthermore, the generality of the protocol allows it to be applied to other boundary theories with no canonical holographic bulk.
}
\toccontinuoustrue

\begin{document}

\maketitle

\section{Introduction and Background}
\subsection{Motivation}

Holographic duality~\cite{maldacena1999large,witten1998anti,gubser1998gauge} is a major forefront in modern theoretical physics that makes contact with many different subdisciplines of physics including quantum gravity, quantum information, and condensed matter. This is because holographic duality offers a means to address two outstanding challenges in theoretical physics; namely, 1) constructing a quantum gravity theory that is ultraviolet (UV) complete, and 2) developing new tools to understand strongly-coupled quantum systems. While these questions seem unrelated at face value, they are connected by the observed duality between a low-energy, semiclassical supergravity theory and a large-$N$ super Yang-Mills gauge theory on a flat background. On the one hand, the duality allows the physics of the strongly-coupled (non-gravitational) gauge theory to be solved more easily in the weakly-coupled gravitational theory variables. On the other hand, taking the duality seriously, one would hope that the quantum gauge theory can provide a non-perturbative UV-completion of supergravity. Holographic duality is most well-understood for gravitational theories with an asymptotically anti-de Sitter (aAdS) geometry, for which the dual theory is a conformal field theory (CFT) in one lower dimension. Accordingly, this duality is also known as the AdS/CFT correspondence. In this case, the isometry group of AdS space is mapped to the conformal symmetry group of the CFT. States in the CFT are mapped to states of normalizable modes in AdS; non-normalizable modes in AdS are mapped to source fields in the CFT~\cite{balasubramanian1999bulk,witten1998anti}

Holographic duality is a manifestation of the more general 
\textit{holographic principle}, which presumes that the information content of a gravitational theory is fully encoded in a theory with one lower dimension~\cite{hooft1993dimensional,susskind1995world}. One piece of evidence for the generalized holographic principle is that the entanglement entropy of a spatial region is bounded by its area---a fact first proven in the context of black holes that is expected to hold more generally~\cite{bekenstein1994entropy,fischler1998holography,bousso2002holographic}. This raises the question of how generic holographic models are. Can such a correspondence exist beyond the AdS/CFT paradigm? Said another way, what are the minimal ingredients necessary for a QFT to have a gravitational bulk dual description? Various important features of known holographic theories have been discussed in the literature, including conformal symmetry, the presence of a gap in scaling dimension, the Ryu-Takayanagi formula of entanglement entropy~\cite{ryu2006holographic}, etc. However, it remains to be seen whether these are all essential for a holographic theory, or if it is possible to extend holographic duality to more general theories.

A related question is how to explicitly construct the holographic dictionary, \textit{i.e.} the mapping between bulk operators and boundary operators. A more explicit construction of the dictionary would help us understand how general holographic theories are. Different protocols for bulk operator reconstruction have been put forth and discussed in literature. For instance, in the Hamilton-Kabat-Liftschytz-Lowe (HKLL) approach~\cite{hamilton2006holographic,hamilton2006local,hamilton2007local}, a quasilocal bulk operator in AdS space is reconstructed from spacelike-separated boundary operators, plus higher order corrections in the gravitational coupling constant $G_N$. The HKLL approach is an explicit realization of causal wedge reconstruction, in which a boundary subregion is used to reconstruct bulk operators within its causal wedge. The boundary representation of a reconstructed bulk operator is nonlocal, with support spread over the spacelike-separated boundary region. The original HKLL construction relied on the bulk equation of motion, although more general constructions based on the boundary modular Hamiltonian have been discussed~\cite{das2018modular,roy2018bulk}.  
Meanwhile, alternate approaches have been proposed which perform local reconstruction of bulk operators in a bigger bulk subregion known as the entanglement wedge, and which apply to more general bulk geometries~\cite{czech2012gravity,wall2014maximin,headrick2014causality,dong2016reconstruction,faulkner2017bulk,penington2020entanglement,almheiri2019entropy}. Compared to entanglement wedge reconstruction, the HKLL approach applies to a smaller bulk subregion, but has the advantage of being more explicit and depending only on simple operators in a boundary spacetime region.

In this paper, we develop a new method for bulk reconstruction which generalizes the HKLL approach and applies to more general boundary theories. Crucially, our method does \textit{not} require prior knowledge of the semiclassical bulk geometry and equations of motion; rather, our reconstruction protocol proceeds entirely from boundary data. In addition, our approach does not require the boundary theory to have conformal symmetry or even approximate conformal symmetry, as conventionally required of holographic boundary models. Instead, the key ingredient in our construction is that the boundary theory describes interacting generalized free fields (GFFs), {\it i.e.} bosonic or fermionic fields for which the correlation functions satisfy Wick's theorem. By introducing a boundary time discretization, we are able to provide a general and explicit algorithm to map a boundary GFF to a canonically-free bulk field (boson or fermion). With the only assumption being the causal structure, the bulk theory is uniquely determined by the boundary two point functions up to a local basis choice. In this paper, we will focus our attention on the correspondence between a $(0+1)$d boundary theory and a $(1+1)$d bulk theory, but the ideas should also apply to reconstruction in higher dimensions as well. 

Using the large-$q$ Sachdev-Ye-Kitaev (SYK) model and variations thereof as a test bench, we show how the bulk dynamics and geometric features emerge from the boundary model. We realize our construction explicitly by performing a discretization of the boundary dynamics, and obtain a discretized bulk. For a generic boundary thermal state, our construction predicts the existence of a black hole horizon in the dual description, in the sense of a reflectionless semi-infinite "bath" for bulk particles. We obtain an explicit formula of the horizon modes in term of boundary operators. The infalling modes have universal correlation functions which are independent of the details of the boundary theory, relying only on generalized free fields and the Kubo-Martin-Schwinger (KMS) condition~\cite{kubo1957statistical,martin1959theory}. We also discuss how to numerically study a continuous limit of the bulk fermion dynamics and geometry, and verify that and AdS bulk geometry indeed emerges from our construction. Our approach also applies to the thermofield double state, and enables us to study entanglement between bulk fermions that are dual to two boundary systems. We find that the mutual information between the two bulk regions is consistent with an eternal black hole geometry with the two exterior regions sharing a bifurcation horizon. In addition to the thermofield double state, we also studied more generic SYK models with a bilinear coupling. When the bilinear coupling is turned on only at one time instant, we obtain a shockwave geometry~\cite{shenker2014black} where the correlation between the two coupled systems is suppressed. Our approach allows us to study such changes in the correlation patterns by studying the mutual information between bulk regions. We also studied the ground state of a coupled SYK model which corresponds to an eternal traversable wormhole~\cite{maldacena2018eternal}. Our approach correctly predicts qualitative changes in the dual geometry, such as the existence of an end-of-the-word brane and the absence of a horizon. It should be emphasized that the known duality for the black hole, shockwave and eternal traversable wormhole geometries only apply to low energy limit of the SYK model, while our construction holds for general temperature and energy scales, which enables the bulk dynamics and geometry to be defined for such systems.

The remainder of the paper is organized as follows: In Sec. \ref{sec:background} we provide an overview of holographic duality and HKLL reconstruction. In Sec. \ref{sec: framework} we discuss our construction procedure starting from generalized free fields. We will show how the bulk dynamics (in the form of a quantum circuit) and bulk quantum state are constructed explicitly by an orthogonalization procedure. In Sec. \ref{sec: syk} and \ref{sec: coupled_syk} we present numeric and analytic results for the example of SYK model and coupled SYK model. Finally in Sec. \ref{sec: discussion}, we summarize our conclusions and discuss several open questions. 

\section{Overview of HKLL Reconstruction}\label{sec:background}

In this section, we review key facts in holographic duality and HKLL reconstruction that will be useful in subsequent discussions. 

If holographic duality is an exact correspondence between the bulk and boundary theories, then their Hilbert spaces should be isomorphic,
\begin{align}
    \mathcal{H}_{\rm boundary}\leftrightarrow \mathcal{H}_{\rm bulk}
\end{align}
and there should be a unitary map between their Hilbert spaces such that each operator of the boundary theory is uniquely mapped to an operator in the bulk theory. This mapping would ensure that correlation functions on the two sides match, as required for a duality. Such an exact correspondence would also enable one to define non-perturbative states in quantum gravity using the boundary theory variables.

It is a nontrivial task to define the full bulk quantum mechanical state space, which includes non-perturbative states with no smooth geometry. However, in the perturbative quantum gravity regime when $G_N\to 0$, the bulk is described by a semiclassical geometry with boundary, and the bulk and boundary Hilbert spaces can be approximately identified. In this limit, the bulk theory is described by a QFT on a curved spacetime with asymptotically AdS geometry. However, the correspondence between the bulk and boundary Hilbert spaces is not perfectly straightforward, because the bulk QFT is local while the boundary theory is defined in one lower dimension, so the bulk-boundary mapping must be nonlocal. The equation of motion for the bulk fields generally admits both normalizable and non-normalizable solutions. The normalizable modes correpond to boundary quantum fields while the non-normalizable modes correspond to source terms in the boundary theory\cite{witten1998anti}. The duality between the bulk and boundary Hilbert spaces is defined perturbatively, with an identification between the low-energy subspace of a small number of normalizable bulk matter excitations, and the low-energy subspace of the boundary theory. 

As an example of the bulk-boundary mapping, we review the extrapolate dictionary of AdS/CFT, which will be useful for setting up our construction. Suppose the low-energy bulk theory is described by a scalar field $\phi(z,\bm{x},t)$ on a semiclassical AdS background with metric
\begin{align}
    ds^2=\frac{R^2}{z^2}\left(-dt^2+d{\bf x}^2+dz^2\right)
\end{align}
where $z$ is the coordinate associated to the extra bulk dimension. A free scalar field $\phi$ in the bulk has equation of motion (EOM)
\begin{align}
    \left(\Box-m\right)\phi=0
\end{align}
with $\Box \equiv g^{\mu \nu}\partial_\mu\partial_\nu$. We would like to relate normalizable solutions of the EOM---which are the quanta of the bulk theory---to the quanta of the boundary theory. This may be done by computing bulk correlators and pulling them to the conformal boundary. While the correlators of normalizable bulk modes vanish at the boundary, their leading order behavior can be used to identify bulk operators with their boundary operator representations. For instance, the two point functions in AdS/CFT are identified as
\begin{align}\label{eq: two_pt_duality}
    \lim_{z\rightarrow 0}z^{-2\Delta}\bra{\Omega_{\rm blk}}\phi(z,{\bf x}_1,t_1)\phi(z,{\bf x}_2,t_2)\ket{\Omega_{\rm blk}}=\bra{\Omega_{\rm bdy}}\mathcal{O}({\bf x}_1,t_1)\mathcal{O}({\bf x}_2,t_2)\ket{\Omega_{\rm bdy}}
\end{align}
where $\OC(\bm x,t)$ are the boundary CFT primary operators with conformal dimension $\Delta$, and $\ket{\Omega_{\rm blk}}$ and $\ket{\Omega_{\rm bdy}}$ denote the ground states of the bulk and boundary Hilbert spaces, respectively. The extrapolate dictionary then provides the following mapping between quasilocal operators in the bulk and boundary theories:
\begin{align}
    \label{eq:extrapolate_dict}
    \lim_{z\rightarrow 0}z^{-\Delta}\phi(z,{\bf x},t)=\mathcal{O}({\bf x},t)
\end{align}
The relevant low-energy bulk Hilbert subspace is generated by normalizable excitations with an $\OO\left(N^0\right)$ number of $\phi$ quanta, where $N$ the number of flavors in the large-$N$ dual boundary theory. This is the small subspace of the full bulk Hilbert space for which the correspondence is well-defined. 

Na\"ively, the bulk appears to contain more states than the boundary theory (even when restricted to the low-energy subspace), since all boundary states can be associated to bulk states generated by the quasilocal bulk $\phi$ operators in the neighborhood of the boundary. There are apparently more bulk states which are generated by quasilocal excitations deeper in the bulk interior. However, such quasilocal excitations will generally propagate to the boundary at later times---at least in the semiclassical gravitational picture---and therefore, may be identified with boundary excitations. In other words, local operators in the bulk correspond to some time-evolved operators on the boundary, which are nonlocal operators (in Heisenberg picture) \cite{banks1998ads}.

Accordingly, a natural way to relate bulk operators and boundary operators is by evolving quasilocal bulk operators to the future or the past. To make this relation more precise, we can explicitly write down bulk fields which satisfy Eq.~\eqref{eq:extrapolate_dict} using the approach by HKLL. Normalizable solutions $\phi(z,\bm x, t)$ to the bulk EOM may be obtained by expanding in a basis of normalizable modes, and imposing the extrapolate dictionary on the asymptotic behavior in the extra bulk coordinate. This provides a linear relation between the boundary operator $\OC(\bm x,t)$ and the part of the bulk normalizable modes that only depends on $(\bm x,t)$. If the bulk EOM is linear, there exists a linear kernel $K$ that relates normalizable bulk operators to the boundary operators as
\begin{align}
    \phi(z,{\bf x},t)=\int dt'd{\bf x}'K(z,{\bf x},t|{\bf x}',t')\mathcal{O}({\bf x}',t')
\end{align}
HKLL showed that $K$ can be taken to have support only on the boundary points $(\bm x',t')$ that are spacelike-separated from the bulk point $(z,\bm x,t)$~\cite{hamilton2006holographic,harlow2018tasi}. Roughly speaking, one can view the kernel as a spacelike Green's function that describes propagation along the $z$ direction rather than time. The existence of such a Green's function is related to the time-like tube theorem applied to the bulk QFT~\cite{araki1963generalization,borchers1961vollstandigkeit,strohmaier2023analytic,strohmaier2023timelike}.

When the bulk is two-dimensional, the reconstruction formula simplifies to
\begin{align}
    \phi(z,t)=\int_{t-z}^{t+z}dt'K(z,t|t')\mathcal{O}(t')
\end{align}
When $\phi$ is not a free field, the bulk interactions can be considered perturbatively. For example a scalar field with a $\phi^4$ interaction satisfies the modified equation of motion $\left(\Box-m\right)\phi=-\lambda \phi^3$, which can be used to derive a nonlinear kernel relation between $\phi$ and $\mathcal{O}$ \cite{harlow2018tasi}. 

In the HKLL approach, the bulk equation of motion of a (nearly) free field plays an essential role in the reconstruction. As a consequence, this approach has only been applied to few cases such as the AdS vacuum where the kernel can be obtained analytically. In the next section we will describe a new bulk reconstruction procedure which is similar to HKLL but proceeds purely from boundary data, without assuming the bulk geometry or equations of motion. In particular, we will describe a procedure for finding the kernel that relates bulk and boundary operators by imposing a bulk causal structure. Our approach may be regarded as a realization of causal wedge reconstruction where we assume a causal structure to constrain the possible kernel solutions. We will show that for a (0+1)d boundary theory of GFFs, this constraint is sufficient to uniquely determine an emergent (1+1)d bulk description, up to a local basis choice.

\section{Bulk Reconstruction from Boundary Generalized Free Fields}
\label{sec: framework}
In this section, we describe our bulk reconstruction approach, which relies purely on boundary data and is applicable to any boundary model with generalized free fields (GFFs). A GFF is any 
(generically interacting) bosonic or fermionic field whose correlation functions satisfy Wick's theorem. Sec.~\ref{subsec: GFF} reviews the definition of GFFs.  Sec.~\ref{subsec: kernel} develops the procedure that determines the bulk-boundary mapping kernel using the boundary two point functions. From the kernels, one can determine the bulk dynamics and bulk state. For simplicity, we shall restrict our attention to a (0+1)d boundary model with an emergent (1+1)d bulk description. Our method can accommodate higher-dimensional boundary models, though bulk locality in all directions is not guaranteed (see discussion in Sec.~\ref{sec: discussion}).


\subsection{Generalized free fields}\label{subsec: GFF}
First, we review the definition of a GFF. GFFs may possess either bosonic or fermionic statistics, but for concreteness and anticipating the future discussion of the SYK model, we shall focus on fermionic GFFs. Consider a set of Majorana fermion fields $\chi_i(t)$ with $~i=1,\dots,N$ in $(0+1)$d. The $\chi_i(t)$ are Heisenberg operators evolved by a given Hamiltonian $H(t)$ (which may be time-dependent), and that satisfy the Heisenberg equation $i\partial_t\chi_i(t)=\left[\chi_i(t),H(t)\right]$. Majorana fields are defined to satisfy the equal-time canonical anticommutation relations
\begin{align}
    \left\{\chi_i(t),\chi_j(t)\right\}=\delta_{ij}
\end{align}
A field $\chi_i(t)$ is a GFF with respect to state $\rho$ if and only if all of its correlation functions computed in the state $\rho$ satisfy Wick's theorem; that is, all $2p$-point correlators are factorizable into a product of $p$ two point correlators, summed over all possible Wick contractions, while odd moments vanish\footnote{For bosonic GFFs, odd moments are nonzero if the field has a nonzero one-point function, but this can be removed by a redefinition of the field.}. For example, the four-point function of the fermionic GFFs $\chi_i(t)$ should satisfy
\begin{equation}
\begin{aligned}
\ev{\chi_i(t_1)\chi_j(t_2)\chi_k(t_3)\chi_l(t_4)}&=\ev{\chi_i(t_1)\chi_j(t_2)}\ev{ \chi_k(t_3)\chi_l(t_4)}\\
    &~-\ev{\chi_i(t_1)\chi_k(t_3)}\ev{ \chi_j(t_2)\chi_l(t_4)}\\
    &~+\ev{\chi_i(t_1)\chi_k(t_4)}\ev{ \chi_j(t_2)\chi_l(t_4)}
\end{aligned}
\end{equation}
where the relative sign between the terms comes from the fermionic nature of the fields. Here, all expectation values are taken with respect to a fixed reference state with $\ev{...}\equiv{\rm tr}\left(\rho...\right)$ and $\rho$ normalized such that $\Tr(\rho) = 1$. Hence, the GFF condition is a requirement on both the state and the Hamiltonian. 

Free fields, which possess a Hamiltonian that is quadratic in the fields, are trivially GFFs provided that $\rho$ is a Gaussian state. For example, a free theory of Majorana fields $\chi_i(t)$ has a quadratic Hamiltonian of the general form
$H(t)=\frac12\sum_{i,j}h_{ij}(t)\chi_i\chi_j$, where $h$ is an antisymmetric Hermitian matrix (i.e. Hermitian and purely imaginary). Note that $H$ is a time-dependent Schr\"odinger picture Hamiltonian. The time evolution of a collection of free fields is linear:
\begin{align}
    i\partial_t\chi_i(t)&= h_{ij}(t)\chi_j(t),\nonumber\\\Rightarrow \chi_i(t)&=U_{ij}(t,t')\chi_j\left(t'\right),~U_{ij}(t,t')\equiv \left(T\exp\left[-i\int_{t'}^t d\tau h(\tau)\right]\right)_{ij}\label{eq: free evolution}
\end{align}
$U_{ij}(t,t')$ is a unitary transformation in the linear space spanned by fermion fields $\chi_i(t)$, and should not be confused with the fermion Hilbert space time evolution operator $\widetilde{U}(t)$, which implements Heisenberg evolution as $\chi_j(t) = \widetilde{U}^\dag(t)\chi_j(0)\widetilde{U}(t)$. For the fields $\chi_i(t)$ to be GFFs, the state $\rho$ needs to be a Gaussian state of the form $\rho\equiv Z^{-1}e^{-\sum_{i,j}\tilde{h}_{ij}\chi_i\chi_j}$ where $\tilde{h}_{ij}$ encodes the quadratic parent Hamiltonian, which is also called the modular Hamiltonian. Note that the modular Hamiltonian is generically different from the Hamiltonian governing Heisenberg evolution.

Interestingly, in the large-$N$ limit where the Hilbert space dimension approaches infinity, it is possible to have \textit{interacting} theories of Majorana fermions in which $\chi_i(t)$ is a GFF for all correlation functions that involve a finite number of field insertions~\cite{greenberg1961generalized,dutsch2003generalized}. An example of such a theory is the Sachdev-Ye-Kitaev (SYK) model~\cite{sachdev1993gapless,kitaev2014hidden,kitaev2015simple}, which we will detail in the following section. For now, we will describe our reconstruction procedure in a general context, without specializing to any particular boundary model. 

The difference between a canonically free field and an interacting GFF plays an essential role in our reconstruction. To illustrate this difference, consider the Wightman functions
\begin{align}
    G_{jk}(t,t')&=\ev{ \chi_j(t)\chi_k\left(t'\right)}=\frac12\left(A_{jk}(t,t')+C_{jk}(t,t')\right)\label{eq: boundary2pt}\\
    A_{jk}(t,t')&=\ev{\left\{\chi_j(t),\chi_k\left(t'\right)\right\}}\label{eq: boundary spectral}\\
    C_{jk}(t,t')&=\ev{\left[\chi_j(t),\chi_k\left(t'\right)\right]}\label{eq: boundary commutator}
\end{align}
For any GFF, we have the following fact as a consequence of Wick's theorem:
\begin{fact}\label{fact: Ajk}
\sloppy For any bosonic multi-point operators $\mathcal{O}_1=\chi_{i_1}(t_1)\chi_{i_2}(t_2)...\chi_{i_{2p}}(t_{2p})$, $\mathcal{O}_2=\chi_{l_1}(t_1)\chi_{l_2}(t_2)...\chi_{l_{2p}}(t_{2q})$, the correlation function satisfies
\begin{align}
    \ev{ \mathcal{O}_1\left\{\chi_j(t),\chi_k\left(t'\right)\right\}\mathcal{O}_2}=A_{jk}(t,t')\ev{ \mathcal{O}_1\mathcal{O}_2}
    \label{eq: factorization}
\end{align}
\end{fact}
This factorization can be directly verified by considering all Wick contractions. The anticommutator is a sum of two terms, $\ev{\OC_1 
\chi_j(t)\chi_k(t')\OC_2}$ and $\ev{\OC_1 
\chi_k(t')\chi_j(t)\OC_2}$. Every Wick contraction of the first term where $\chi_j(t')$ and $\chi_k(t')$ are paired with different $\chi$'s in $\OC_1$ and $\OC_2$ is canceled by the equivalent index contraction coming from the second term where $\chi_j(t)$ and $\chi_k(t')$ start off in reverse order. Thus, the only surviving contractions are when $\chi_j(t)$ is paired with $\chi_k(t')$, yielding Eq.~\eqref{eq: factorization}. The above equation implies that we can write the operator identity
\begin{align}
    \left\{\chi_j(t),\chi_k\left(t'\right)\right\}=A_{jk}(t,t')\Id
\end{align}
for the subspace of states with finite number of fermion excitations on top of the reference state $\rho$. More precisely, we can define a linear space of states spanned by $\mathcal{O}_1\rho\mathcal{O}_2$ with each $\OC_i$ taking on the form $\mathcal{O}_i=\sum_{p=0}^k\sum_{i_1,\dots,i_{2p}}\psi_{i_1,i_2,...,i_{2p}}\chi_{i_1}(t_1)\chi_{i_2}(t_2)...\chi_{i_{2p}}(t_{2p})$. As long as the maximal length $2k$ satisfies $2k/N\rightarrow 0$ in the large $N$ limit, the operator identity holds.

The key difference between free and interacting GFFs is in the structure of the spectral function $A_{jk}(t,t')$. For a free fermion theory, one can show using Eq.~\eqref{eq: free evolution} that $A_{jk}(t,t')=U_{jk}(t,t')$ is unitary in the flavor indices $j,k$\footnote{This essentially follows from the fact that the time evolution of free fields is linear and that the coupling matrix $h$ is Hermitian.}. In contrast, for interacting GFFs, $A_{jk}(t,t')$ is generically non-unitary. For instance, for the large-$N$ SYK model in the zero temperature limit, we have $A_{jk}(t,t')\propto\delta_{jk}\left(t-t'\right)^{-2/q}$. Physically, the unitarity of the spectral function of a free fermion theory reflects the fact that single-particle states at different times only differ by a single-particle basis transformation. In contrast, in an interacting theory, a single particle can decay into multiple particles, and $A_{jk}(t,t')$ encodes the probability that a single-particle state remains a single-particle state at later time, which generically decays with time. The ability for single particles to decay is crucial to our reconstruction; as we will soon show, this is what allows the emergent bulk direction to appear.

The key observation that leads to our reconstruction is that all multipoint functions of a given GFF may be reproduced by those of a free theory with one extra spatial dimension, which we call the bulk theory. Correspondingly, the original theory with GFFs is regarded as the boundary theory. Because of Wick's theorem, all bulk correlation functions are encoded by the boundary Wightman function $G_{jk}(t,t')$. We will further show that the bulk theory constructed from a boundary theory of fermionic GFFs is determined by the spectral function $A_{jk}(t,t')$, and that the bulk fields are unique up to some unimportant local basis choice. An analogous bulk construction holds for bosonic GFFs, where $A_{jk}(t,t')$ in Eq.~\eqref{eq: factorization} is replaced by the commutator matrix $C_{jk}(t,t')$. 

\subsection{Bulk-boundary mapping from boundary two point functions}\label{subsec: kernel}

Here, we explicitly delineate our bulk reconstruction protocol, using Majorana fermion operators as our archetypal example. We consider a (0+1)d boundary theory with $N$ Majorana fields that are GFFs. We begin by discretizing the boundary time $t$ by sampling $t = \tb\Delta t$, $\tb \in \Z$. Note that boundary evolution itself still occurs in continuous time, but we are simply taking snapshots at discrete time points to simplify the mathematical problem we are dealing with.
For these discrete time points, the two point function $G_{jk}(t,t')$ becomes a matrix $G$ of dimension $N\cdot n_t \times N\cdot n_t$ where $n_t$ is the number of discrete time points being considered\footnote{Note that for finite $n_t$, the discretized two point functions cannot be truly time-translation invariant. This has implications on the discrete bulk theory which is constructed from the sampled boundary data. However, in certain cases, the discrete bulk theory approaches a continuum theory whose behavior in certain bulk regions reflects the underlying time-translation invariance of the original boundary data. See Sec.~\ref{subsec: bulk entropy} for an example in the case study of the SYK model.}. Correspondingly, the spectral function containing the anticommutator data becomes a matrix $A$ with matrix elements $A_{jk}(t,t')$. Note that $A$ is real and positive semi-definite\footnote{To show that $A$ is positive semi-definite: Let $M \equiv \sum_j v_j\chi_j$ be a linear combination of the $N$ Majorana operators. Then $v^\dag A v = (v^\dag)_j \Tr(\rho \{\chi_j,\chi_k\}) v_k = \Tr\left(\rho \left(v_j^*\chi_j\chi_k v_k + v_j^*\chi_k\chi_j v_k \right)\right)= \Tr\left(\rho M^\dag M\right) + \Tr\left(MM^\dag)\right)$. Both terms are non-negative because $M^\dagger M, MM^\dagger$ and $\rho$ are all positive semi-definite.
}.

The key result of this paper is that the matrix $A$ determines a bulk dual theory which is described by a Gaussian quantum circuit as illustrated in Fig. \ref{fig: spacelike_sep}, and the bulk-boundary correspondence is given by an HKLL-like kernel which is also completely determined by $A$. We first state the result and then provide a constructive proof. 

\begin{result}\label{result: main}
    For Majorana GFFs $\chi_i(t)$ with a given time discretization $\Delta t$, there exists a dual description given by a discrete bulk theory of canonically free Majorana fields $\psi_{ia}(z,t)$ with $a = L,R$ and $(z,t)=(\zb\Delta t,\tb\Delta t)$ and $\zb \pm \tb\in\mathbb{Z},~\zb \geq 0$, for which the dynamics is described by a local unitary Gaussian quantum circuit with reflective boundary conditions at $z=0$, as illustrated in Fig. \ref{fig: spacelike_sep}. $\psi_{ia}(z,t)$ is related to the boundary fermions in the region $[t-z,t+z]$ by a linear transformation
    \begin{align}
        \psi_{ja}(z,t)=\sum_{t'=t-z}^{t+z}K_{ja,k}(z,t|t')\chi_k(t'),~a=L,R
        \label{eq: kernel def}
    \end{align}
    The bulk theory satisfies the following conditions:
    \begin{enumerate}
        \item Bulk fermions that are spacelike-separated satisfy the canonical anticommutation relations. (A more precise definition of "spacelike-separated" for quantum circuits will be presented below.) 
        \item The kernel $K$ is uniquely determined by the boundary spectral function matrix $A$, up to a local basis choice $K_{ja,k}(z,t|t')\rightarrow \sum_m O_{jm}(z,t)K_{ma,k}(z,t|t')$ with $O(z,t)$ an orthogonal matrix that can be chosen independently at each bulk spacetime point.
        \item The unitary gates $U(z,t)$, $V(z,t)$ in the bulk (see Fig. \ref{fig: spacelike_sep}) are also uniquely determined by $A$, up to the same freedom of local orthogonal transformations. 
        \item The unitary circuit has a reflective boundary condition at $z=0$ $\psi_{iL}(z=0,t)=\psi_{iR}(z=0,t)=\chi_i(t)$. This is the discrete analog of the extrapolate dictionary. The circuit is either semi-infinite or has another reflective boundary condition at some finite $z=z_{\rm max}(t)$. This secondary reflective boundary condition is also determined by $A$.
    \end{enumerate}
\end{result}

Let us now explain the derivation of Result \ref{result: main}. We will first define the Gaussian unitary circuit corresponding to the bulk theory and discuss its causal structure. Then we will explain how the kernel $K$ and the gates are determined by an orthogonalization procedure. In general, a local quantum circuit can be viewed as implementing a discretized time evolution using unitary gates where each gate only couples a finite number of nearby qubits. The circuit we consider for our bulk description consists of gates which couple qubits on neighboring sites, as shown in Fig. \ref{fig: spacelike_sep}. The qubits and corresponding bulk fermion operators are associated to the circuit lines/links while the circuit gates are denoted by the colored vertices. We label each vertex by coordinates $(z,t)$ where $z,t$ are discretized by half-integer multiples of $\Delta t$, i.e. $z\equiv \zb\Delta t$ and $t\equiv \tb\Delta t$ with $\zb,\tb$ both integer or half-integer. The gates for half-integer coordinates are denoted by $\Ut$ while those for integer coordinates are denoted by $\Vt$. For each gate at vertex $(z,t)$, we denote the fermion operators at the left-moving input leg as $\psi_{iL}(z,t)$, and denote those at the right-moving output leg as $\psi_{iR}(z,t)$. For convenience, we will occasionally label the bulk gates and fermions by their discrete labels $(\zb,\tb)$ instead of $(z,t)$, so $\psi_{iL}(\zb,\tb)$ and $\Ut(\zb,\tb)$ has the same meaning as $\psi_{iL}(z,t)$ and $\Ut(z,t)$, respectively. The Hilbert space of the circuit at a given time $\tb$ is a direct product of the single-site Hilbert spaces of all links intermediated by a horizontal cut. Each quantum gate acts on the $2N$ Majorana fermions on two neighboring sites. Generically, the gates could be any unitary operator in the $2^N$-dimensional Hilbert spaces of the two sites, but we will focus on Gaussian quantum circuits in which all gates are Gaussian unitaries, which implement the discrete analog of Eq.~\ref{eq: free evolution}. A Gaussian unitary gate $\Ut(\zb,\tb)$ satisfies\footnote{Note that following our convention (see Fig. \ref{fig: spacelike_sep}b), $\psi_{iR}(\zb-\frac12,\tb-\frac12)$ and $\psi_{iL}(\zb,\tb)$ are the $2N$ input fermions to $U(\zb,\tb)$, and the two operators on the right-hand side of the equation are the outputs.}
\begin{align}
    \Ut^\dagger(\zb,\tb)\psi_{iR}\left(\zb-\frac12,\tb-\frac12\right)\Ut(\zb,\tb)&=\alpha_{ij}\psi_{jR}\left(\zb-\frac12,\tb-\frac12\right)+\beta_{ij}\psi_{jL}\left(\zb,\tb\right)\nonumber\\
    \Ut^\dagger(\zb,\tb)\psi_{iL}(\zb,\tb)\Ut(\zb,\tb)&=\gamma_{ij}\psi_{jR}\left(\zb-\frac12,\tb-\frac12\right)+\zeta_{ij}\psi_{jL}\left(\zb,\tb\right)\label{eq: U components}
\end{align}
\begin{figure}[t]
    \centering
    \includegraphics[width=\linewidth]{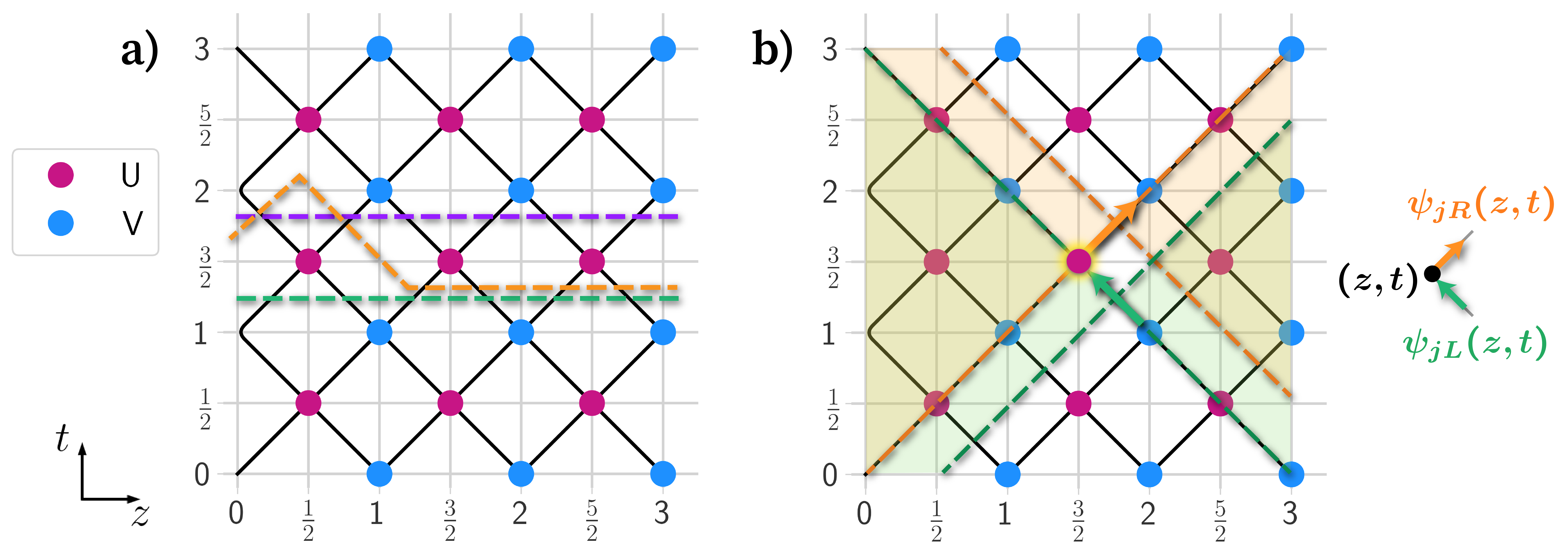}
    \caption{a) Imposing anticommutation relations for the bulk fermions along the horizontal cuts, such as those given by the green and purple dashed lines, implies that the fermions residing along the orange dashed line also obey anticommutation relations. This essentially follows from the brickwork locality of the circuit. b) The spacelike-separated wedges for the right (left)-moving fermions indicated in orange (green) are demarcated by dashed lines. Bulk fermions associated to the vertices fully contained in the shaded orange (green) region anticommute with the right (left)-moving fermion at the highlighted vertex. To each vertex coordinate $(z,t)$, one can associate a pair of outgoing and incoming bulk fermion modes, as shown in the legend on the right. }
    \label{fig: spacelike_sep}
\end{figure}
Grouping the linear coefficients on the right hand side into a matrix $U(\zb,\tb)$, we have
\begin{equation}
\begin{aligned}\label{eq: Uzt}
    \left(\begin{array}{c}\psi_{L}\left(\zb-\frac12,\tb+\frac12\right)\\\psi_{R}(\zb,\tb)\end{array}\right)&\equiv \Ut^\dagger(\zb,\tb)\left(\begin{array}{c}\psi_{R}\left(\zb-\frac12,\tb-\frac12\right)\\ \psi_{L}(\zb,\tb)\end{array}\right)\Ut(\zb,\tb)\\
    &=U(\zb,\tb)\left(\begin{array}{c}\psi_{R}\left(\zb-\frac12,\tb-\frac12\right)\\ \psi_{L}(\zb,\tb)\end{array}\right)
\end{aligned}
\end{equation}
with $U(\zb,\tb)$ an orthogonal $2N\times 2N$ matrix. Similarly, each gate $\Vt(\zb,\tb)$ corresponds to an orthogonal single-particle transformation matrix $V(\zb,\tb)$. 

Although the quantum circuit is not endowed with a metric\footnote{In fact, we shall later show in some example cases how metric emerges from the bulk circuit dynamics in the continuum limit $\Delta t\rightarrow 0$.}, it has a well-defined causal structure, defined by a discrete analog of light cones. 
Because unitary gates preserve the canonical anticommutation relations between fermions, if fermions along a horizontal cut ($\psi_{jR}(\zb,\tb)$ and $\psi_{jL}(\zb,\tb+1/2)$ for fixed $\tb$) are anticommuting with each other, then they remain anticommuting after application of one or more local unitary gates. This is what defines whether a particular cut in the circuit is spacelike. For all spacelike-separated pairs of points $(a,\zb,\tb)$ and $(b,\zb',\tb')$
\begin{align}\label{eq: bulk anticommutation}
    \left\{\psi_{ia}(\zb,\tb),\psi_{jb}(\zb',\tb')\right\}=\delta_{ij}\delta_{ab}\delta_{\zb\zb'}\delta_{\tb\tb'}
\end{align}
It is straightforward to see that all pairs of points with $|\zb-\zb'|>|\tb-\tb'|$ are spacelike-separated. For points along diagonal cuts with $|\zb-\zb'|=|\tb-\tb'|$, we need to separately consider the left and right-movers. As is shown in Fig. \ref{fig: spacelike_sep}(b), $\psi_{iR}(\zb,\tb)$ (orange arrow) anticommutes with fermions $\psi_{iR}(\zb-\frac{n}{2},\tb+\frac{n}{2})$, $n \in \Z$ residing on the links \textit{crossing} the downwards-sloping orange dashed line. By contrast, right-movers along the other orange dashed line $\psi_{iR}(\zb+\frac{n}{2},\tb+\frac{n}{2})$ are not spacelike-separated. However, the left-movers along the parallel cut indicated by the green dashed line anticommute with each other.

\begin{figure}[ht]
    \centering
    \includegraphics[width=.8\linewidth]{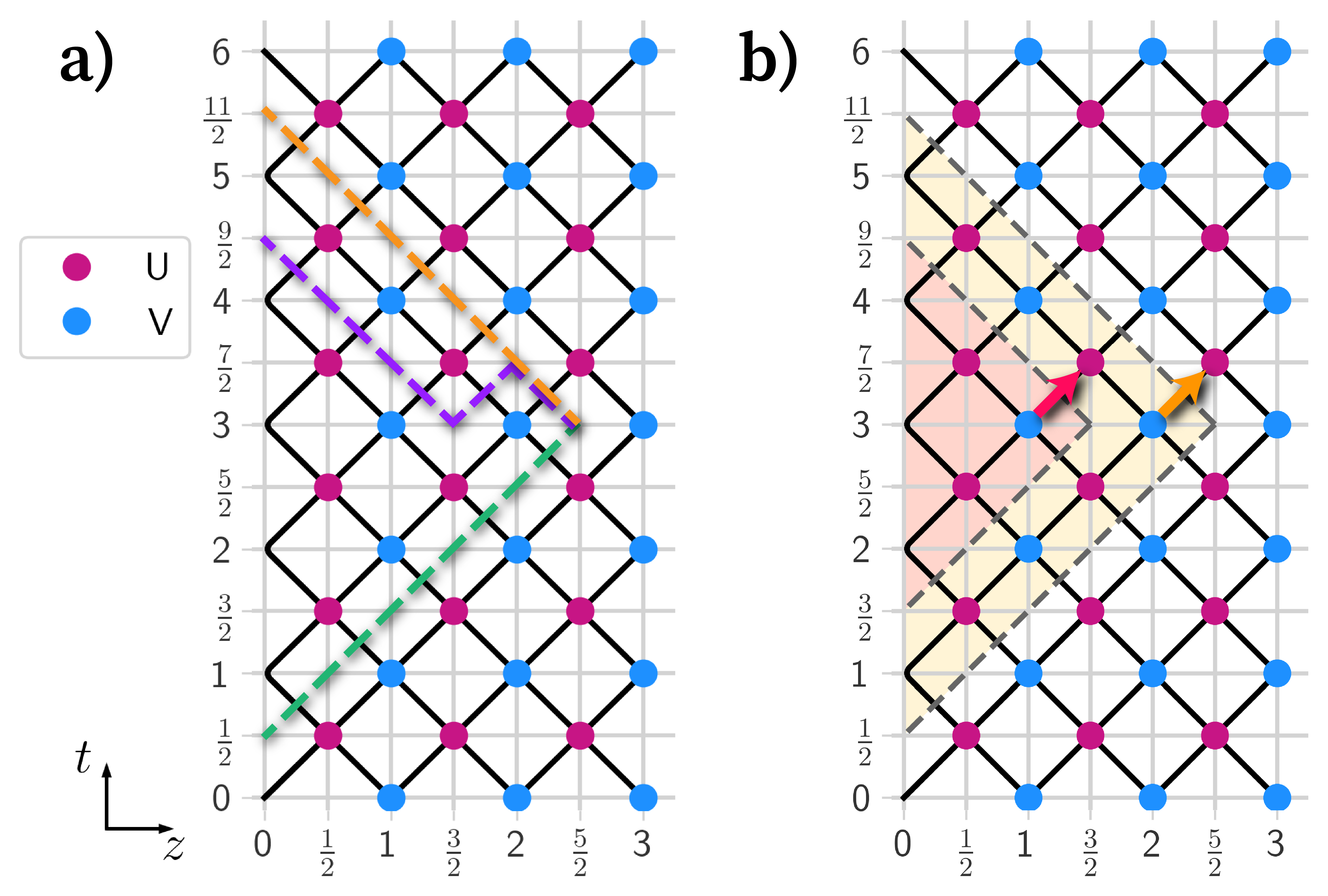}
    \caption{Quantum circuit representation of the discretized (1+1)d bulk. The two types of gates $U$ (pink) and $V$ (blue) live on the integer and half-integer vertices, respectively.  a) Various cuts indicated by the orange, purple, and green lines represent choices for the spacelike-separated bulk operators (residing on the black lines) to use when finding the kernels. b) Choosing diagonal cuts allows one to find right-moving and left-moving bulk fermion modes.}
    \label{fig: circuit_cuts}
\end{figure}

With this preparation, we may now determine the kernel map from the boundary to bulk fermion operators defined in Eq.\eqref{eq: kernel def}. Before deriving the kernel, we first present a heuristic counting argument why such a bulk-boundary correspondence makes sense in this discrete setup. In Fig. \ref{fig: circuit_cuts}(a), the time interval $t\in[1,5]$ corresponds to a causal wedge, which is the triangular region enclosed by the green and orange dashed lines. One observes that each spatial cut in this causal wedge, such as the orange, green or purple dashed lines, intersects $5$ links. For each boundary time interval, the number of bulk fermion sites in its causal wedge region is the same as that of boundary time points.  Following the convention in relativistic field theory, we call each spatial cut a Cauchy surface. Each Cauchy surface $\Sigma$ corresponds to a square matrix $K_\Sigma$ such that 
\begin{align}\label{eq: bulk_ansatz}
    \psi_\Sigma=K_\Sigma\chi
\end{align}
Here $\psi_\Sigma$ denotes the vector of $\psi_{ia}(z,t)$ for all links intersecting $\Sigma$, and $\chi$ denotes the vector of $\chi_i(t)$ for $t$ in the corresponding time interval. Since any two points on a Cauchy surface are spacelike-separated, $\psi_\Sigma$ satisfies the canonical anticommutation relation~\eqref{eq: bulk anticommutation}, which implies 
\begin{align}
    \left\{\psi_\Sigma,\psi_\Sigma^\T\right\}&=K_\Sigma \left\{\chi,\chi^\T\right\}K^\T_\Sigma=\Id
\end{align}
In other words, $K_\Sigma$ must satisfy the matrix equation
\begin{align}
    K_\Sigma AK_\Sigma^\T=\Id\label{eq: bulk_alg_cond}
\end{align}
Eq.~\eqref{eq: bulk_alg_cond} by itself does not determine $K_\Sigma$ since any orthogonally-transformed kernel $\widetilde{K}_\Sigma \equiv O K_\Sigma$ also satisfies the same equation. This ambiguity is mostly removed by the requirement (\ref{eq: kernel def}) on the support of the kernel, which restricts the distribution of nonzero matrix elements in $K_\Sigma$. The only remaining ambiguity is the local basis transformation mentioned in Result \ref{result: main}. 

To demonstrate the procedure of determining the kernel more explicitly, we choose to work with the diagonal cuts (orange and green dashed lines in Fig. \ref{fig: circuit_cuts}(a)). We shall refer to these cuts as the past and future "light cones" even though the cuts are one-dimensional and the bulk fermions crossing the cuts are spacelike-separated. We choose the convention of placing the reflective boundary of the causal wedge to the left. For a given boundary time interval $t\in[t_1,t_2]$ with $t_{1,2}=\tb_{1,2}\Delta t$, $\tb_2-\tb_1\in\Z$, the future light cone $\Sigma_f$ corresponds to the right-moving bulk fermions associated to vertices $(z,t)=(\zb \Delta t,\tb \Delta t)$ with $\zb=0,\frac12,1,...,\frac{\tb_2-\tb_1}2$ and $\tb=\tb_2-\zb$. Similarly, the past light cone $\Sigma_p$ corresponds to the left-moving bulk fermions on $(z,t_1+z)$ with $z=\zb\Delta t$ for the same $\zb$. First, consider the right-moving fermions on the future light cone. The support of the kernel~\eqref{eq: kernel def} requires that $\psi_{iR}(z,t_2-z)$ is a linear transformation of the boundary fermions in the time interval $[t_2-2z,t_2]$. For example, $\psi_{iR}(0,5)=\chi_i(5)$ are the boundary fermions at time $5$, while $\psi_{iR}\left(\frac12,\frac 92\right)=\sum_{n =4,5}K_{ij}\left(\frac 12,\frac92|n \right)\chi_j(n)$ is a linear superposition of $\chi_j(4)$ and $\chi_j(5)$, which is required to anticommute with all $\chi_j(5)$, and anticommute among themselves. Mathematically, the condition $KAK^\T=\Id$ can be interpreted as an orthonormality condition for the row vectors of $K$ with respect to the metric $A$. From this point of view, the boundary fermion operators may be regarded as non-orthogonal vectors with respect to the inner product defined by $\left(\chi_i(\tb),\chi_j(\tb')\right)\equiv A_{ij}(\tb,\tb')$ while the bulk fermions along the future light cone correspond to orthonormal vectors $\left(\psi_{iR}(\zb,\tb_2-\zb),\psi_{jR}(\zb',\tb_2-\zb')\right) = \delta_{ij}\delta_{\zb \zb'}$.
The bulk operators along $\Sigma_f$ may be obtained by an iterative Gram-Schmidt orthogonalization procedure for the boundary operators in $[\tb_1,\tb_2]$ starting from $\zb=0$ and moving along $\Sigma_f$. The kernel is therefore the basis transformation between the non-orthogonal boundary operators and the orthogonal bulk operators.
A schematic of the orthogonalization procedure is illustrated in Fig. \ref{fig: op_ortho}. 
An analogous procedure may be applied for the left-moving bulk fermions along $\Sigma_p$, whose kernels are supported on the same boundary time interval. Thus, the bulk fermions on the past light cone correspond to another orthogonal basis of the same linear space, with reverse ordering of the iterative Gram-Schmidt process.\footnote{It is interesting to note the similarity between our orthogonalization procedure and the discussion of Krylov space~\cite{parker2019universal}. More quantitative comparison between these two approaches is reserved for future work.} We summarize the key observations regarding the kernels for $\Sigma_{f,p}$ in the following:
\begin{result}\label{result: orthogonalization}
    Let $\mathbb{V}_{\tb}={\rm span}\left\{\chi_i(\tb),i=1,2,...,N\right\}$ be the linear space of single fermion operators at a given time, and $\mathbb{V}_{\left[\tb_1,\tb_2\right]}=\cup_{\tb\in\left[\tb_1,\tb_2\right]}\mathbb{V}_\tb$ be the space spanned by all fermion operators in the time interval
    $\left[\tb_1,\tb_2\right]$. We have the following results for this linear space:
    \begin{enumerate}
        \item The anticommutator matrix defines an inner product in this linear space $\left(\chi_i(\tb),\chi_j(\tb')\right)\equiv A_{ij}\left(\tb,\tb'\right)$. \item The mapping $\left[\tb_1,\tb_2\right]\longrightarrow \mathbb{V}_{\left[\tb_1,\tb_2\right]}$ preserves inclusion. In equation form, \begin{align}\left[\tb_1',\tb_2'\right]\subseteq\left[\tb_1,\tb_2\right]\implies\mathbb{V}_{\left[\tb_1',\tb_2'\right]}\subseteq \mathbb{V}_{\left[\tb_1,\tb_2\right]}\label{eq: nesting}
        \end{align}
        \item For a bulk coordinate $(\zb,\tb)$, the right-moving fermion operators $\psi_{iR}(\zb,\tb)$ span a linear subspace $\mathbb{W}^{f}_{\zb,\tb}$ which is the orthogonal complement of $\mathbb{V}_{\left[\tb-\zb+1,\tb+\zb\right]}$ in the (generically bigger) space $\mathbb{V}_{\left[\tb-\zb,\tb+\zb\right]}$. In equation, $\mathbb{W}^{f}_{\zb,\tb}\oplus \mathbb{V}_{\left[\tb-\zb+1,\tb+\zb\right]}=\mathbb{V}_{\left[\tb-\zb,\tb+\zb\right]}$. Similarly, the left-moving bulk fermion operators $\psi_{iL}(\zb,\tb)$ span a space $\mathbb{W}^{p}_{\zb,\tb}$ that is the orthogonal complement of $\mathbb{V}_{\left[\tb-\zb,\tb+\zb-1\right]}$ in the space $\mathbb{V}_{\left[\tb-\zb,\tb+\zb\right]}$. 
        \end{enumerate}
\end{result}

\begin{figure}[ht]
    \centering
    \includegraphics[width=0.65\linewidth]{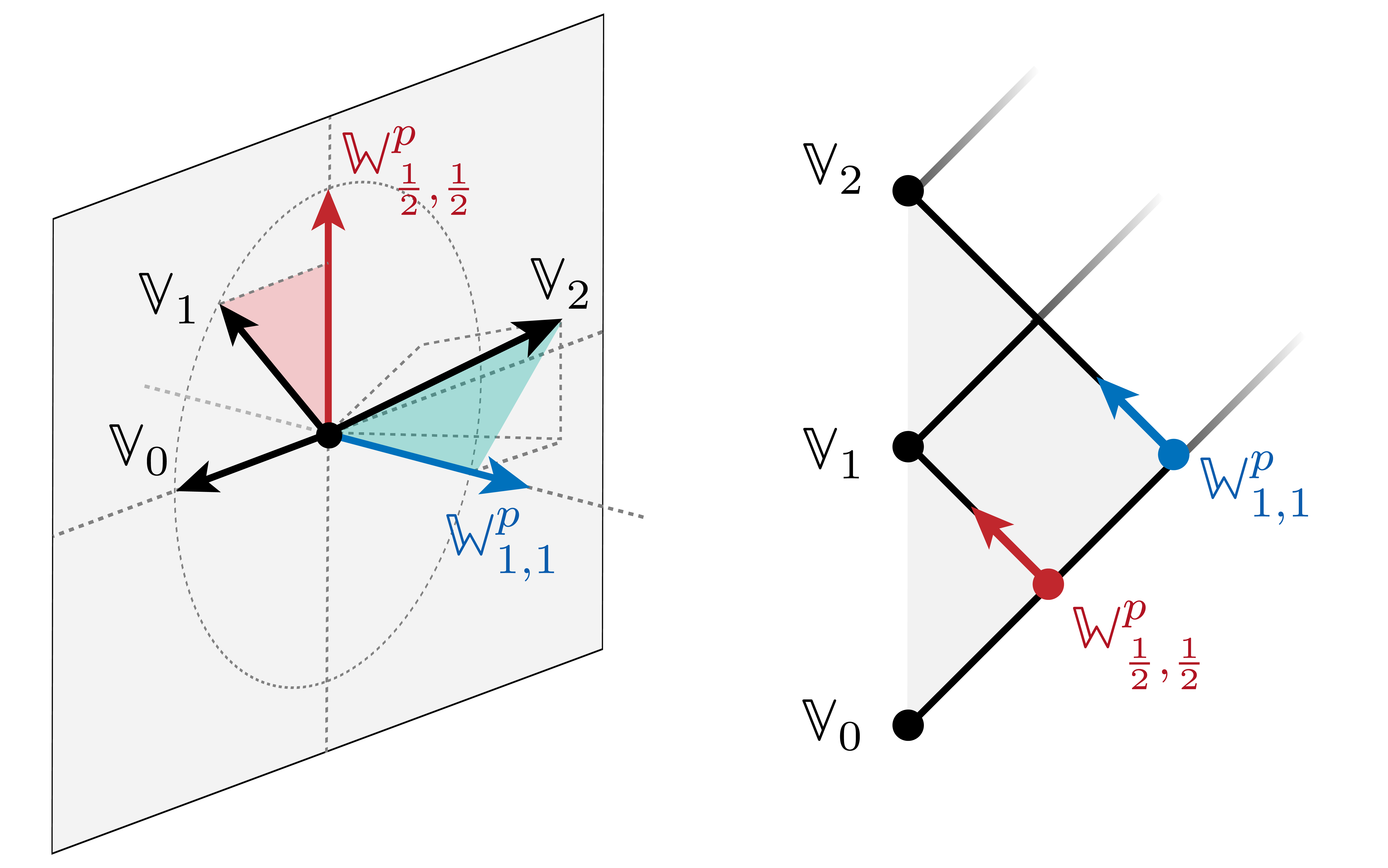}
    \caption{Schematic illustrating the orthogonalization procedure for generating $\mathbb{W}^p_{z,t}$. At $t=0$, we have a linear space $\mathbb{V}_0$ generated by $\{\chi_i(0)\}$. $\{\chi_i(1)\}$ and $\{\chi_i(2)\}$ generate linearly-independent subspaces $\mathbb{V}_1$ and $\mathbb{V}_2$. One can iteratively orthogonalize these linear subspaces to generate $\mathbb{W}^p_{\frac12,\frac32}$ and $\mathbb{W}^p_{1,1}$. The linear spaces are associated to vertices in the bulk circuit as shown on the right.}
    \label{fig: op_ortho}
\end{figure}

We now make some more practical remarks on the computation of kernels $K_{f,p}$ associated to $\Sigma_{f,p}$. The kernel matrices can be subdivided into $N\times N$ blocks with row blocks labeled by the bulk coordinates $(\zb,\tb)$. For the right-movers on the future light cone, the associated kernel matrix $K_f(\zb,\tb_2-\zb|\tb)$ is block upper right-triangular if we arrange $\zb$ in descending order. For example, for the time interval $\tb\in[1,5]$ in Fig. \ref{fig: circuit_cuts}(a), \cref{eq: bulk_ansatz} takes on the form
\begin{align}\label{eq: kernel_matrix}
    \begin{pmatrix}
    \psi_{iR}(2,3)\\
    \psi_{iR}(\frac32,\frac72)\\
    \vdots\\
    \vdots\\
    \psi_{iR}(0,5)
    \end{pmatrix} = 
    \begin{pNiceMatrix}
        k^f_{2,1} & k^f_{2,2} &\Cdots & & k^f_{2,5}\\
        0 & k^f_{\frac 32,2} & k^f_{\frac 32,3}&\Cdots & k^f_{\frac 32,5} \\
        \Vdots & & & & \Vdots \\
        \Vdots & & & & \Vdots\\
        0 & \Cdots & & 0 & \Id\\
    \CodeAfter
        \line{2-1}{5-4}
        \line{2-2}{5-5}
    \end{pNiceMatrix}
    \begin{pmatrix}
    \chi_j(1)\\
    \chi_j(2)\\
    \vdots\\
    \vdots\\
    \chi_j(5)
    \end{pmatrix}
\end{align}
where each $k^f_{\zb,\tb}$ is an $N\times N$ block that acts on the flavor indices. Likewise, the past light cone kernel $K_p(\zb,\tb_1+\zb|\tb)$ with block matrix elements labeled by $(\zb,\tb)$ is a lower left block triangular matrix, if we sort $\zb$ in ascending order. However, as mentioned in Result \ref{result: main}, there is an $O(N)$ ambiguity associated to a local basis transformation at each bulk site, which preserves the block triangular structure of the kernel matrix. Hence, we can choose $K_{f,p}$ to not only be block triangular, but also strictly triangular. This can be viewed as a gauge fixing. Such a choice allows us to obtain the kernels using QR/QL decomposition. In particular, to find $K_f$, we can rewrite $K_fAK_f^\T=\Id$ as $K_fA^{1/2}=Q_f$, with $Q_fQ_f^\T=\Id$, so that $Q_f^\T K_f=A^{-1/2}$. Hence, choosing $K_f$ to be upper triangular implies that it is obtained by a QR decomposition of matrix $A^{-1/2}$ where $Q_f$ is an orthogonal matrix of the decomposition. Likewise, $K_p$ can be obtained by a QL decomposition $A^{-1/2}=Q_p^\T K_p$ with $Q_{p}$ another orthogonal matrix. Since $A$ is real, $K_{f,p}$ may also be chosen to be real. We summarize the essential algorithm behind our numerics in the following pseudocode:
\begin{align}
Q_f^\T,K_f&=\texttt{QR}\left[A^{-1/2}\right]\nonumber\\
Q_p^{\T},K_p&=\texttt{QL}\left[A^{-1/2}\right]\label{eq: QR}
\end{align}
With the above gauge choice, the only remaining freedom is an overall $\pm$ sign in the linear expansion of each bulk fermion operator (App.~\ref{app: qr}). One is free to choose the signs so long as the choice is consistent with fermion number parity conservation. We postpone the discussion of this detail to Sec. \ref{sec: circuit_and_horizon}. 

Once the kernels $K_{f,p}$ are obtained, the bulk dynamics are fully determined. For example, $A^{-1/2}=Q_f^{\T}K_f=Q_p^{\T}K_p$ implies that $K_f=Q_fQ_p^{\T}K_p$, such that $\psi^f=Q_fQ_p^{\T}\psi^p$. Here $\psi^f$ ($\psi^p$) denotes all right-moving (left-moving) fermions on the future (past) light cone. Therefore the matrix $Q_fQ_p^{\T}$ is exactly the single-particle time evolution of $\psi^p$ into $\psi^f$. Likewise, we can use the kernels to determine the single-particle linear transformations $U(\zb,\tb)$ and $V(\zb,\tb)$ (Eq.~\eqref{eq: Uzt}) for each gate in the quantum circuit. For clarity we postpone details of the derivation of the gate transformations to Sec. \ref{sec: circuit_and_horizon}, where we can use the SYK model as a concrete example.

The bulk gates determine the bulk dynamics, but do not specify the quantum state in the bulk. Just as a system with a given fixed Hamiltonian can evolve different quantum states, to fully specify the bulk physics, we need to determine the bulk quantum state. A Gaussian state on a given Cauchy surface $\Sigma$ is fully determined by the two point functions $\left\langle \psi_{ia}(z,t)\psi_{jb}(z',t')\right\rangle$ on $\Sigma$. Since the anticommutator between bulk fermions on $\Sigma$ is canonical, the bulk Gaussian state is fully determined by the bulk commutators on $\Sigma$:
\begin{align}
\left({C_{b\Sigma}}\right)_{ia,jb}\left(z,t;z',t'\right)=\left\langle \left[\psi_{ia}(z,t),\psi_{jb}(z',t')\right]\right\rangle\label{eq: commutator}
\end{align}
The linear relationship between bulk and boundary fermions means the bulk commutators are determined by the boundary commutators and the kernels. In matrix form, we have
\begin{align}
    C_{b\Sigma}=K_\Sigma C K_\Sigma^\T\label{eq: bulk commutator}
\end{align}
where $C$ is the boundary commutator matrix defined in Eq.~\eqref{eq: boundary commutator}, and is imaginary and Hermitian. All properties of the bulk state are encoded in $C_{b\Sigma}$. For example, in Sec.~\ref{subsec: bulk entropy}, we use $C_{b\Sigma}$ to study the entanglement entropy in the bulk dual of the SYK model. 

Results \ref{result: main}, \ref{result: orthogonalization} and Eqs.~\eqref{eq: QR}, \eqref{eq: bulk commutator} summarize the key conclusions for our bulk construction. Our setup explicitly shows that the bulk theory is completely determined by the boundary two point function \eqref{eq: boundary2pt}. 

Before concluding this section, we make a few remarks about our construction. The first point is that the method of QR decomposition \eqref{eq: QR} only applies if $A$ is non-singular. If $A$ is singular, the bulk fermions along the light cones may still be obtained by iterative orthogonalization as summarized in Result \ref{result: orthogonalization}, though the \textit{number} of bulk fermion modes obtained may be smaller than $N$. For example, if $A_{jk}(\tb,\tb')$ is non-singular for time interval $[\tb_1+1,\tb_2]$ but singular for $[\tb_1,\tb_2]$, that means the corresponding linear space $\mathbb{V}_{\left[\tb_1,\tb_2\right]}$ has a dimension smaller than the number of boundary operators $(\tb_2-\tb_1+1)N$. Consequently, the number of orthonormal bulk fermion modes at position $\zb=\frac{\tb_2-\tb_1}{2},~\tb=\frac{\tb_2+\tb_1}2$ is ${\rm dim}\left(\mathbb{V}_{\left[\tb_1,\tb_2\right]}\right)-{\rm dim}\left(\mathbb{V}_{\left[\tb_1+1,\tb_2\right]}\right)$. Our construction continues to apply to increasingly longer boundary intervals, unless ${\rm dim}\left(\mathbb{V}_{\left[\tb_1,\tb_2\right]}\right)={\rm dim}\left(\mathbb{V}_{\left[\tb_1+1,\tb_2\right]}\right)$, in which case there is no bulk fermions at the corresponding bulk vertex. In that case, this point $(\zb,\tb)$ becomes a boundary point. In Sec. \ref{subsec: ETW}, we discuss an example of a coupled SYK model which corresponds to a geometry with a finite boundary at $\zb=\zb_{\rm max}$. 

As a second remark, Result \ref{result: orthogonalization} implies that the only mathematical requirement for the construction of bulk operators is a nested family of linear spaces satisfying Eq.~\eqref{eq: nesting}. This implies our construction of bulk unitary dynamics in the form of a quantum circuit is very general,
which also applies to other problems where the inner product has a different physical interpretation. The case of bosonic GFF can be studied in parallel with the fermion case, but the inner product needs to be replaced by a symplectic form. This case is discussed in Appendix \ref{app: bosonGFF}.

The third point we would like to mention is that it is not strictly necessary to assume a local quantum circuit representation for the operator dynamics \textit{a priori}. Instead, if we begin by assuming that there exists a kernel which defines bulk fermions by Eq. (\ref{eq: kernel def}), and the bulk fermions satisfy cannonical anticommutation relation on a spacelike slice, we can {\it prove} that the discretized bulk dynamics must be given by a local quantum circuit. The choice of coordinates $(z,t)=(\zb\Delta t,\tb\Delta t)$ is simply a discrete version of conformal coordinates. In 2d curved space, one can always find coordinates in which the metric is conformally flat, taking on the form $ds^2=\Omega^2(z,t)\left(dz^2-dt^2\right)$. Therefore the requirement that the kernel for bulk fermion at $(z,t)$ is supported at time interval $[t-z,t+z]$ is a gauge choice that fixes the coordinate choice of the bulk. Further discussion of the bulk metric is presented in Sec.~\ref{sec: bulk_curvature}. The canonical anticommutation relations defined by the orthogonalization procedure in Result \ref{result: orthogonalization} is sufficient to prove that the bulk fermion dynamics is given by a local Gaussian quantum circuit. Details of the proof may be found in App.~\ref{app: bulk_locality}.


\subsection{Notation and Conventions}
For convenience, here we summarize our chosen notation and conventions throughout the paper.
\begin{center}
\begin{NiceTabular}{c *{2}{>{$}c<{$}}}
    \toprule
    \RowStyle{\bfseries} Quantity & \textbf{Boundary} & \textbf{Bulk}\\\midrule
    continuous coordinates & t & (z,t) \\
    discrete coordinates & \tb & (\zb,\tb) \\
    fermionic operators & \chi_i(t) & \psi_{ia}(z,t),~a=L,R \\
    matrix of 2-point Wightman functions & G & G_b\\
    matrix of anticommutators (real part of 2-pt function) & A & A_b \\
    matrix of commutators (imaginary part of 2-pt function) & C & C_b \\
    \bottomrule
\end{NiceTabular}  
\end{center}
We shall normalize the Majorana operators such that
\begin{align}
    \{\chi_j(t),\chi_k(t)\} = \delta_{jk} \implies \chi^2_j(t) = \frac12\Id
\end{align}
Bulk circuit coordinates $(z,t)$ are discretized in units of $\Delta t/2$, where $\Delta t$ is the discretization of the boundary time and $z \equiv \zb\Delta t$ and $t\equiv \tb\Delta t$ where $\zb,\tb$ are either both integers or both half-integers. We will go between labelling operators and gates with the continuum coordinate labels $z,t$ or the discrete integer coordinates $\zb,\tb$, depending on the context. In other words, we will use the notations $\chi_j(t), \psi_{ja}(z,t), U(z,t), V(z,t)$ and $\chi_j(\tb), \psi_{ja}(\zb,\tb), U(\zb,\tb), V(\zb,\tb)$ interchangeably, depending on the context.

\section{The SYK model thermal state}\label{sec: syk}
In this section and the next, we will apply our reconstruction procedure to various SYK-type models. The SYK model consists of $N$ Majorana fermions with random $q$-body interactions,  and is of interest as a low-dimensional toy model for holography. At large $N$ and strong coupling (the latter of which is often given in the form of a low temperature limit), the SYK model acquires an approximate conformal symmetry in the IR, where its low-energy behavior is captured by (1+1)d JT gravity---a model that also arises as the effective action near the horizon of an extremal black hole. Accordingly, the SYK model exhibits many features expected of black holes, including fast scrambling and a macroscopic ground state degeneracy in its semiclassical limit. However, the dual of SYK is not a simple JT gravity theory since there is a tower of other correlation functions that are not captured by JT gravity~\cite{sarosi2017ads,gross2017all,gross2017bulk}. It is an open question whether the SYK model admits a bulk dual description beyond the low-temperature limit. Using our generalized HKLL construction, we may determine the bulk fermion operators and their dynamics from the SYK two point functions. We show that the boundary thermal state admits a dual description as free bulk fermion dynamics in (1+1)d with a horizon, even away from the low temperature limit. In the dual description, we can study the bulk quantum state and its properties such as entanglement entropy. Furthermore, we can apply our construction to the thermofield double state, and confirm that the low-temperature physics of the bulk dual is consistent with a two-sided black hole geometry. Later in Sec.~\ref{sec: coupled_syk}, we study coupled SYK models, which correspond to two-sided geometries which are modified by positive or negative energy shockwaves. For our investigations, we use the analytic results for the large-$q$, finite-temperature two point functions in the usual SYK model~\cite{maldacena2016comments} and the coupled, two-sided variant~\cite{maldacena2018eternal,lensky2021rescuing}.

\subsection{Overview of the SYK model and its generalizations}
In this subsection, we briefly review the main results of the SYK model. For key computations in the SYK model, see \cite{maldacena2016remarks,maldacena2016comments,sachdev2015bekenstein}, and for pedagogical reviews on the SYK model and its connection to holography, see \cite{sarosi2017ads,trunin2021pedagogical} and references therein.

The SYK model describes a system of $N$ Majoranas fermions $\{\chi_i\}_{i=1}^N$ which are Hermitian representations ($\chi_i^\dagger = \chi_i$) of the algebra $\{\chi_i,\chi_j\} = \delta_{ij}$. Finite-dimensional representations of the Majorana algebra may be found, and are given by square matrices of minimum dimension $2^{N/2} \times 2^{N/2}$. The SYK Hamiltonian is given by
\begin{equation}
\begin{aligned}\label{eq: syk_hamilt}
    H_\text{SYK} &= i^{q/2} \sum_{1\leq i_1< \dots< i_q \leq N} J_{i_1,\dots,i_q} \chi_{i_1} \dots \chi_{i_q}
\end{aligned}
\end{equation}
where the coupling constants $J_{i_1,\dots,i_q}$ are random real variables. One can also regard the coupling constants as the entries of a random antisymmetric tensor $J$ whose antisymmetry enforces the anticommutation relations of the Majoranas. Typically, the coupling constants are i.i.d. Gaussian variables of mean zero and variance $\sigma^2(\SJ)$ where $\SJ$ is a tunable parameter which dictates the characteristic energy scale. Following the convention in \cite{maldacena2016comments}, we will choose the variance to be
\begin{align}\label{eq:var_scaling}
    \sigma^2(\SJ) \equiv \overline{J_{i_1,\dots,i_q}^2} = \frac{2^{q-1}}{q}\frac{\SJ^2(q-1)!}{N^{q-1}}
\end{align}
where the scaling of the variance with $N$ is chosen to produce sensible large-$N$ asymptotics.

One can think of the SYK model as representing an ensemble of Hamiltonians, each with a different instantiation of the random couplings. As with any model with disorder, one typically considers the disorder-average of physical quantities like correlators. For $q \geq 4$, the model is self-averaging for many quantities, meaning that a typical instantiation of $J$ gives values close to the disorder-average. 

The SYK model is analytically treated at large $N$, which is the regime in which it classicalizes. In this limit, one can find equations for the normalized two point function,
\begin{align}
    G(t,t') = \frac{1}{N}\sum_{i=1}^N \ev{\chi_i(t)\chi_i(t')}_\beta
\end{align}
where $\ev{.}_\beta$ denotes the expectation value computed in the thermal or KMS state with inverse temperature $\beta$. The Schwinger-Dyson equations for $G(t,t')$ are given by
\begin{align}\label{eq: schwinger_dyson}
    \partial_{t}G(t,t') - \int dt'' \Sigma(t,t'')G(t'',t') = \delta(t-t') \qquad \Sigma(t,t') = \frac{\SJ^2}{q} \left(2G(t,t')\right)^{q-1}
\end{align}
where $\Sigma$ plays the role of the self-energy (i.e. sum of all one-particle irreducible diagrams) in the Feynman diagrammatic approach. Because of the time translation invariance of the bilocal fields $G$ and $\Sigma$, the first equation is often written in Fourier space, giving the pair of equations
\begin{align}
    \frac{1}{G(\omega)} = \frac{1}{G_0(\omega)}- \Sigma(\omega)\qquad \Sigma(t) = \SJ^2\left(2G(t)\right)^{q-1}
\end{align}
where $G_0(\omega)$ is the Fourier transform of the free two-point function and $t$ in the second equation is the difference between the two times in the bilocal fields. For finite $q$ and finite temperature, Eq.~\eqref{eq: schwinger_dyson} can only be solved numerically. 
In the large-$q$ limit defined by $N,q \to \infty$ with $q^2/N \to 0$, an analytic solution can be obtained for all temperatures, and is given by:
\begin{equation}
\label{eq: G_largeq}
    \begin{aligned}
        G(t) &= \frac{1}{2}\left[\frac{\cos(\pi\nu/2)}{\cos\left(\pi\nu\left(\frac{1}{2}-i\frac{t}{\beta}\right)\right)} \right]^{2/q} = \frac{1}{2}\left[\frac{\cos(\pi\nu/2)}{\cosh\left(\pi\nu\left(\frac{i}{2}+\frac{t}{\beta}\right)\right)}\right]^{2/q}\\
        \beta \SJ &= \frac{\pi\nu}{\cos(\pi\nu/2)}
    \end{aligned}
\end{equation}
Note that this two point function decays exponentially for finite $\beta$, as a consequence of thermalization. 
From the quantum mechanical perspective, this is because local $\chi$ operators are developing components on products of many single-fermion operators. As we will show later, using our holographic circuit construction of the SYK model, we can attribute a bulk geometric explanation for the decay of correlations: localized boundary fermion wavepackets propagate deeper into the bulk and escape into a bulk bath, which may be interpreted as an effective horizon.

The fact that the Majorana fermion fields are GFFs in the large $N$ limit is a consequence of the suppression of connected four-point function by $\frac 1N$. In the path integral language, the disorder-averaged partition function can be expressed as a path integral over collective fields $G(t,t')$ and $\Sigma(t,t')$, and this path integral is dominated by the saddle point in the large $N$ limit. The saddle point equations are precisely the Schwinger-Dyson equations given in Eq.~\ref{eq: schwinger_dyson}.

In the low energy limit $1\ll \beta\mathcal{J}\ll N$, the model has an approximate symmetry under reparameterization, 
\begin{equation}
    \begin{aligned}
        G(t,t') \mapsto \left[f'(t)f'(t')\right]^{\Delta} G\left(f(t),f(t')\right)\\
        \Sigma(t,t') \mapsto \left[f'(t)f'(t')\right]^{\Delta(q-1)}\Sigma\left(f(t),f(t')\right)
    \end{aligned}
\end{equation}
which is broken by the $i\omega$ term in Eq.~\eqref{eq: schwinger_dyson}. As a consequence, long wavelength fluctuations are described by the Schwarzian theory of the reparametrization field $f(\tau)$~\cite{maldacena2016remarks,kitaev2018soft}. The low temperature solution has the form
\begin{align}
    G(t) = \sgn{t}\frac{c_\Delta}{\abs{t}^{2\Delta}} \qquad c_\Delta = \left(\left(\frac{1}{2}-\Delta\right)\frac{\tan(\pi\Delta)}{\pi\Delta}\right)^{\Delta}
\end{align}
which is conformally-invariant under $\SL(2,\R)$ transformations with conformal dimension $\Delta = \frac{1}{q}$. The Schwarzian theory agrees with the 2d JT gravity, which is the grounds for the statement that the SYK model admits a holographic interpretation. The fact that JT gravity emerges as the effective dimensionally-reduced theory of the near-horizon physics of an extremal black hole then connects the physics of the SYK model to that of black holes. Indeed, the SYK model reproduces many of the features of near extremal black holes, including scrambling behavior and a macroscopic zero temperature degeneracy. Since the holographic duality of SYK model relies on the Schwarzian theory, the dual theory is not known for finite $\beta\mathcal{J}$. One key result of our approach is a bulk dual theory that applies to arbitrary finite temperature, at the level of two point functions. 

\subsection{Bulk reconstruction for the SYK thermal state}\label{sec: circuit_and_horizon}

In this subsection, we perform the reconstruction detailed in Sec.~\ref{subsec: kernel} using the thermal two point function given by Eq.~\eqref{eq: G_largeq}. In the large-$N$ limit, the two point function is diagonal in the Majorana flavor indices, {\it i.e.} $G_{ij}(t,t')=G(t,t')\delta_{ij}$. Our construction starts from the anticommutators of the SYK Majorana fermions, which is twice the real part of the Wightman function:
\begin{align}
    A_{ij}(t,t')\equiv \left\langle\left\{\chi_{i}(t),\chi_{j}(t')\right\}\right\rangle= 2\Re\left[G(t,t')\right]\delta_{ij}
\end{align}
We construct $A$ by considering a discrete set of times $t,t'=n\Delta t,~n\in \mathbb{Z}$, and then obtain the kernels $K_{f,p}$ via QR (QL) decomposition of $A^{-1/2}$, as in Eq.~\eqref{eq: QR}. Note that because the matrix $A$ is diagonal in the Majorana flavor indices, the resulting kernel matrix will also be diagonal in the flavor indices and possess identical components for each index $i$. Thus, the bulk Majorana fermions along the future and past light cones may be written as
\begin{align}
    \psi_{R(L)i}(z,t)=\sum_{t-z\leq t'\leq t+z}K_{f(p)}(z,t|t')\chi_{i}(t')
\end{align}
where the right (left) moving fermions reside on the future (past) light cones. As a consequence, the bulk dual that we will construct may be thought of as comprised of $N$ identical copies of a (1+1)d bulk Majorana theory. We will henceforth omit the flavor indices for both the bulk and boundary fermions since they play no role in the construction.

\begin{figure}[t]
    \centering
    \includegraphics[width=1\linewidth]{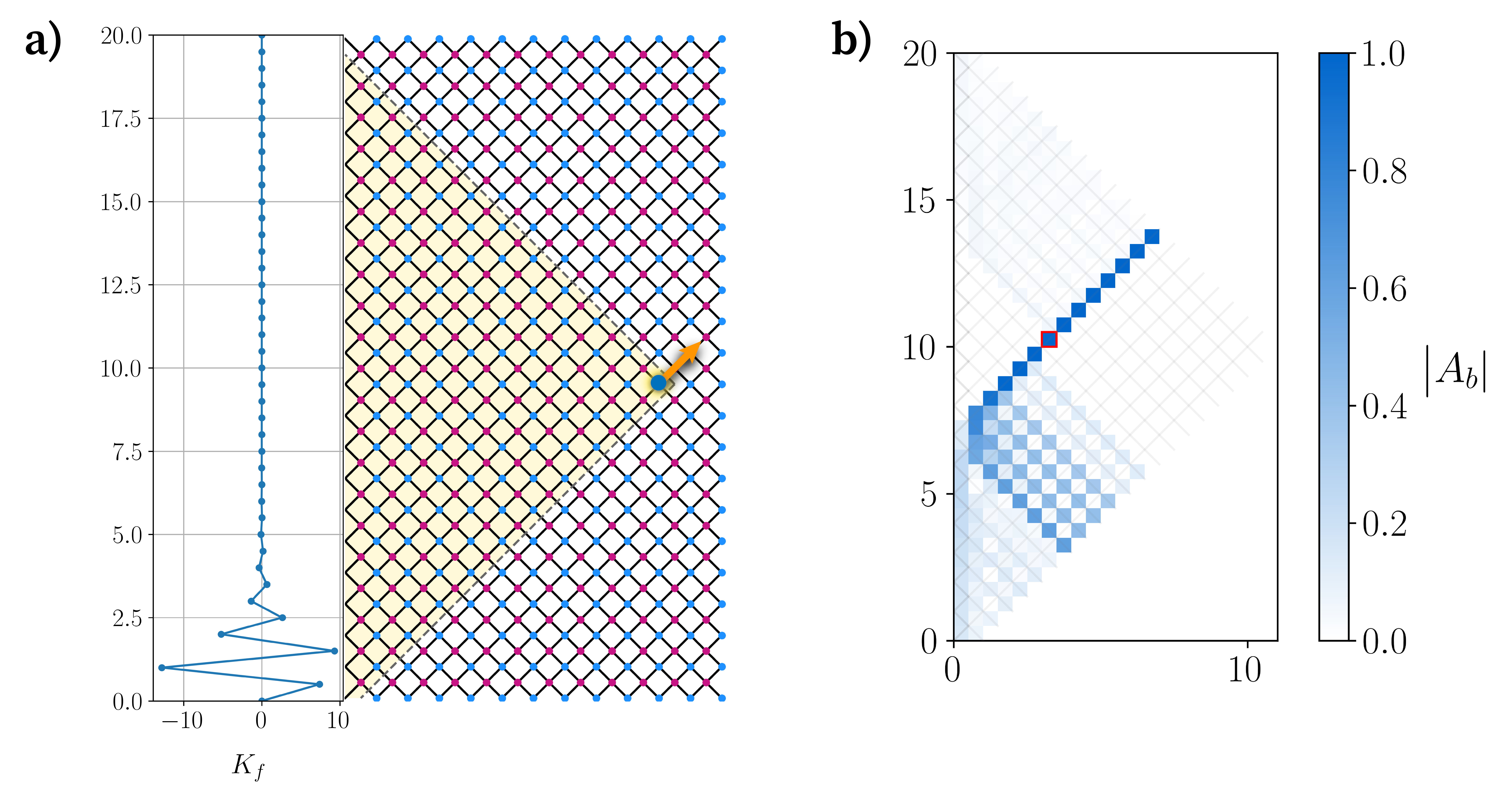}
    \caption{Bulk fermions for the thermal state of large-$q$ SYK model for $(q,\beta,J,\Delta t) = (16,10^6,1,0.5)$. a) We show the components of the kernel $K_f(z,t|t')$ for a right-moving bulk fermion at a fixed location $(z,t)$, indicated schematically on the circuit diagram by the dot and arrow. The left panel is a plot of $K_f(z,t|t')$ (horizontal axis) as a function of $t'$ (vertical axis). b) Evidence of locality as indicated by the bulk propagator \eqref{eq: bulk propagator}. The first bulk fermion in the anticommutator is a right-moving fermion fixed at the circuit leg in the red outlined square. The location of the second bulk fermion is swept over all locations in the reconstructed causal wedge. As the colormap shows, the anticommutator vanishes for spacelike-separated fermions, and is strongly supported in the light-like direction corresponding to the direction of propagation of the first bulk fermion.}
    \label{fig: kernel}
\end{figure}

An example of the kernel components for a right-moving bulk fermion $\psi_{Ri}(z,t)$ is shown in Fig.~\ref{fig: kernel}(a). The bulk fermion is a linear superposition of boundary fermions with an oscillatory coefficient, as shown in the left panel of the subfigure. Note that the kernel for the right-moving bulk fermion is mainly supported in the past end of the boundary interval $[t-z,t+z]$. Likewise, the left-mover at $(z,t)$ will be a superposition of the boundary fermions whose support is concentrated on the future end of the same interval. 
We may use the kernels to understand the bulk dynamics via the bulk-bulk anticommutator
\begin{align}
    {A_b}\left(c,z',t',d,z,t\right)=\left\langle\left\{\psi_c(z',t'),\psi_d(z,t)\right\}\right\rangle\label{eq: bulk propagator}
\end{align}
with $c,d \in \{L,R\}$. Because the bulk operators are canonically free (App.~\ref{app: bulk_locality}), ${A_b}\left(c,z',t',d,z,t\right)$ represents the probability amplitude $\psi_d(z,t)$ evolves to $\psi_c(z',t')$. Furthermore, because spacelike-separated bulk operators anticommute, $A_b$ is guaranteed to be an orthogonal matrix, with one entry indexed by $c,z',t'$ and the other entry indexed by $d,z,t$. Fig.~\ref{fig: kernel}(b) shows the absolute value of the matrix elements of $A_b$ where the first bulk fermion in the anticommutator is fixed to be a particular right-mover $c=R$ at location $(z',t')$ indicated by the pixel with red boundary. Each pixel resides at the center of a link in the circuit, so the figure captures the anticommutator with both right- and left-moving bulk fermions. The color intensity of each pixel represents the magnitude of the corresponding matrix element of $A_b$. We see that evolution to the past brings the local bulk fermion to a superposition of boundary and bulk fermions, while the evolution to the future is mainly a free motion with the speed of light ({\it i.e.} maximum speed in the quantum circuit), with some small probability of reaching other points in the future light cone. This is consistent with the existence of black hole horizon, which we will analyze below.


To understand the bulk dynamics more quantitatively, we obtain the gates in the bulk circuit. In particular, we find the orthogonal linear transformation that corresponds to the RHS of Eq.~\eqref{eq: Uzt}. For a given gate in the bulk, the input and output fermions are both expressed as linear superpositions of the boundary fermions, so that one can find the linear transformation corresponding to the gate itself. Because the flavor indices decouple in the large-$N$ SYK model, the gate corresponds to a simple $2\times 2$ linear transformation. 
A gate at vertex $(z,t)$ acts on the incoming operators $\psi_R(z-\frac{1}{2},t-\frac{1}{2})$ and $\psi_L(z,t)$, and outputs the operators $\psi_L(z-\frac{1}{2},t+\frac{1}{2})$ and $\psi_R(z,t)$. For convenience, let us call the two inputs $\psi_a$ and $\psi_b$, and the two outputs $\psi_c$ and $\psi_d$, respectively, and denote the gate of interest by $U$ (\cref{fig: gate}).
\begin{figure}[H]
    \centering
    \begin{tikzpicture}
        \node[draw,circle, minimum size=1cm] (c){$U$};
        \node [below left=.5cm of c.south west] (ll) {$\psi_a$} edge [thick] node [] {} (c.south west) ;
        \node [below right=.5cm of c.south east] (lr) {$\psi_b$} edge [thick] node [] {} (c.south east) ;
        \node [above left=.5cm of c.north west] (ul) {$\psi_c$} edge [thick] node [] {} (c.north west) ;
        \node [above right=.5cm of c.north east] (ur) {$\psi_d$} edge [thick] node [] {} (c.north east) ;
    \end{tikzpicture}
    \caption{Schematic of an arbitrary bulk gate $U$ with labeled legs.}\label{fig: gate}
\end{figure}
The gate $U$ acts linearly on the bulk operators as:
\begin{align}
    \begin{pmatrix} \psi_c\\  \psi_d \end{pmatrix}
    = \begin{pmatrix}
        U_{ca} & U_{cb}\\
        U_{da} & U_{db}
    \end{pmatrix}
    \begin{pmatrix} \psi_a\\ \psi_b \end{pmatrix}
\end{align}
Because the gate preserves the bulk operator norm ($\psi^2 = \frac{1}{2}$) and the spacelike separation of the pair of inputs and the pair of outputs ($\{\psi_a,\psi_b\} = \{\psi_c,\psi_d\}$ = 0, i.e. the operators anticommute), the gate is an orthogonal matrix. The input and output fermions are linear superpositions of boundary fermions in different time windows. We denote this by $\psi_i=K_i\chi$ with $i=a,b,c,d$. $K_i$ are corresponding row vectors of the kernel matrix \eqref{eq: kernel_matrix}. Thus the orthogonal matrix $U$ can be determined from $K_a,K_b$:
\begin{align}
    \begin{pmatrix} \text{------}\; K_c \;\text{------}\\~\\ \text{------}\;K_d \;\text{------} \end{pmatrix} \begin{pmatrix} | \\ \chib\\ |\end{pmatrix} \crel{=} \begin{pmatrix}
        U_{ca} & U_{cb} \\ U_{da} & U_{db}
    \end{pmatrix}\begin{pmatrix} \text{------}\;K_a\;\text{------}\\~\\ \text{------}\;K_b\; \text{------} \end{pmatrix}\begin{pmatrix}| \\ \chib\\ |\end{pmatrix}\nonumber \\
    \crel{\Downarrow}\nonumber \\
    \begin{pmatrix}\label{eq: circuit_gate}
        U_{ca} & U_{cb} \\ U_{da} & U_{db}
    \end{pmatrix} \crel{=} \left(\begin{array}{cc}K_c^\T AK_a&K_c^\T AK_b\\ K_d^\T AK_a&K_d^\T AK_b\end{array}\right)
\end{align}
In a generic system, the kernel needs to be separately obtained for the four fermions $a,b,c,d$, but for the thermal state the computation can be simplified due to time translation invariance of the state. The time translation invariance also implies time reflection symmetry of the matrix $A(t,t')$, since $G^*(t,t')=G(t',t)$ and thus $A(-t,-t')=A(-t',-t)=A(t,t')$. Therefore, for the time translation invariant state, the kernels $K_p$ and $K_f$ are related by time reflection: $K_p(z,t|t')=\pm K_f(z,t|2t-t')$. The $\pm$ sign is an ambiguity, which does not affect the canonical anticommutation relations of the bulk fermions on $\Sigma_{f,p}$. However, the gauge choice for this sign must respect fermion parity conservation. In general, for $2M$ Majorana fermions $\psi_k,k=1,2,...,2M$, if we implement an orthogonal transformation $\tilde{\psi}_k=\sum_lO_{kl}\psi_l$, with $O^\T O=\Id$, then the associated fermion parity operators $F=i^M\psi_1\psi_2...\psi_{2M}$ and $\tilde{F}=i^M\tilde{\psi}_1\tilde{\psi}_2...\tilde{\psi}_{2M}$ are related by $\widetilde{F}={\rm det}\left(O\right)F$. A physical gate acting on fermions should always preserve the fermion parity number, which means we should require ${\rm det}\left(U(z,t)\right)={\rm det}\left(V(z,t)\right)=1$. Therefore, $U,V$ are $2\times 2$ special orthogonal matrices which may be parametrized as
\begin{equation}
\begin{aligned}\label{eq: gate 2by2}
    U(z,t) &= \begin{pmatrix}
        \cos(\theta_U(z)) & -\sin(\theta_U(z))\\
        \sin(\theta_U(z)) & \cos(\theta_U(z))
    \end{pmatrix}\\[1em]
    V(z,t) &= \begin{pmatrix}
        -\cos(\theta_V(z)) & -\sin(\theta_V(z))\\
        \sin(\theta_V(z)) & -\cos(\theta_V(z))
    \end{pmatrix}
\end{aligned}
\end{equation}
where the gates $U(z,t)$ and $V(z,t)$ each depend on a single parameter $\theta_U(z)$, $\theta_V(z)$. The gate angles $\theta_{U,V}(z)$ are plotted in Fig. \ref{fig: gate_angle}(a). It can be shown that this gauge choice can be realized by requiring $K_f$ to always have positive diagonal components, and requiring $K_p$ to have an alternating sign defined by
\begin{align}
    K_p(z,t|t')=(-1)^{2z}K_f(z,t|2t-t')
\end{align}

\begin{figure}[ht]
    \centering
    \includegraphics[width=\linewidth]{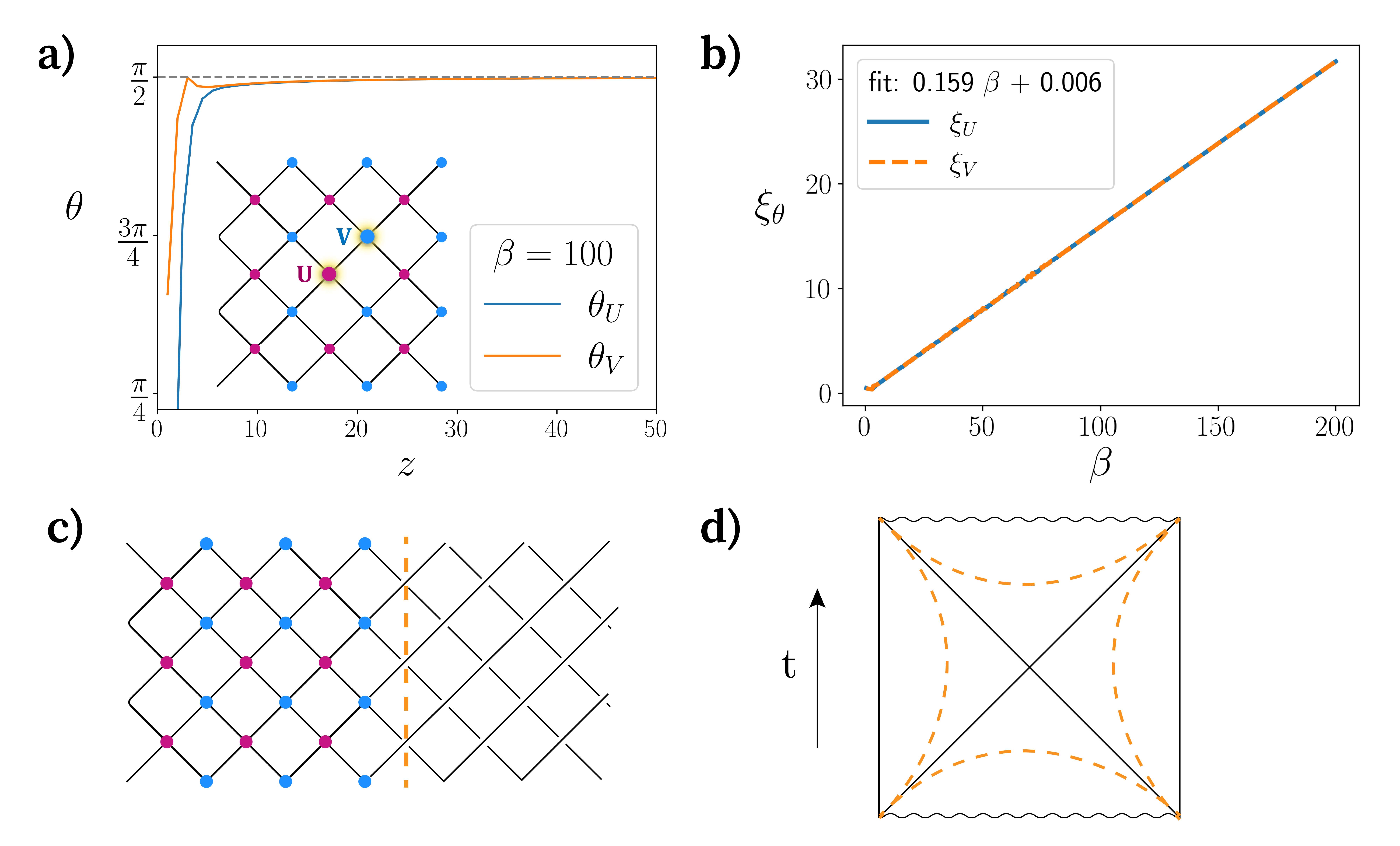}
    \caption{Depth dependence of the gate angles of the bulk circuit and the emergence of an effective horizon. The data shown is obtained for $(q,\SJ,\Delta t)=(22,1,0.5)$. a) The gate angles $\theta_U,\theta_V$ defined in Eq.~\eqref{eq: gate 2by2} both quickly approach $\frac{\pi}{2}$ at large $z$. b) The characteristic depth $\xi$ that controls the exponential decay $\frac{\pi}2-\theta_{U,V}\propto e^{-z/\xi_\theta}$, which depends linearly on $\beta$ at low temperature. c) A schematic of the bulk circuit dual to the thermal state. The region $z\gg \xi_\theta$ demarcated by the orange dashed line has $\theta_U,\theta_V\simeq \frac{\pi}2$, which means the left-movers and right-movers decouple. Note that this is consistent with Fig. \ref{fig: kernel}(b). d) A schematic of the maximally-extended, two-sided Penrose diagram for an eternal AdS$_2$ black hole. The $z\gg \xi$ regions in subfigure (c) corresponds to the near horizon region in the Penrose diagram. }
    \label{fig: gate_angle}
\end{figure}

\subsection{Emergence of black hole horizon}
Due to time translation invariance of the thermal state correlation functions, the $U$ (half-integer) and $V$ (integer) gates only depend on the $z$ coordinate. As shown in Fig.~\ref{fig: gate_angle}(a) and (b), we find that upon increasing $z$, the gate angle for both $U$ and $V$-type gates exponentially approaches $\frac{\pi}{2}$:
\begin{align}
    \theta_\beta(z) \sim \frac{\pi}{2}-e^{-z\Delta t/\xi_\theta}
\end{align}
with a characteristic length $\xi_\theta\propto \beta$ at low temperature. As a consequence, in the region $z\gg \xi$ each gate is well-approximated by a swap gate, as is illustrated in Fig. \ref{fig: gate_angle}(c). In this region, the right-moving fermion modes decouple from the left-moving ones, which means the right-moving fermion will "fly freely" for ever, with the deep bulk interior playing the role of a reflectionless geometric bath\footnote{In the many-body notation, the bulk fermion operators in the swap region satisfy $\psi_R(z,t)=\psi_R\left(z+\frac12,t+\frac12\right)$, $\psi_L(z,t)=-\psi_L\left(z-\frac12,t+\frac12\right)$ for the gauge choice we made.}. This explains why in Fig. \ref{fig: kernel}(b), the right-moving fermion propagates along the light cone, almost without dispersion.

The existence of a reflectionless bath region is a direct consequence of the decay of the two point function. As long as the two point function does not have a recurrence at long times, the reflectionless bath region is semi-infinite. Physically, this is because the bulk fermions are obtained by a constrained orthogonalization procedure. The right-mover $\psi_R(z,t)$ is obtained by taking boundary fermion $\chi(t-z)$ and orthogonalizing it with all $\chi(t')$ in the range $(t-z,t+z]$. By comparison, $\psi_R\left(z+\frac12,t+\frac12\right)$ (the next fermion along the lightlike direction of propagation) is obtained by orthogonalization of $\chi(t-z)$ with $\chi(t')$ in the range $(t-z,t+z+1]$. Therefore, $\psi_R\left(z,t\right)$ and $\psi_R\left(z+\frac12,t+\frac12\right)$ differ only by an extra orthogonalization with $\chi(t+z+1)$. If the boundary time separation $2z+1$ is much larger than the thermal correlation length, $\chi(t-z)$ and $\chi(t+z+1)$ are almost already orthogonal, so $\psi_R\left(z,t\right)$ and $\psi_R\left(z+\frac12,t+\frac12\right)$ are almost the same. Analogous reasoning holds for $\psi_L(z,t)$ and $\psi_L\left(z+\frac12,t-\frac12\right)$. 

The presence of a reflectionless bath is also consistent with the holographic description of the thermal equilibrium state of SYK as an eternal black hole. In this interpretation, the thermal correlators of the SYK model at low temperatures are dual to the thermal correlators of JT gravity for an AdS black hole. The region $z\gg \xi_\theta$ corresponds to the region near the black hole horizon, as illustrated in Fig. \ref{fig: gate_angle}(d). The decoupling between left-movers and right-movers can be viewed as a consequence of the gravitational redshift near the horizon. When a bulk fermion falls into the black hole, the proper time it spends in the near-horizon region is short, even if the corresponding asymptotic boundary time (in the exterior coordinate system) is infinite. The free-flying, right-moving fermion modes correspond to the infalling modes crossing the horizon. More quantitative analysis of the bulk geometry will be presented in Sec.~\ref{sec: bulk_curvature}. 

Importantly, our approach clarifies that for generalized free fields the concepts of a black hole horizon and horizon modes are well-defined for all systems with exponentially decaying two point functions, even if these systems do not have a holographic dual in the usual sense (such as the finite temperature SYK model). This observation suggests that holographic duality between chaotic models and black hole geometries may occur more generally than one would conservatively expect from the AdS/CFT correspondence. However, it should be noted that our approach does not directly distinguish the black hole horizon from other causal horizons such as an AdS Rindler horizon. If we consider the state in the limit $\beta\mathcal{J}\rightarrow +\infty$, the boundary two point function has a power-law decay behavior, so the corresponding gate angles $\theta_U,\theta_V$ approach $\frac{\pi}2$ in large $z$ but with a deviation that decreases as a power law of $z$. In this case, the bulk fermion also has a reflectionless bath, which corresponds to the Poincare horizon in AdS$_2$ space.

Now let us further discuss the boundary interpretation of the horizon modes. The fact that the gates approach $\theta=\frac{\pi}2$ at large $z$ guarantees that the right-movers (left-movers) converge to a well-defined mode at large bulk depths. For instance, the future horizon modes are well-defined as the following limit of the right-moving fermion operators:
\begin{align}
    \psi^f(u)&\equiv\lim_{z\rightarrow \infty}\psi_R(z,z+u) \qquad \text{with }u = n\Delta t
\end{align}
where $u=0$ corresponds to the furthest interior point on the future light cone of the causal wedge, and $z=\zb\Delta t$ with half-integer $\zb$. Note that the horizon modes are only well-defined at discrete points $u$ along the light cone. Recall that for different $z$, $\psi_R(z,z+u)$ is a superposition of boundary fermions in the interval $[u,u+2z]$ where the coefficients are given by the kernel $K_f(z,z+u|\tau+u)$ and $\tau \in [0,2z]$ labels the time relative to the first time in the interval. For large $z$, the kernel is mostly supported near the initial time $u$ with a spread of $\xi_\theta$, and thus converges to a $z$-independent "wavepacket", as illustrated earlier in Fig. \ref{fig: kernel}(a). The bulk fermions on the future horizon can be expressed in terms of the boundary fermions as
\begin{align}
\psi^f(u)&=\sum_{\tau=0}^{+\infty}W(\tau)\chi(\tau+u),\qquad W(\tau)=\lim_{z\rightarrow \infty}K_f(z,z-\tau)
\end{align}
where we have written the kernel as $K_f(z,z-\tau)=K_f(z,z+u|\tau+u)$ since it only depends on the time difference $z-\tau$ due to time translation symmetry. The past horizon modes can be analogously obtained by $\psi^p(u)=\lim_{z\rightarrow \infty}\psi_L(z,z-u)$, and are supported on the boundary time interval $(-\infty,u]$. The future and past horizon modes for the SYK fermions are discussed in Appendix A and Ref.~\cite{gu2022two}.

To understand the physical interpretation of the horizon modes, note that
\begin{align}
    \left\{\psi^f(u),\chi(t)\right\}=0,~\forall t>u\label{eq: horizon anticommutator}
\end{align}
which follows from the orthogonalization procedure. Consider a boundary probe, which is a generic bosonic operator $O \equiv \chi_{j_1}(t_1)\chi_{j_2}(t_2)\dots \chi_{j_{2n}}(t_{2n})$ with $t_k > u \forall k=1,2,...,2n$. Eq. \eqref{eq: horizon anticommutator}) implies that
\begin{align}
 \label{eq: horizon_comm}
    \ev{[\psi^f(u),O]} &= 0
\end{align}
Therefore, $\psi^f(u)$ is {\it invisible} for arbitrary boundary measurements made after time $u$ and involving a finite number of fermions. In fact, for any given time $u$, all future horizon modes $\psi^f(u')$ with $u'\leq u$ commute with a generic operator $O$ at $t>u$, so the same must true for a product of such operators. As a consequence, there is an infinite family of quantum states
\begin{align}
\sigma\left(u_1,u_2,...,u_\ell\right)=\Bigr(\psi^f\left(u_1\right)\psi^f\left(u_2\right)...\psi^f\left(u_\ell\right)\Bigr)\rho_{\rm th} \Bigr(\psi^f\left(u_\ell\right)\psi^f...\left(u_2\right)\psi^f\left(u_1\right)\Bigr)\end{align}
with $u_j\leq u~\forall j=1,2,...,\ell$ which are all indistinguishable from the thermal equilibrium state $\rho_{\rm th}$ after time $u$:
\begin{align}
    {\rm tr}\left[O\left(\sigma-\rho_{\rm th}\right)\right]=0
\end{align}
The above is an explicit consequence of the fact that the operator algebra for all boundary operators after a certain time (or more generally, for the spacetime region in the future of a Cauchy surface) for GFFs depends nontrivially on the Cauchy surface~\cite{leutheusser2021emergent}.

Physically, the invisibility of the horizon modes with respect to operators after a certain time is equivalent to the statement that the information carried by the horizon modes is sent into the interior of the black hole and cannot be recovered. Of course this is only true in the large-$N$ limit. At finite $N$, the commutator in Eq.~\ref{eq: horizon_comm} acquires a finite value $\propto \frac1N$, which allows infalling qubits to be reconstructed by the boundary observer. 


\subsection{Bulk state and entanglement entropy}
\label{subsec: bulk entropy}
As mentioned in Sec.~\ref{subsec: kernel}, the bulk state may be characterized using the reconstructed bulk operators. Since the bulk state is Gaussian, the particular state and all of its correlation functions are determined by the bulk two point functions. Let us temporarily denote the bulk fermions on a Cauchy surface $\Sigma$ as $\psi_\Sigma(z)$ since no two fermions on the same Cauchy surface will have the same $z$. Because the $\psi_\Sigma(z)$ are spacelike-separated, the anticommutators vanish independent of the bulk state, and the two point functions on $\Sigma$ reduce to the bulk commutators:
\begin{align}
    C_{b\Sigma}(z,z')&\equiv\left\langle\left[\psi_\Sigma (z),\psi_\Sigma (z')\right]\right\rangle
\end{align}
Given $\psi_\Sigma=K_\Sigma \chi$, we can express $C_{b\Sigma}$ in terms of the boundary commutators:
\begin{align}
    C_{b\Sigma}&=K_\Sigma CK_\Sigma^\T
\end{align}
where $C(t,t')=\left\langle\left[\chi(t),\chi(t')\right]\right\rangle$. The bulk correlators $C_{b\Sigma}(z,z')$ may be analyzed to extract properties of the bulk state.

\begin{figure}[ht]
    \centering
    \includegraphics[width=1\linewidth]{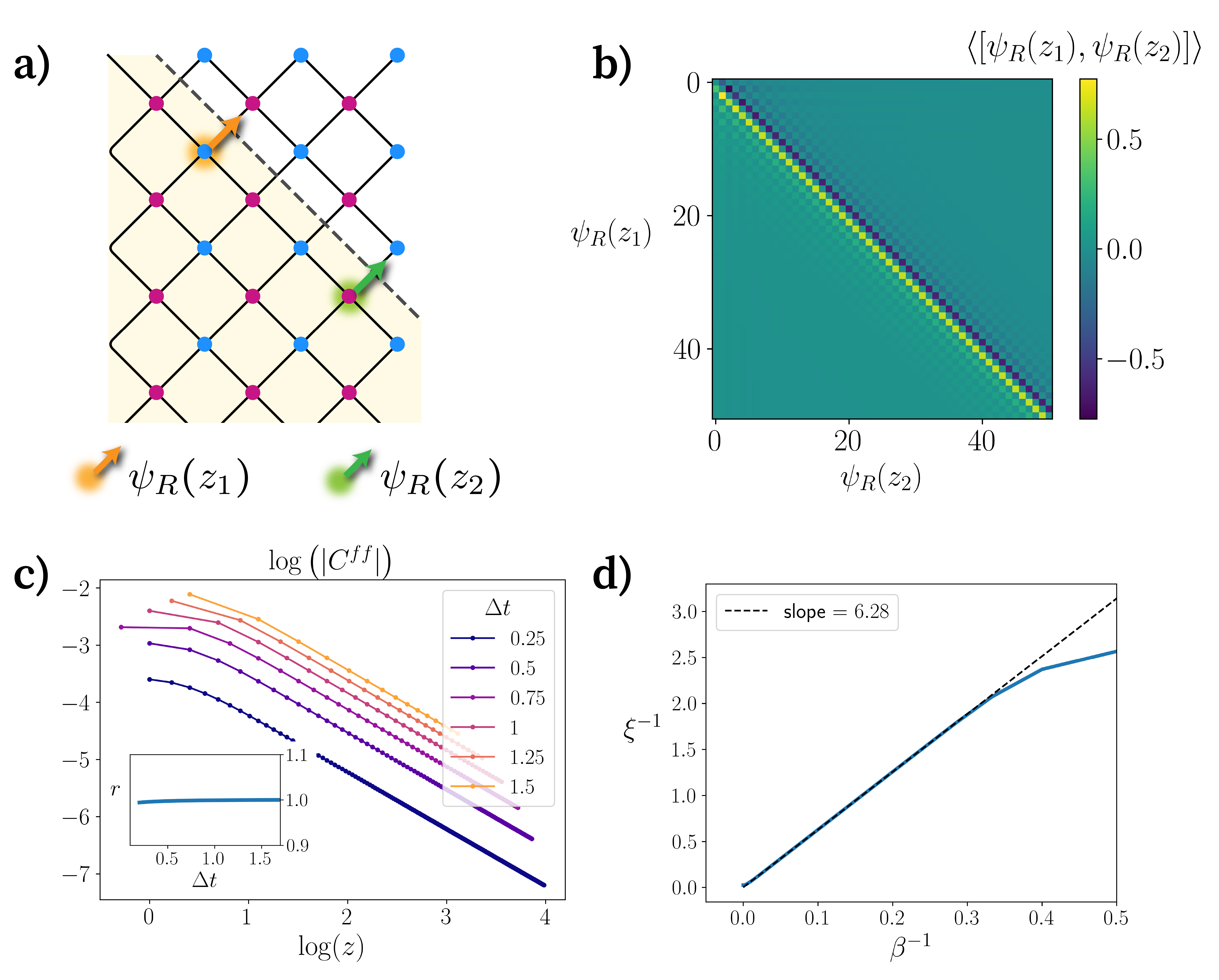}
    \caption{Bulk correlators for $(q,\SJ)=(22,1)$. a) Schematic of a pair of bulk fermions $\psi_R(z_1)$ and $\psi_R(z_2)$ on a Cauchy slice (dashed line) used to study commutators. We chose the Cauchy slice to be the future light cone, and denote the bulk commutators along this light cone by $C^{ff}(z_1,z_2)$. b) Bulk commutators along a Cauchy slice shows evidence of locality in the correlations in the state. Color plot shown is for $\beta=10^5$, $\Delta t = 0.2$. c) Spatial dependence of the bulk commutator on $z\equiv z_2-z_1$ for fixed $\beta = 10^5$. We observe a power law decay of correlations $|C^{ff}|\sim (z\Delta t)^{-r}$, with a constant power law exponent $r\sim 1$ in the continuum limit as shown in the inset. d) For fixed $\Delta t = 0.25$, we fit an exponential to the decay of correlators deeper in the bulk and observe a correlation length $\xi$ that is proportional to $\beta$ with a slope that is approximately $2\pi$, as expected from Eq.~\eqref{eq: G_largeq}.    
    }
    \label{fig: bulk_comm}
\end{figure}

For simplicity, we study the bulk correlators for the right-moving fermions along the future light cone $\Sigma_f$, as is illustrated in Fig.~\ref{fig: bulk_comm}(a). The bulk commutator matrix is thus $C_b=K_fCK_f^\T$. As shown in Fig. \ref{fig: bulk_comm}(b), the correlation function decays with the distance between bulk points. For low temperature ($\beta = 10^5$), the commutator decays like a power law with an exponent that remains fixed in the continuum limit $\Delta t\to 0$, as demonstrated in Fig.~\ref{fig: bulk_comm}(c). This is consistent with the expected power law decay of the boundary two point function at low temperatures. For small, finite $\beta$, the boundary correlators decay approximately exponentially with a correlation length that behaves as $\xi=\frac{\beta}{2\pi}$, as seen from Eq.~\ref{eq: G_largeq}. Fig.~\ref{fig: bulk_comm}(d) shows that the relationship between $\xi^{-1}$ and $\frac1\beta$ is indeed approximately linear, with a slope of $2\pi$. The small flat region at low temperatures is indicative of the power law behavior observed in Fig.~\ref{fig: bulk_comm}(c), while the deviation at high-temperatures is a reflection of the $\nu$-dependence of the decay of the boundary two point function in Eq.~\ref{eq: G_largeq}, which becomes relevant when $\beta^{-1} \sim \Delta t$.


We may further characterize the bulk state by computing its entanglement entropy. For the Gaussian bulk state, the von Neumann entanglement entropy of a bulk subregion is determined by the commutator matrix $C_b$ for that region \cite{peschel2009reduced,peschel2003calculation,araki1970quasifree}:
\begin{align}
    S_{vN}(\rho)&= -\frac12\Tr\left(\frac{1 + C_b}{2}\ln\frac{1 + C_b}{2} + \frac{1 - C_b}{2}\ln\frac{1 - C_b}{2}\right)
\end{align}
Physically, the eigenstates of $C_b$ correspond to complex fermion modes built out of pairs of Majorana fermions, and the eigenvalues of $C_b$ determine their expected occupation numbers. Because $C_b$ is imaginary and Hermitian, its eigenvalues come in pairs $\pm \lambda_i$, with $\lambda_i\geq 0$. The probability of having occupation number $n_i=1$ for the $i^\text{th}$ mode is $p_i=\frac{1-\lambda_i}2\in[0,1]$. The entropy of a bulk subregion is a sum of the entropies of each complex fermion degree of freedom, and essentially counts the number of complex fermions in the subregion.

We numerically study the entropy of a bulk spacelike interval that ends at the boundary. Since states related by a unitary transformation have the same entropy, it suffices to study the entropy of an interval along the future light cone, as indicated by the purple line in Fig.~\ref{fig: entropy}(a). As a technical aside, the bulk entropy can be computed without explicitly obtaining the associated kernel $K$. Given that $K$ satisfies $K=QA^{-1/2}$ with $QQ^\T=\Id$, we can define a "rotated" version of the bulk commutators
\begin{align}\label{eq:entropy_0z}
    \Ct_b=Q^\T C_bQ=A^{-1/2}CA^{-1/2}
\end{align}
and compute the von Neumann entropy from $\Ct_b$ directly, which only involves the boundary quantities $A^{-1/2}$ and $C$. To compute the entropy of a bulk interval along $\Sigma_f$ with spatial coordinates in $[0,z]$, we sample the boundary time interval $[-z,+z]$ with an integer number $n_t$ of points, which results in discretization $\Delta t = \frac{2z}{n_t-1}$. The entropy $S(z,\Delta t)$ is computed using Eq.~\eqref{eq:entropy_0z} and the results are summarized in Fig.~\ref{fig: entropy}. As seen in Fig.~\ref{fig: entropy}(b), for a fixed $z$, $S(z,\Delta t)$ approaches a finite value as $\Delta t\to 0$. We estimate the continuum limit value for the entropy of the interval by parametrizing the small $\Delta t$ regime by a quadratic function $S(z,\Delta t)=S_0(z)+\gamma_1\Delta t+\gamma_2\Delta t^2$. For sufficiently low temperatures, $S_0(z)$ depends linearly on $z$ as shown in Fig.~\ref{fig: entropy}(d), which means the bulk entropy along the light cone satisfies a "volume law" in the particular continuum limit considered (finite total boundary time with $\Delta t \to 0$). Importantly, the fact that the entropy appears to have a finite continuum limit suggests that the continuum limit of the circuit is qualitatively different from a quantum field theory, which should have UV-divergent entropy.

\begin{figure}[h]
    \centering
\includegraphics[width=0.9\linewidth]{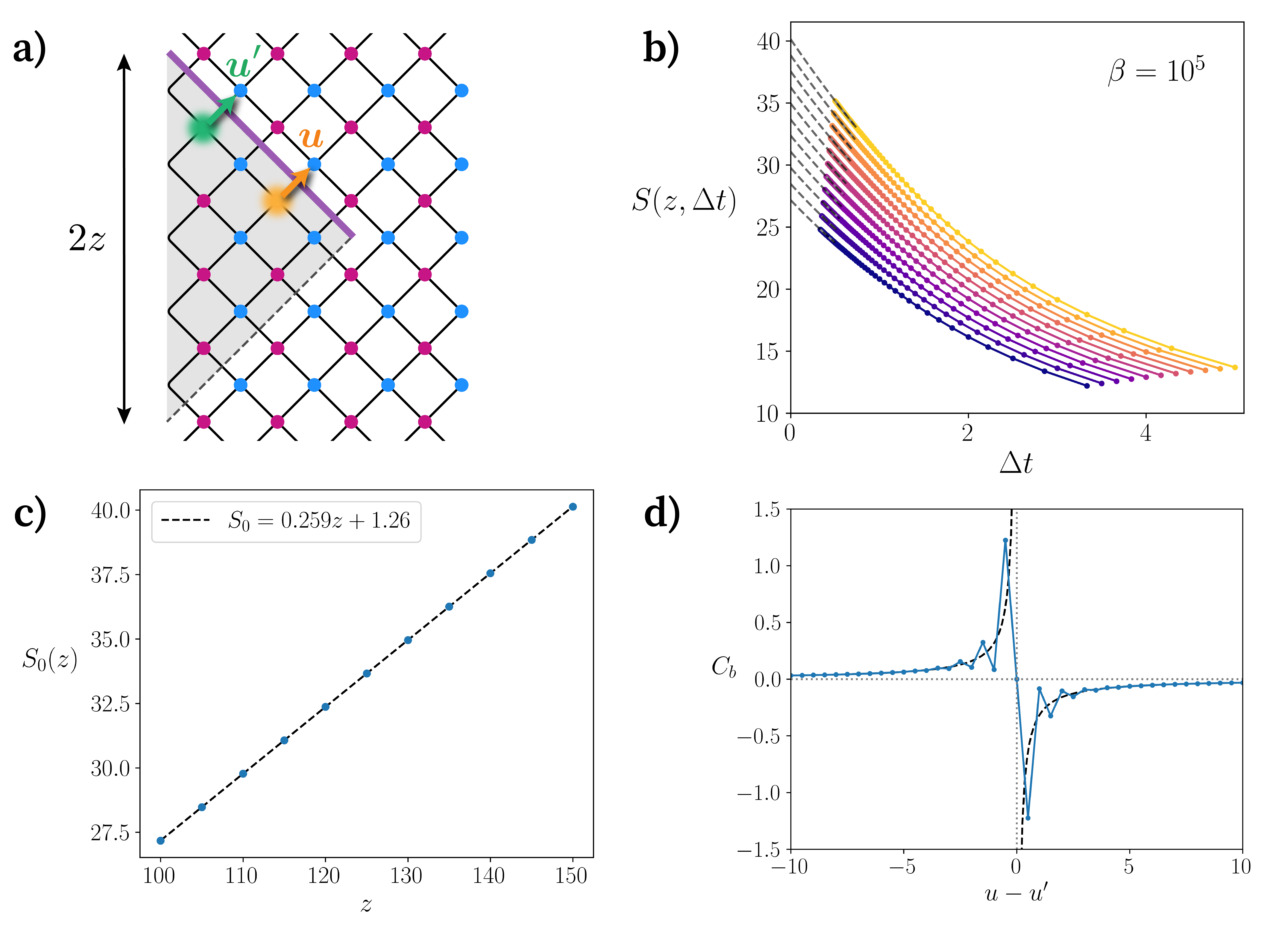}
    \caption{
    Entanglement entropy for a bulk subinterval for $(q,\SJ,\beta) = (22,1,10^5)$. a) The von Neumann entanglement entropy is computed for a subinterval along the future light cone of a fixed total boundary time interval of duration $2z$. b) The entropy $S$ is computed as a function of $\Delta t$ for different values of $z$ in the range $[100,150]$. The colors get darker for ascending $z$ values. The dashed line shows a quadratic fit of $S$ versus $\Delta t$ as $\Delta t \to 0$. c) The continuum limit entropy $S_0(z)$ extrapolated to $\Delta t=0$ as a function of $z$ for $\beta \sim 10^5$. The black dashed line is the linear fit for this value of $\beta$. d) The blue data points show the scaled future light cone correlator $C_b(n\Delta t,n'\Delta t')$ divided by $\Delta t$ for pairs of points on the future light cone for a fixed boundary interval $[-150,150]$ for $\Delta t = 0.5$. In the continuum limit of the near horizon region, we anticipate that $C_b(n\Delta t,n'\Delta t)$ approaches a continuum function $C_b(u,u')$, with the blue points representing samples $u$ whose coordinates are given by well-defined fractions of the total boundary length $2z$. The data shown is for the first fermion fixed at $u\equiv 2t = 240$, which is well into the near horizon region for $\beta=10^5$. The black dashed line is the analytic continuum theory result in \cref{eq:horizon_correlator_continuous}.}
    \label{fig: entropy}
\end{figure}

This conformal coordinate volume law can be understood as a consequence of the universal behavior of $\Ct_b$ deep in the bulk, and the existence of a "near-horizon" regime. For sufficiently large $z$, the bulk fermions along discrete points along the future light cone approach the horizon modes discussed in the previous subsection. For these horizon modes, discrete translation of the bulk coordinate along the light cone corresponds to discrete time translation of the boundary wavepacket representations. In the near-horizon region, the bulk two point functions along the future light cone are approximately lattice translation invariant with lattice constant determined by $\Delta t$. If the near-horizon region deep in the bulk admits a continuum description, then the horizon modes are parameterizable by a continuum light-cone coordinate $u$ and their correlators approach a continuum function $C_b(u,u')$ which is translation invariant in the continuous sense. This is equivalent to the condition that the bulk is continuous near the horizon. Our numerical observations show that this is approximately true. Note that the volume law in Fig.~\ref{fig: entropy}(c) is obtained by first taking the continuum limit $\Delta t\to 0$ for a \textit{fixed} boundary time interval $[-z,z]$, and then the taking the limit $z \to \infty$. In this continuum limit, the matrix $\Ct_b$ defined by \cref{eq:entropy_0z} approaches a continuum function $\Ct_b(t,t)$ supported on $[-z,z]$. Nominally, this function is not time-translation invariant, as the matrix product and its continuum generalization as an integral is supported only over a finite interval. However, we find that for any $\Delta t$, the matrix $\Ct_b$ gives correlators that track well with the translation invariant correlators for a KMS state (\cref{fig: entropy}(d)).

It turns out the boundary quantity $\Ct_b(t,t')$ is given by a universal expression for thermal states when $t$ is allowed to range over $(-\infty,\infty)$. The continuum functional analog of \cref{eq:entropy_0z} implies that $\Ct_b(t,t')$ is given by a convolution integral, and thus, is a translation invariant function. Accordingly, we can consider its Fourier decomposition $\Ct_b(\omega)$ which satisfies
\begin{align}
    \Ct_b(t-t')&=\int_{-\infty}^{+\infty} \frac{d\omega}{2\pi} \Ct_b(\omega)e^{i\omega(t-t')}\nonumber\\
    \Ct_b(\omega)&=\frac{C(\omega)}{A(\omega)}\label{eq:horizon correlator}
\end{align}
where $A(\omega)$ and $C(\omega)$ are the Fourier transforms of $A(t-t')$ and $C(t-t')$, respectively. The right-hand side of Eq.~\eqref{eq:horizon correlator} is entirely universal for thermal states\footnote{We acknowledge Yingfei Gu for pointing out this to us.} as determined by the KMS condition~\cite{kubo1957statistical,martin1959theory}, and is given by:
\begin{align}
    \frac{C(\omega)}{A(\omega)}=\tanh\frac{\beta\omega}2
\end{align}
Performing the inverse Fourier transform back into the time domain gives $\Ct_b(t-t')$. Identifying this with $C_b(u-u')$ implies that the two point function for the future horizon modes $C_b(u,u')$ in the continuous limit $\Delta t\rightarrow 0$ has the universal form
\begin{align}
\left\langle\left[\psi^f(u),\psi^f(u')\right]\right\rangle=-\frac{i2\pi}{\beta\sinh\frac{\pi(u-u')}{\beta}}\label{eq:horizon_correlator_continuous}
\end{align}
where we recall the definition of a future horizon mode as a limit of a right-moving fermion, $\psi^f(u)=\lim_{z\rightarrow \infty}\psi_R(z,u+z)$. Thus, the frequency $\omega$ on the boundary is the conjugate variable to the horizon coordinate $u$. Accordingly, the bulk horizon modes have occupation number $n(\omega)=\frac{1}{2}\left( \tanh\frac{\beta\omega}2\right)=\left(e^{\beta\omega}+1\right)^{-1}$, which precisely agrees with that of a chiral Majorana fermion of momentum $\omega$. In other words, given an {\it arbitrary} Majorana fermionic GFF on the boundary, so long as the boundary is in thermal equilibrium such that KMS condition holds, a bulk fermion at the horizon is always equivalent to the thermal state of a chiral Majorana fermion with linear dispersion and inverse temperature $\beta$.

This in turn implies that the entropy of a bulk interval can be computed using the conformal field theory formula~\cite{calabrese2009entanglement}:
\begin{align}
    S_A=\frac 1{12}\log\left(\frac{\beta}{\pi}\sinh\frac{\pi z}{\beta}\right)+{\rm const}
\end{align}
(Note that the coefficient is $\frac{c}6$ rather than $\frac{c}3$ because the theory is a chiral Majorana fermion, with $c=\frac12$.) For a region with length $z\gg \beta$ we obtain
\begin{align}
    S_A\simeq \frac{\pi z}{12\beta}+{\rm const.}
\end{align}
Numerically, we observe a linear $z$ dependence in the entropy, but the coefficient does not match the expected value of $\frac{\pi}{12\beta}$. This discrepancy is likely a consequence of the discretization, but a more quantitative explanation is unknown to us. 

\subsection{Geometry and Curvature}\label{sec: bulk_curvature}

We have constructed a bulk Gaussian quantum circuit which possesses the causal structure expected of a local bulk theory. However, one may wonder if it is possible to extract a more refined notion of bulk locality from the circuit representation. 
To address this question, we can assume that in the small $\Delta t$ limit the circuit dynamics is equivalent to a non-interacting Majorana fermion field theory in curved space, and then investigate the bulk geometry based on this assumption.  
While in general, quantum circuits do not inherently carry a notion of geometry, in the (1+1)d case, we can actually \textit{deduce} the bulk geometry by comparing the circuit dynamics to the discretized dynamics of free Majorana fermions on a general (1+1)d spacetime. This comparison is facilitated by the fact that any two-dimensional metric is conformally-equivalent to the flat-space Minkowski metric, which implies that one can always find coordinates $(t,z)$ such that the metric takes on the following form:
\begin{align}
    ds^2 = \Omega^2(t,z) \left(-dt^2 + dz^2\right)\label{eq: general metric}
\end{align}
where $\Omega(t,z)$ is called the conformal factor. It is natural to compare the discretized circuit coordinates $t,z$ with the conformal coordinates in the equation above, since they both have light cones at $t=\pm z$. The conformal form of metric does not completely fix the choice of bulk coordinate system, since we can define $u=t+z$ and $v=t-z$ and any coordinate change $u=f(\Ut),~v=g(\Vt)$ will preserve the conformal form of the metric. However, such a conformal transformation will also change the boundary metric. If the boundary is defined by taking $z\rightarrow 0$ uniformly, and the boundary metric component $g^{00}=1$ is fixed by the boundary dynamics, then there is a unique coordinate choice which corresponds to a unique $\Omega(t,z)$, up to a trivial rescaling $t\rightarrow \lambda t$, $z\rightarrow \lambda z$, $\Omega\rightarrow \lambda^{-1}\Omega$. More details of the relation between bulk and boundary metric and the comparison with quantum circuit at $z\rightarrow 0$ are discussed in Appendix \ref{app: circuit_geo}. As an example, the AdS Rindler wedge is expected to be the dual description of the $(0+1)$d thermal state at low temperatures has a metric with conformal factor
\begin{align}
    \Omega(t,z)=\frac{2\pi}{\beta\sinh\frac{2\pi z}{\beta}}
\end{align}
where we have defined the coordinates such that the time periodicity is $\beta$ when the metric is Wick rotated to imaginary time. 

Free Majorana fermion dynamics on a fixed curved background with the metric in \eqref{eq: general metric} is described by the action
\begin{align}
    \mathcal{A}=\int dtdz\left[\psi_R\left(i\partial_t+i\partial_z\right)\psi_R+\psi_L\left(i\partial_t-i\partial_z\right)\psi_L-2im\Omega(t,z)\psi_L\psi_R\right]
\end{align}
Note that the conformal factor $\Omega(t,z)$ only appears in the mass term, since the kinetic energy terms are simply the action of a massless Majorana fermion, which is conformally invariant. Canonical quantization leads to the following Hamiltonian
\begin{align}
    H(t)=\int dz\left[\psi_R\left(-i\partial_z\right)\psi_R-\psi_L \left(-i\partial_z\right)\psi_L+2im\Omega(t,z)\psi_L\psi_R\right]\label{eq: Majorana_Hamiltonian}
\end{align}
where $\psi_{L,R}$ satisfies the canonical anticommutation relation $\left\{\psi_a(z,t),\psi_b(z',t)\right\}=\delta_{ab}\delta\left(z-z'\right)$.

\begin{figure}[t]
    \centering
\includegraphics[width=0.5\textwidth]{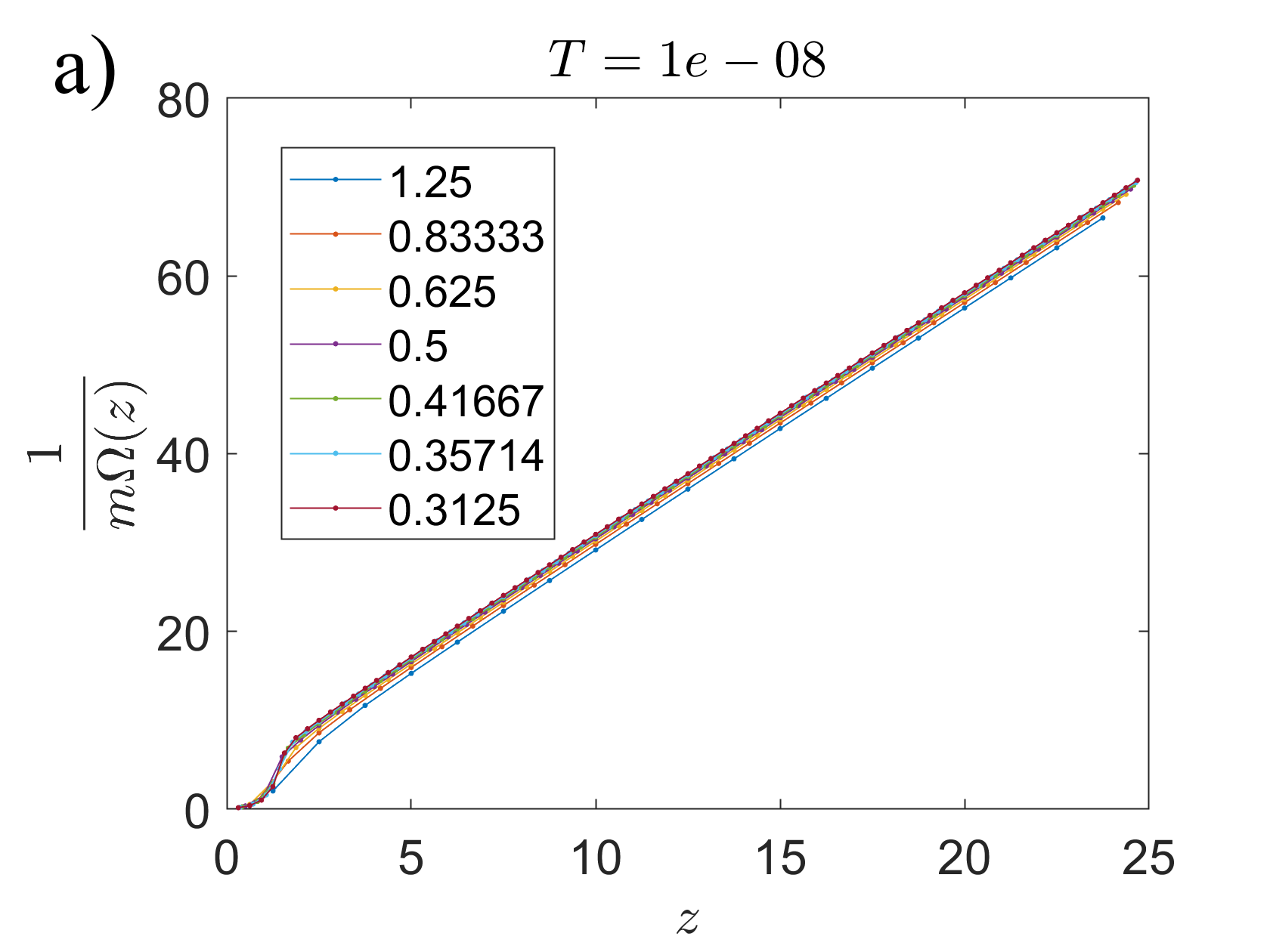}\includegraphics[width=0.5\textwidth]{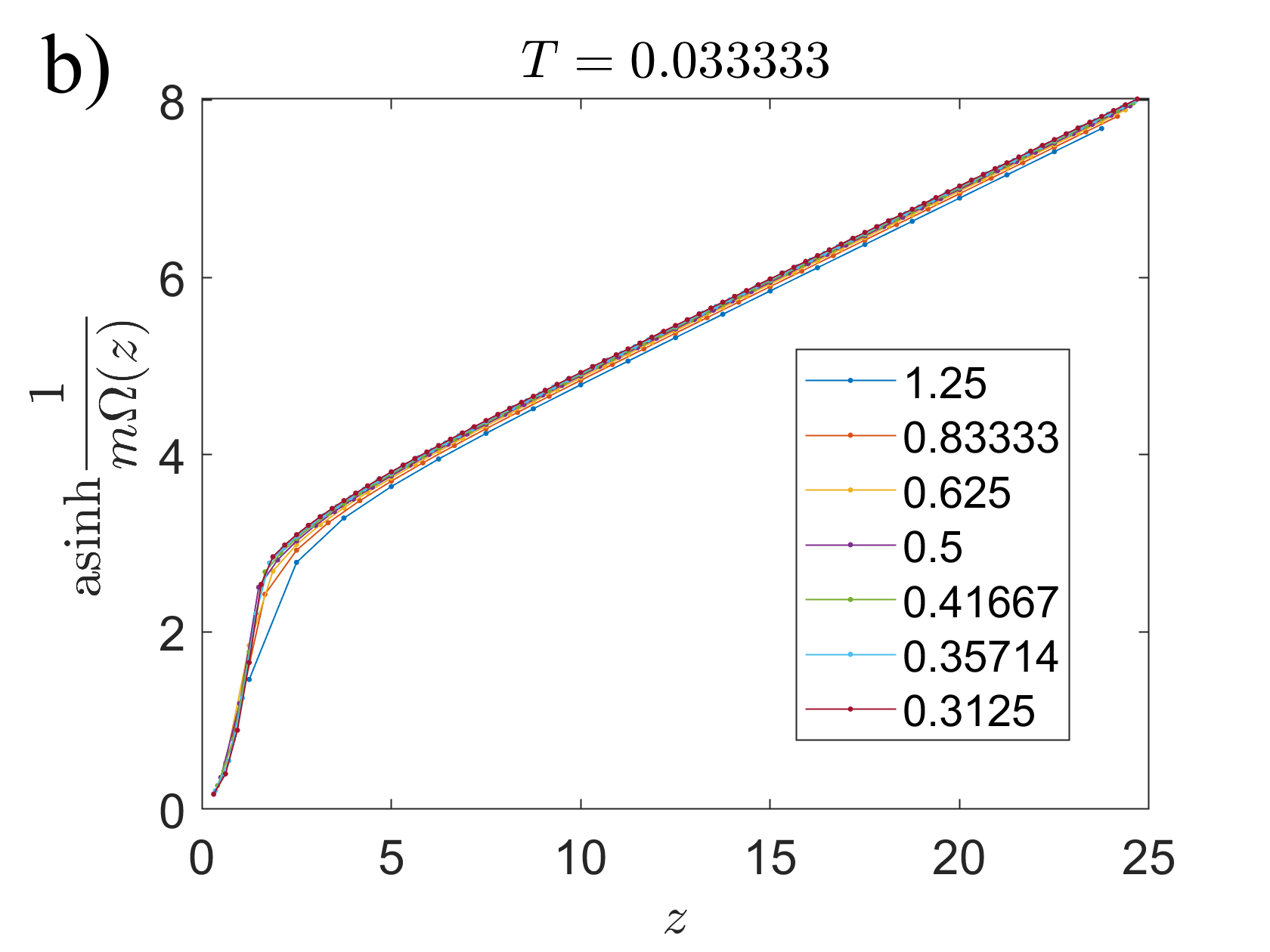}\\
\centering
\includegraphics[width=0.5\textwidth]{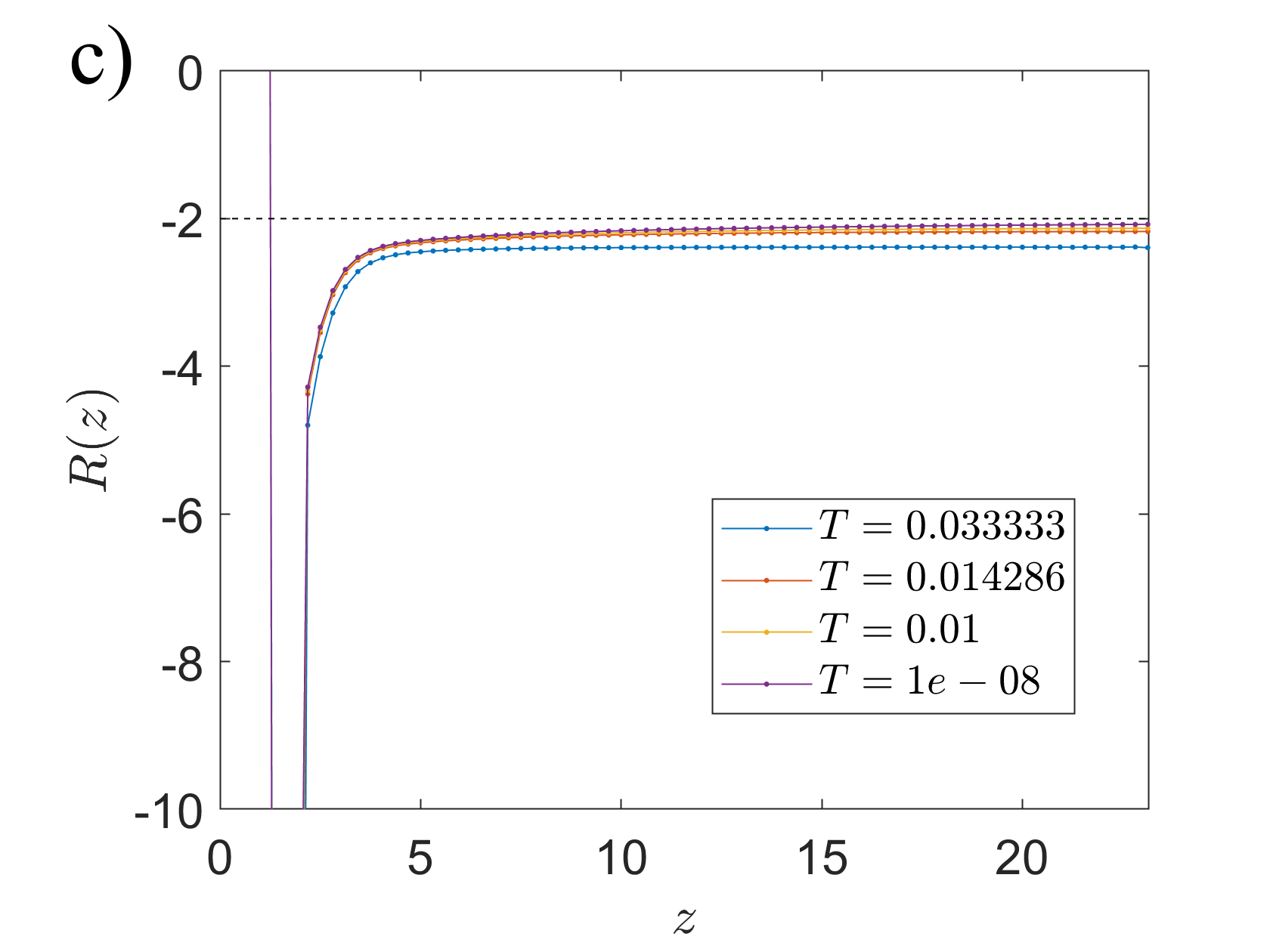}\includegraphics[width=0.5\textwidth]{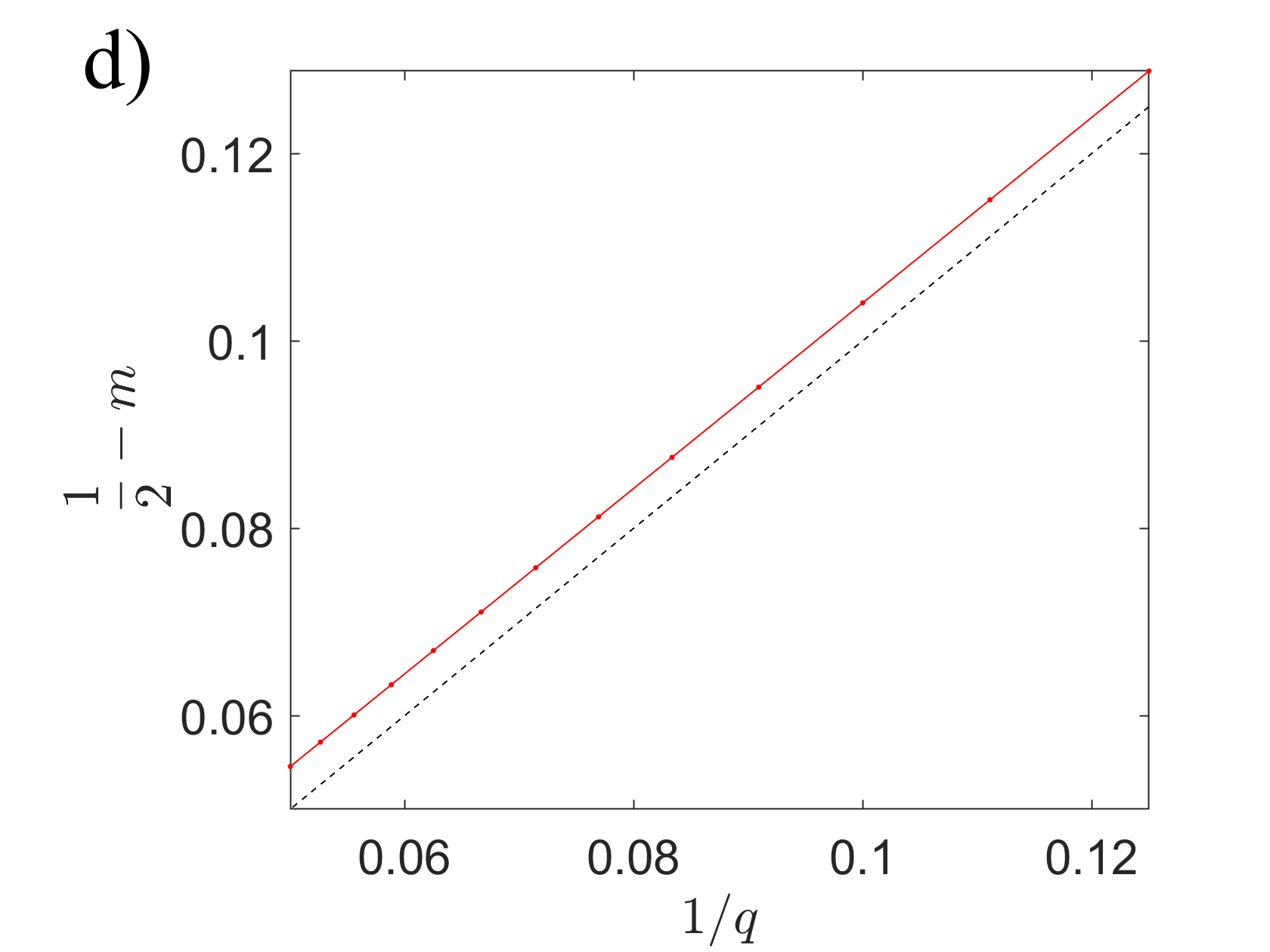}
    \caption{(a) and (b) Mass parameter as a function of $z$ for zero temperature and finite temperature, respectively. The value in the legend is $\Delta t$. Note that we plotted ${\rm asinh}\frac1{m(z)}$ at finite temperature and $\frac1{m(z)}$ at zero temperature. (c) Ricci curvature $R(z)$ as a function of $z$, computed using Eqs.~(\ref{eq: mOmegaz}) and \ref{eq: Ricci}), for different temperature. (d) $\frac 12-m$ from zero temperature fitting (Eq.~\eqref{eq:mass zero T}) as a function of $1/q$. The dashed line is the theoretical value $m=\frac12-\frac1q$. The computation in (a)-(c) are done for $q=10$. }
    \label{fig:cont_limit}
\end{figure}

We would like to compare the fermion dynamics determined by Eq.~(\ref{eq: Majorana_Hamiltonian}) with the quantum circuit in small $\Delta t$ limit. The right-movers and left-movers in the quantum circuit correspond to the continuous fermions $\psi_R$ and $\psi_L$, respectively, since they move with speed of light if they are decoupled from each other. The coupling between left-movers and right-movers are determined by $m\Omega(t,z)$ in the continuous theory, and by the gate angle $\theta_U$ and $\theta_V$ in the quantum circuit. Therefore, for a given quantum circuit, we can determine the value of $m\Omega(t,z)$ that provides an optimal approximation to the dynamics. Here, we present the result with a simplified argument, leaving the precise derivation to Appendix. \ref{app: circuit_geo}. In the long wavelength limit, the transition amplitude from a right-mover to a left-mover in time $\Delta t$, {\it i.e.} after applying one layer of $U$ gates \textit{and} one layer of $V$ gates, is given by $\sin(\theta_U+\theta_V)$. By comparison, in the continuous theory this amplitude is $m\Omega(z)\Delta t$. Therefore, we obtain
\begin{align}
    m\Omega(z)=\lim_{\Delta t\rightarrow 0}\frac{\pi-\theta_U(z)-\theta_V(z)}{\Delta t}\label{eq: mOmegaz}
\end{align}
This result assumes that $\Omega(t,z)$ is smooth over a scale much larger than lattice constant $\Delta t$, which means this equation does not apply to $z$ very close to the boundary, where $\theta_U,\theta_V$ change rapidly. Numerically, we confirm that $\theta_U+\theta_V$ depends linearly on $\Delta t$ so that the $\Delta t\rightarrow 0$ limit is well-defined. 

The resulting $m\Omega(z)$ for zero temperature and finite temperature are shown in Fig. \ref{fig:cont_limit}(a) and (b), respectively. At zero temperature, $m\Omega(z)$ has the following behavior (except for a region near the boundary):
\begin{align}
    m\Omega(z)=\frac{m}{z-z_0}\label{eq:mass zero T}
\end{align}
In AdS/CFT, we expect $m=\frac12-\Delta=\frac12-\frac1q$. A numerical fit for the coefficient of the $z^{-1}$ decay in the expression above yields $m\simeq 0.495-\frac1q$ (Fig. \ref{fig:cont_limit}(d)), which agrees with the analytic result except a small offset $\sim 0.005$. At finite temperature, $m\Omega(z)$ fits well with the function $m\Omega(z)=\frac1{\sinh\left(\alpha z+\gamma\right)}$ which qualitatively agrees with the expected AdS behavior.

To make a more quantitative comparison with the AdS geometry, we can compute the Ricci scalar of $\Omega(z)$, which for the conformal coordinate is given by
\begin{align}
    R=\frac{2}{\Omega^2(t,z)}\left(\partial_t^2-\partial_z^2\right)\log\Omega\label{eq: Ricci}
\end{align}
The result is shown in Fig. \ref{fig:cont_limit}(c). As expected, the curvature approachs $-2$ in the bulk. Interestingly, the curvature appears to diverge near the boundary, though as discussed above, the comparison fails when $\theta_U+\theta_V$ is not a smooth function of $z$. Nevertheless, the curvature divergence near the boundary is expected to be a consequence of UV physics at the boundary.\footnote{A similar kind of curvature divergence occurs if we take the Eulidean two point function of large-$q$ SYK model and interpret it as decaying exponentially with geodesic distance on a disk with the metric $ds^2=\Omega(z)^2(d\tau^2+dz^2)$.\cite{lensky2023}}

\section{Two-sided Bulk Reconstruction for Coupled SYK models}\label{sec: coupled_syk}

In the previous section, we conducted bulk reconstruction of the thermal state of SYK model using our generalized protocol. The resulting bulk dual possessed only one asymptotic boundary. In this section, we will apply our reconstruction scheme to two-sided models, \textit{i.e.} boundary models dual to bulk geometries with two boundaries, including the thermofield double state of the SYK model and the coupled SYK model. In the limit of low temperature/small bilinear coupling, these models are dual to JT gravity theory with a modified dilaton equation of motion. Using our method, we can study the bulk dual without relying on the low temperature/coupling limit.

\subsection{Thermofield double state and the coupled SYK model}\label{subsec: TFD_coupledSYK}
Here, we review some key facts about the thermofield double formalism and the coupled SYK model.

Given a generic quantum many-body system with Hamiltonian $H$, we can define a new system comprised of two identical copies of the original system with a Hilbert space given by the tensor product $\SH=\SH_1\otimes \SH_2$ where the subscripts label the copy. Let $H_1=H\otimes \Id_2$ and $H_2=\Id_1\otimes H$. The thermofield double state (TFD) is a particular entangled state in $\SH$ of the joint system that is invariant under $H_1-H_2$:
\begin{equation}
    \begin{aligned}
        \ket{TFD} &= \frac{1}{\sqrt{Z}}\sum_n e^{-\beta E_n/2} \ket{n}_1\otimes \ket{n}_2 \qquad Z \equiv \Tr(e^{-\beta H})
    \end{aligned}
\end{equation}
where $E_n$ are the energy eigenvalues of $H$, and the trace in $Z$ is defined for a single copy of the original system. Note that the reduced density matrix for either copy is the thermal state. Hence, the TFD state may be regarded as a purification of the thermal state density operator $\rho=Z^{-1}e^{-\beta H}$. Note that a unitary change of basis of one copy in the tensor product Hilbert space does not affect any physical properties, so long as one considers the covariantly transformed operators. For example, one can pick a unitary operator $U_2$ in system $2$ and redefine $\left|TFD\right\rangle\rightarrow U_2\left|TFD\right\rangle$, $H_2\rightarrow U_2H_2U_2^\dagger$, and physical properties will be unaffected as long as all operators in system $2$ are conjugated by the same unitary. By construction, $\left|TFD\right\rangle$ is invariant under the decoupled time evolution $H_1-H_2$, with $\left(H_1-H_2\right)|TFD\rangle=0$.  
For the SYK model, a convenient choice of the TFD state is $H_1=(-1)^{q/2}H_2$, which is chosen such that the infinite temperature state $\ket{I}\equiv 2^{-N/4}\sum_n\ket{n}_1\otimes U_2\ket{n}_2$ is a Gaussian state satisfying $\left(\chi_{i1}+i\chi_{i2}\right)\ket{I}=0$. This condition is essential for ensuring that the fermion correlation functions in $\ket{TFD}$ satisfy the GFF condition under time evolution given by $H_1+H_2$ where $H$ is the SYK Hamiltonian.

In AdS/CFT, the TFD state of two copies of a $(d+1)$-dimensional CFT admits an interpretation as a maximally-extended asymptotically-AdS black hole in $(d+2)$ dimensions ~\cite{maldacena2003eternal,maldacena2013cool}. The boost symmetry of the maximally extended geometry corresponds to the symmetry of the TFD state under $H_1-H_2$. The two exterior regions of the two-sided black hole geometry are causally disconnected, which corresponds to the fact that the two boundaries evolve under independent Hamiltonians $H_1$ and $H_2$, and the joint boundary Hilbert space factorizes into that of the two systems. Based on this correspondence, the TFD state of the low temperature SYK model, with decoupled time evolution of the two copies, corresponds to a two-sided eternal black hole geometry in JT gravity, which in 2d is also the same as an AdS-Rindler geometry. Using our construction, we can analyze the dual description of the TFD state of the large-$q$ SYK model for generic temperatures, beyond the applicability of the JT gravity dual description. Importantly, for the TFD state, we regard both copies of the SYK model as physical, encoding the physics of the two-sided geometry. We can also consider a setup where the two copies of the SYK model are coupled, which will modify the geometric interpretation above, and which we now discuss.


There are various generalizations of SYK model~\cite{gu2017local,bi2017instability,chowdhury2022sachdev, banerjee2017solvable,fu2017supersymmetric,li2017supersymmetric,maldacena2018eternal} that are solvable using the standard large-$N$ solution and associated Schwinger-Dyson equation. Our approach applies to all such generalizations as long as the fermionic degrees of freedom in the model are still GFFs. The generalization that we will consider is the coupled SYK model with a bilinear coupling term~\cite{maldacena2018eternal,lensky2021rescuing}:
\begin{align}
    H=H_{1}+H_{2}+i\mu(t)\sum_{j}\chi_{j1}\chi_{j2}\label{eq: Ham_coupledSYK}
\end{align}
For constant $\mu$, Ref. \cite{maldacena2018eternal} shows that for small $\mu/\mathcal{J}$, the ground state of the above model is dual to a global ${\rm AdS}_2$ geometry, with two boundaries that are causally \textit{connected}. This geometry is a two-dimensional version of a traversable wormhole. Interestingly, in the large-$q$ limit, the ground state of this model is identical to TFD state with a particular temperature $\beta(\mu)$, but undergoing different dynamics from the usual invariant TFD state due to the coupling term in the Hamiltonian. 

Defining $\mu(t)=\hat{\mu}(t)/q$ and taking $q\rightarrow \infty$ for fixed $\hat{\mu}(t)$ yields a simplified Schwinger-Dyson equation, which allows one to obtain an analytic solution for the two point functions in the coupled model. Ref.~\cite{lensky2021rescuing} obtains explicit formulas for the two point functions for a generic $\hat{\mu}(t)$, with the condition that the initial state of time evolution is a TFD state. (One can equivalently view the condition on the initial state as the condition that $\hat{\mu}$ is constant before the initial time $t=0$ so that the evolution starts from the ground state of a certain $\hat{\mu}$.) In the following, we review key formulas for the two point functions in the coupled SYK model without derivation. More details can be found in Ref. \cite{lensky2021rescuing}. 

The two point functions in the coupled model may be expressed as
\begin{align}
G_{11}(t,t')&=\left(-\frac{\zeta'(t)\zeta'(t')^*}{\sin^2\frac{\zeta(t)-\zeta^*(t')}2}\right)^{1/q}\label{eq: G11_general}\\
G_{12}(t,t')&=i\left(\frac{\zeta'(t)\zeta'(t')^*}{\cos^2\frac{\zeta(t)-\zeta^*(t')}2}\right)^{1/q}\label{eq: G12_general}
\end{align}
where $\zeta(t)$ is a complex function defined by
\begin{align}
    \zeta'(t)=\frac{e^{ip(t)}}{\sqrt{e^{2\phi(t)}-1}}
\end{align}
The variables $p(t)$ and $\phi(t)$ behave like the canonical momentum and position in classical Hamilton dynamics with the Hamiltonian\footnote{It's worth pointing out that the quantum version of this Hamiltonian can also be derived from the double scaling limit of the SYK model~\cite{lin2022bulk}.}
\begin{align}
H=-\cos p\sqrt{1-e^{-2\phi}}+\mu\phi
\end{align}
$\zeta(t)$ may be related to $\hat{\mu}(t)$ by the Hamilton equations of motion
\begin{align}
    \dot{p}=-\frac{\partial H}{\partial\phi}=-\hat{\mu}+\frac{e^{-2\phi}}{\sqrt{1-e^{-2\phi}}}\cos p\label{eq: EoM_coupledSYK}
\end{align}
As an example, the thermofield double state corresponds to $\hat{\mu}=0$ and
\begin{align}
    \zeta(t)=2{\rm atan}\left[\tanh\left(\frac{\pi \nu}{2\beta}\left(t-\frac{\beta}2\right)-i\frac\pi 4\right)\right]\label{eq:zetat_TFD}
\end{align}
with $\nu$ defined in Eq.~\eqref{eq: G_largeq}. 
The eternal traversable wormhole solution of Ref. \cite{maldacena2018eternal} corresponds to a linear function
\begin{align}
    \zeta(t)&=V_Gt-i\delta\label{eq:parameter_ETW1}\\
    \delta&={\rm atanh}e^{-\phi_0},~V_G=\left(e^{2\phi_0}-1\right)^{-1/2}\nonumber\\
    e^{-2\phi_0}&=\frac{2\hat{\mu}}{\sqrt{\hat{\mu}^2+4}+\hat{\mu}}\label{eq:parameter_ETW2}
\end{align}

In the low temperature limit, the coupled SYK system with a time-dependent coupling $\mu(t)$ corresponds to JT gravity with a modified dilaton profile, 
induced by sending in matter with energy-momentum from the boundary. The sign of $\mu(t)$ dictates whether the matter's energy density is positive or negative, which correspondingly suppresses or enhances the correlation between the two boundaries. However, it should be noted that the large-$q$ solution applies for arbitrary temperature/energy, though the gravitational interpretation is only known for the low temperature limit.

\subsection{Bulk dual of the thermofield double state}
\label{subsec: ETW}

\begin{figure}
    \centering
\includegraphics[width=0.46\textwidth]{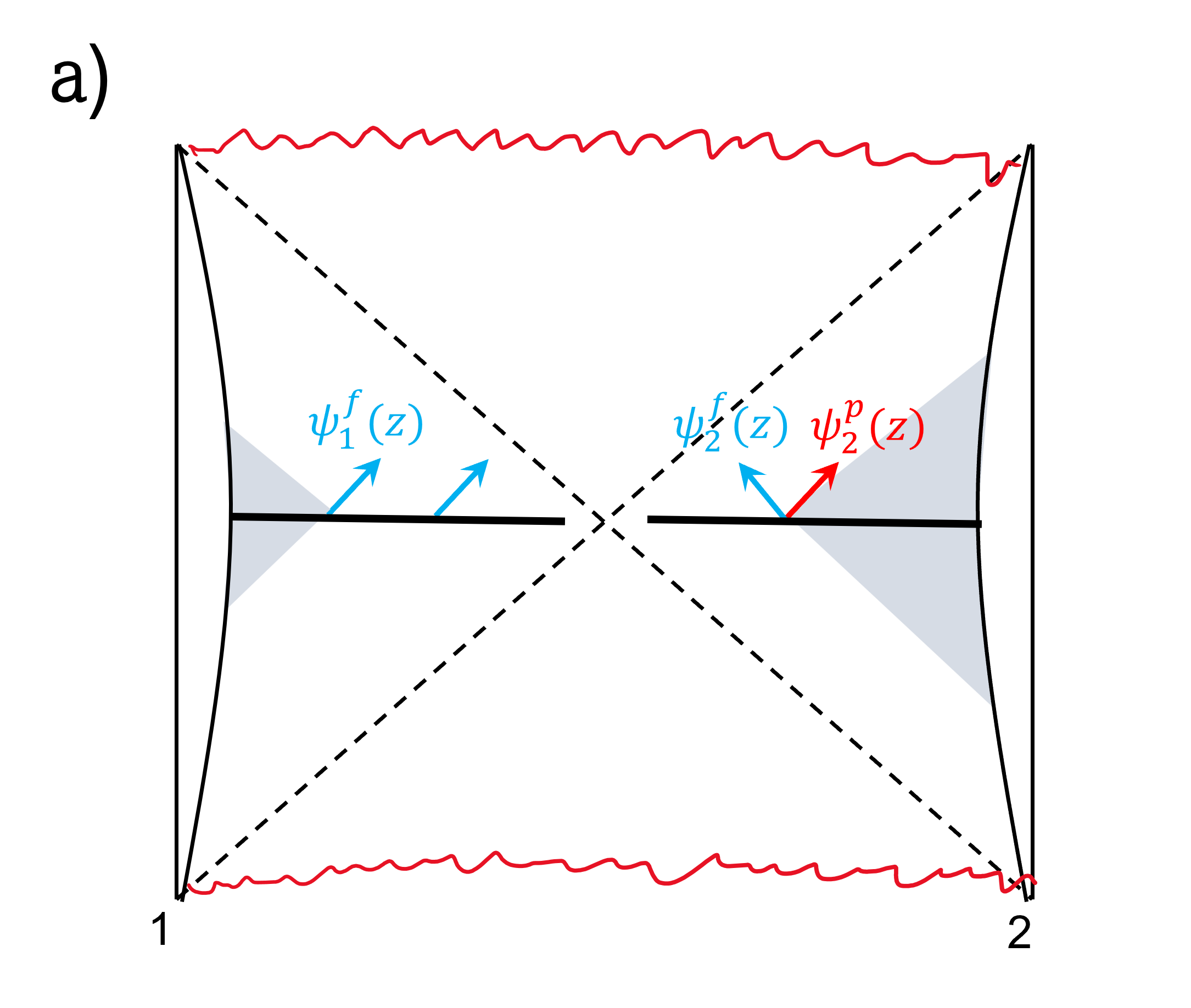}\includegraphics[width=0.5\textwidth]{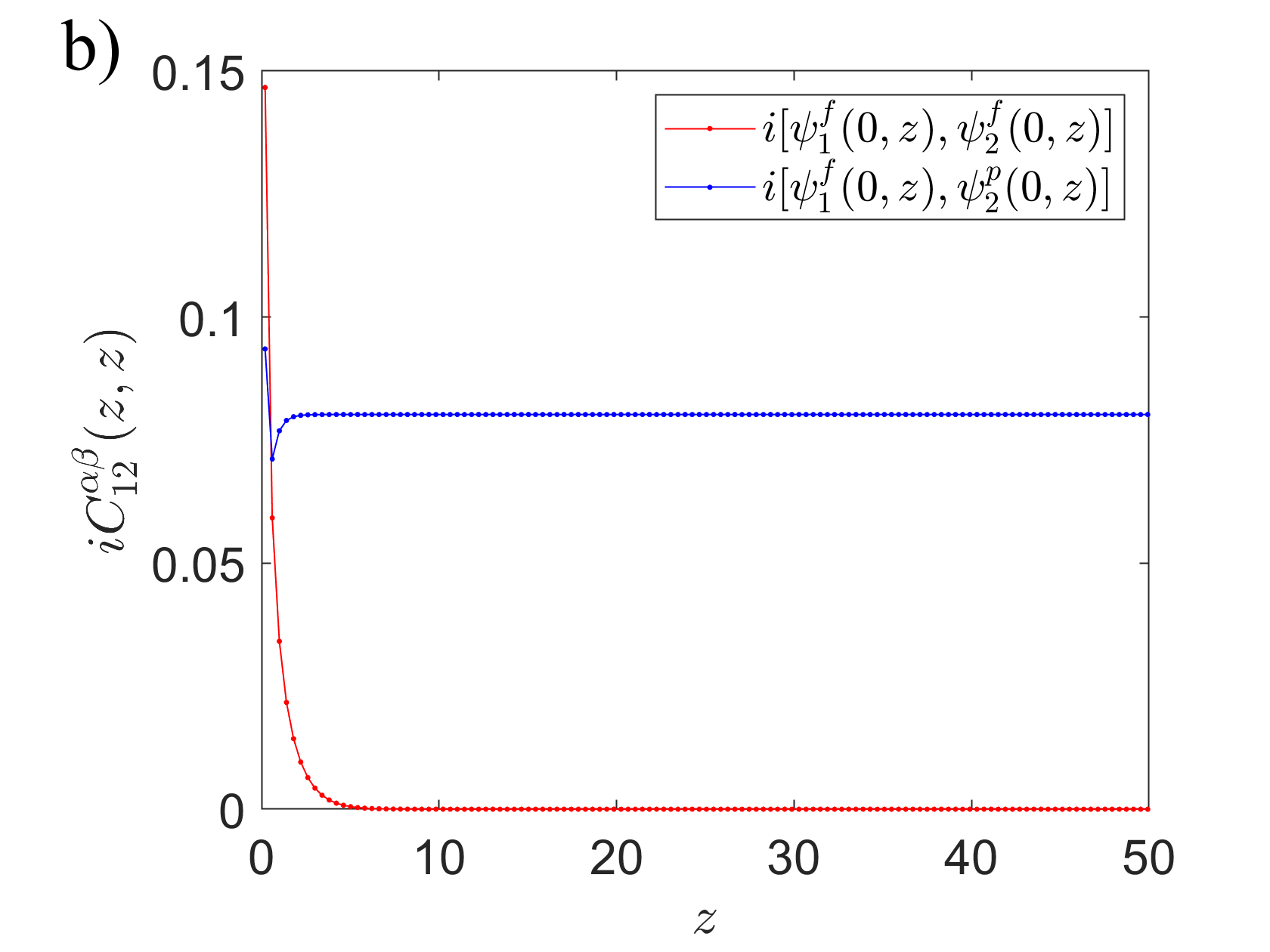}\\
    \centering
\includegraphics[width=0.5\textwidth]{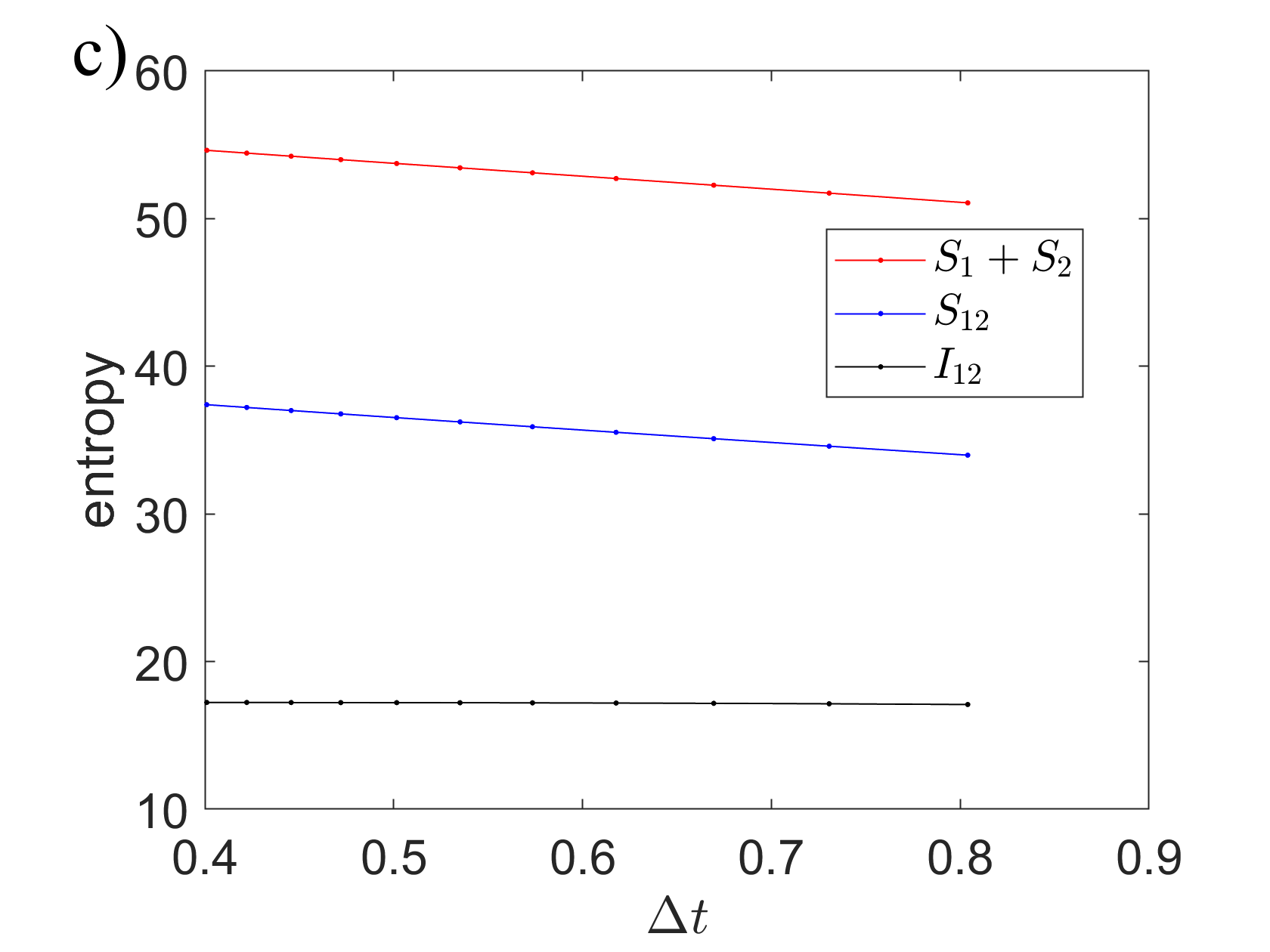}\includegraphics[width=0.5\textwidth]{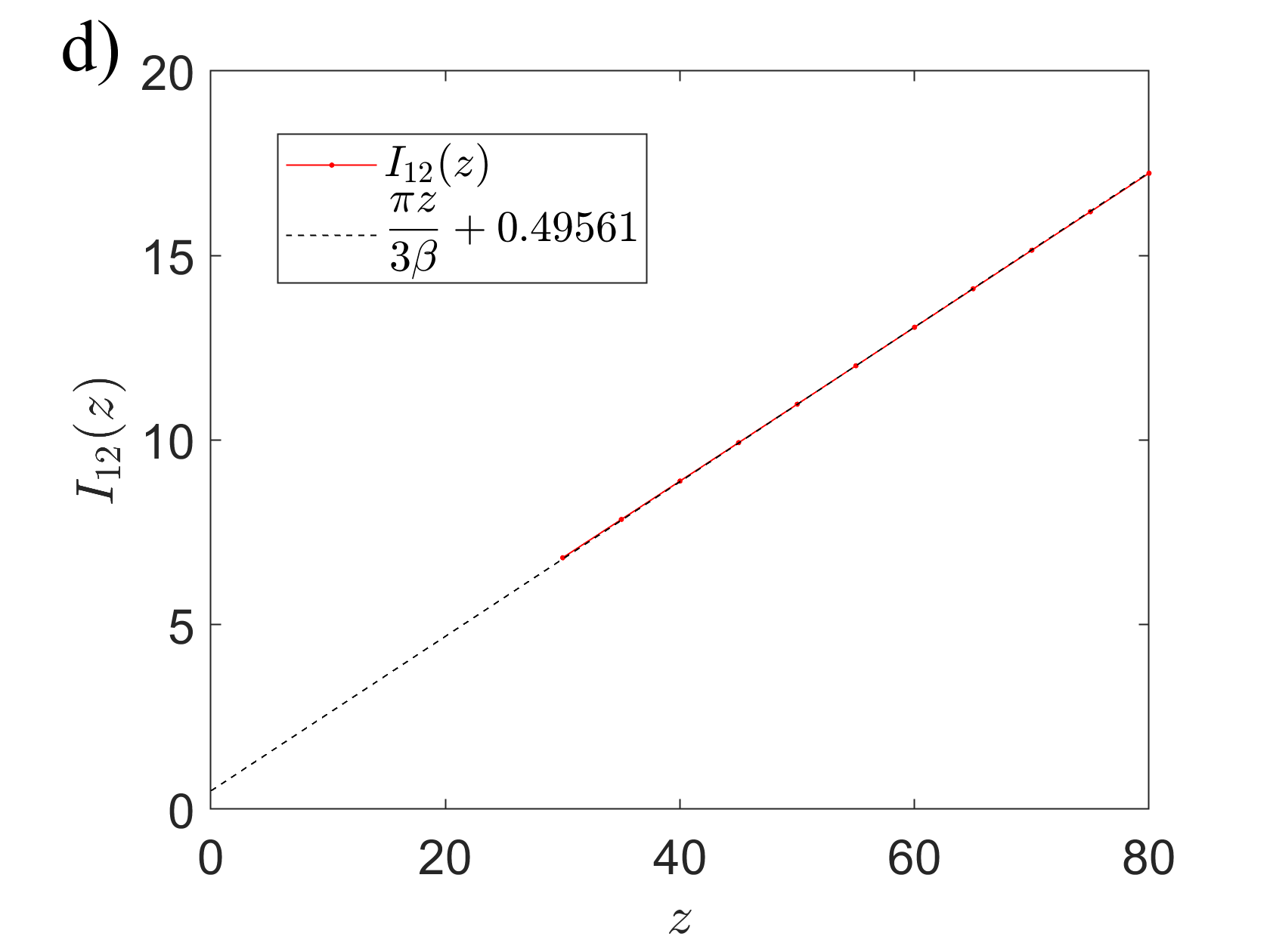}
    \caption{(a) Illustration of the bulk fermions in the eternal black hole geometry. (b) two point correlation function between the two exterior regions. $\psi^f$ and $\psi^p$ correspond to the blue and red arrows in subfigure (a) respectively. (c) Entropy and mutual information as a function $\Delta t$, for a fixed interval $[0,z]$ with $z=80$. (d) Mutual information $I_{12}$ as a function of $z$ (red line with dots). The dashed line is a fitting with the slope $\frac{\pi}{3\beta}$. The calculation is done with $\beta=5$.}
    \label{fig:TFD_entropy}
\end{figure}

In this subsection, we apply our bulk construction to the TFD state undergoing decoupled time-evolution $H=H_{\rm SYK_1}+H_{\rm SYK_2}$. Unlike the previous investigation of the SYK thermal state, we will be considering two-sided correlators, and the constructed bulk operators will be generically be supported on two boundaries. The two point function matrix can be constructed as before, but now carries an additional "flavor" index for the copy $s,s'=1,2$. Due to large-$N$ factorization, we again drop the Majorana species index $j=1,\dots,N$. The corresponding anticommutator matrix has the form
\begin{align}
    A_{ss'}(t,t')=\left\langle TFD\right|\left\{\chi_s(t),\chi_{s'}(t')\right\}\left|TFD\right\rangle=\delta_{ss'}A(t-t')
\end{align}
where the second equality is due to the fact that the two copies are decoupled. However, the TFD state of the two copies has nontrivial entanglement between the two copies so that the commutator matrix
\begin{align}
    C_{ss'}(t,t')=\left\langle TFD\right|\left[\chi_s(t),\chi_{s'}(t')\right]\left|TFD\right\rangle
\end{align}
has nontrivial off-diagonal elements $C_{12}(t,t')$. Since $A$ is diagonal in the copy index, the kernel that determines the bulk fermions is 
decoupled for the two sides, so bulk fermions in the causal wedge of one boundary can be written using the same kernel used for the thermal state with only one boundary. Applying the kernel to boundary fermions $\chi_{s}(t),~s=1,2$, one defines bulk fermions $\psi^s_{L,R}$, where the superscript indicates whether the bulk fermion resides in the causal wedge of the first or second boundary. In principle, we can choose different discretizations of time for the two copies, which corresponds to different discretizations of the two bulk boundaries, but for simplicity we will always choose the same $\Delta t$ for the two boundaries.

Upon obtaining the $\psi_{L,R}^s$, we can study the quantum entanglement between the two bulk causal wedges. The TFD state is invariant under time evolution by $H_1-H_2$, which leads to the symmetry $C_{12}(t,t')=C_{12}(t+\tau,t'-\tau)$, $C_{ss}(t,t')=C_{ss}(t+\tau,t'+\tau)$, $s=1,2$. As a consequence of this symmetry and the entanglement of the TFD state, the bulk fermion $\psi_{1R}(0,z)$ is strongly correlated with $\psi^2_{L}(0,z)$ which is independent of $z$. (Fig. \ref{fig:TFD_entropy}(a) and (b)) It should be noted that here $L,R$ are defined with respect to the causally-connected boundary, so the orientation for wedge $2$ is opposite from that in the Penrose diagram of an eternal black hole. We will often use the labels $f,p$ to clarify the choice of orientation. The $z$-independent coupling leads to a large mutual information between the two systems. Fig. \ref{fig:TFD_entropy}(c) shows the entropy and mutual information of the two regions in bulk wedges $1$ and $2$ from the boundary to bulk point $(0,z)$ as a function of $\Delta t$. (Note that the coordinate starts from $z=0$ from both boundaries, and $z$ increases towards the horizon.) While the entropy has a nearly linear dependence on $\Delta t$ (like the single-sided case), the mutual information is essentially independent of $\Delta t$. This is because mutual information is a UV finite quantity. When $z$ is varied, we observe a linear dependence of the mutual information in $z$ (Fig.~\ref{fig:TFD_entropy})(d). In the bulk interpretation, this unbounded growth of mutual information is consistent with the fact that the distance between the endpoints of the two causal wedges decreases exponentially. In particular, for AdS, the distance $d(z)$ is given by
\begin{align}
d(z)=2\int_{z}^{+\infty}\Omega(z)dz\simeq 4e^{-2\pi z/\beta}
\end{align}
If the bulk fermions are described by a Majorana fermion CFT in the UV, we expect the mutual information to have logarithmic divergence:
\begin{align}
I_{12}(z)\simeq \frac{c}{3}\log d(z)\simeq \frac{\pi z}{3\beta}
\end{align}
This agrees well with our numerics, as is shown in Fig. \ref{fig:TFD_entropy}(d).

\subsection{Shockwave geometry and traversable wormhole}
We now investigate the coupled SYK model, and study the effect of the coupling on the bulk fermion state and dynamics. First, we consider the simple case of a $\delta$-function shockwave. The Hamiltonian is Eq.~\eqref{eq: Ham_coupledSYK} with $\mu(t)=\frac{\nu}q\delta(t)$. In the large-$q$ limit, the effect of the $\delta$-function shock is to induce a first-order discontinuity in the complex function $\zeta(t)$ in the expressions for the two point functions. For simplicity, we choose a symmetrically-applied shock, as illustrated in Fig. \ref{fig:shockwave}(a). The shockwave-altered $\zeta(t)$ is obtained by patching the TFD solution \eqref{eq:zetat_TFD} for time $t<-t_0$ and $t>t_0$ after a translation along the real axis. Physically, this situation corresponds to a time-evolved TFD state $\ket{TFD(-t_0)}$ which is "kicked" by the shock to $\ket{TFD(+t_0)}$. This is a fine tuned situation since $t_0$ and $\hat{\nu}$ have to be related in order for the final state to be $\ket{TFD(+t_0)}$, but this choice does not qualitatively affect the physics we are interested in.

Given this function $\zeta(t)$, the boundary two point function $G_{ab}(t_1,t_2)$ is determined in Eq.~(\ref{eq: G11_general}) and (\ref{eq: G12_general}). As a side remark, for $t_1>0,~t_2<0$, the derivative 
\begin{align}
    \left.\frac{\partial G_{ss'}(t_1,t_2)}{\partial\nu}\right|_{\nu=0}=-\bra{TFD}\left[\chi_{i1}(0)\chi_{i2}(0),\chi_{js}(t_1)\right]\chi_{js'}(t_2)\ket{TFD}
\end{align} 
is the out-of-time-order correlation function (OTOC)~\cite{larkin1969quasiclassical,shenker2014black,shenker2014multiple,kitaev2014hidden}. For SYK models, only the early time behavior of this four point function can be computed, which grows exponentially with $t$ if we take $t_1=-t_2=t$. For the standard large-$q$ SYK model, $\partial_\nu G_{ss'}(t_1,t_2)$ can be computed exactly for finite $\nu$. 

From $G_{ss'}(t_1,t_2)$, we can construct the bulk fermions. The largest correlation between the two boundaries occurs for time $t$ on boundary $1$ and time $-t$ on boundary $2$, which implies high correlation between the future-directed bulk fermion $\psi^f_1(z)$ (blue arrows in the left wedge of Fig. \ref{fig:shockwave}(d)) and the past-directed bulk fermion $\psi^p_2(z)$ (red arrows in the right wedge of the same figure). As is shown in Fig. \ref{fig:shockwave}(b), this two point function $iC_{12}^{fp}(z,z)=\left\langle i\left[\psi^f_1(z),\psi^p_2(z)\right]\right\rangle$ saturates to a finite value for the TFD state $\nu=0$ but exponentially decreases for $\nu\neq 0$. This illustrates how the correlation between the two wedges is disrupted by the shockwave, and the effect is significant for any finite $\nu$ as a consequence of chaos in the system. 

We can study the correlation between the two boundaries from the bulk perspective using the mutual information between the two bulk intervals $[0,z]$, as illustrated in Fig. \ref{fig:shockwave}(c). In contrast to the linearly-growing mutual information for the TFD in the previous subsection, in the presence of a shockwave, the mutual information (red curve) saturates to a finite value. This is consistent with the gravitational picture that the two bulk exterior regions become separated by a finite distance, rather than sharing a bifurcation horizon. The key finding here is that our framework makes the geometric concepts such as distance between horizons well-defined, even for the finite temperature SYK model where the duality with JT gravity does not apply.

\begin{figure}[h]
    \centering
    \includegraphics[width=0.5\textwidth]{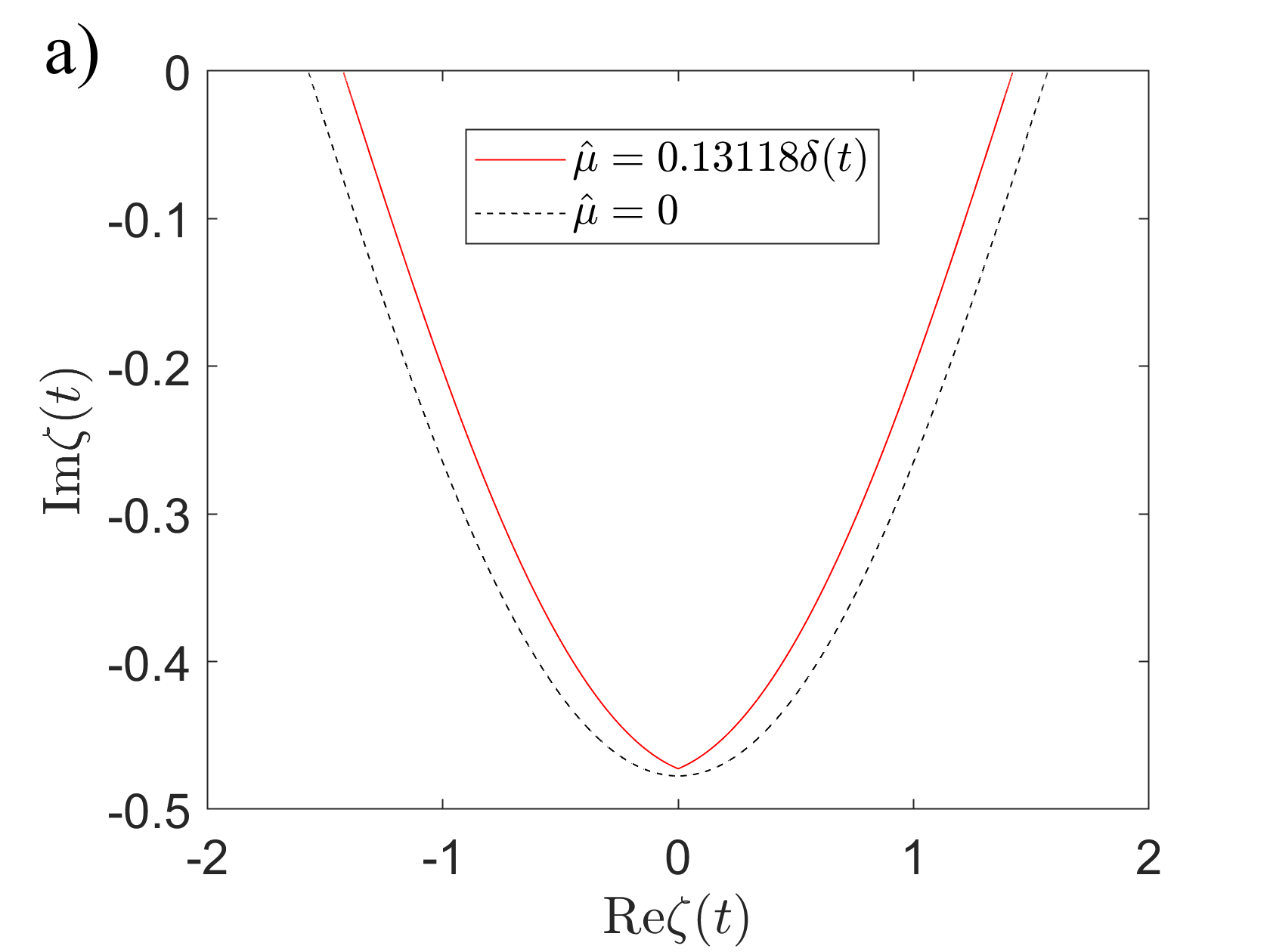}\includegraphics[width=0.5\textwidth]{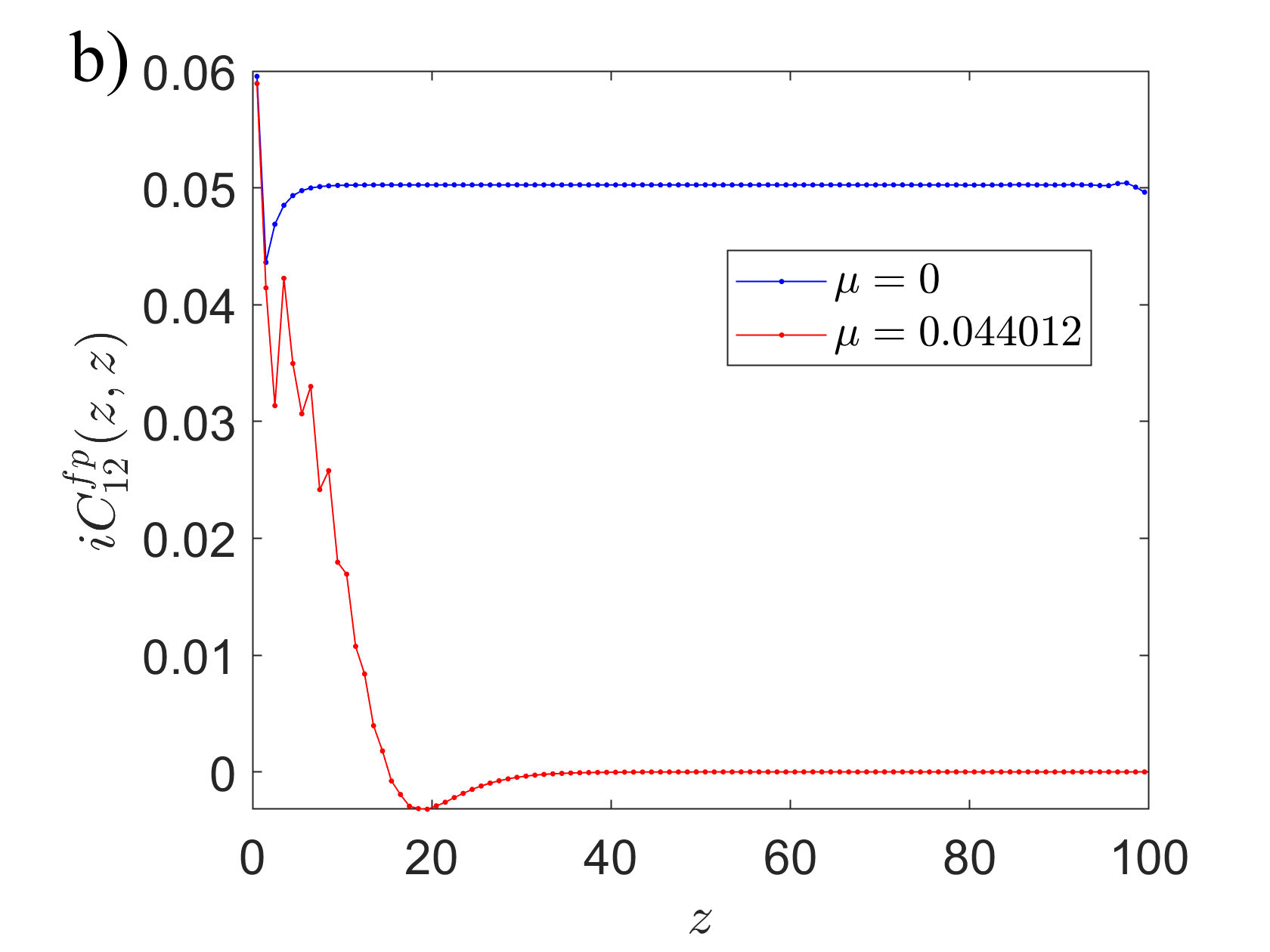}\\
    \centering
    \includegraphics[width=0.5\textwidth]{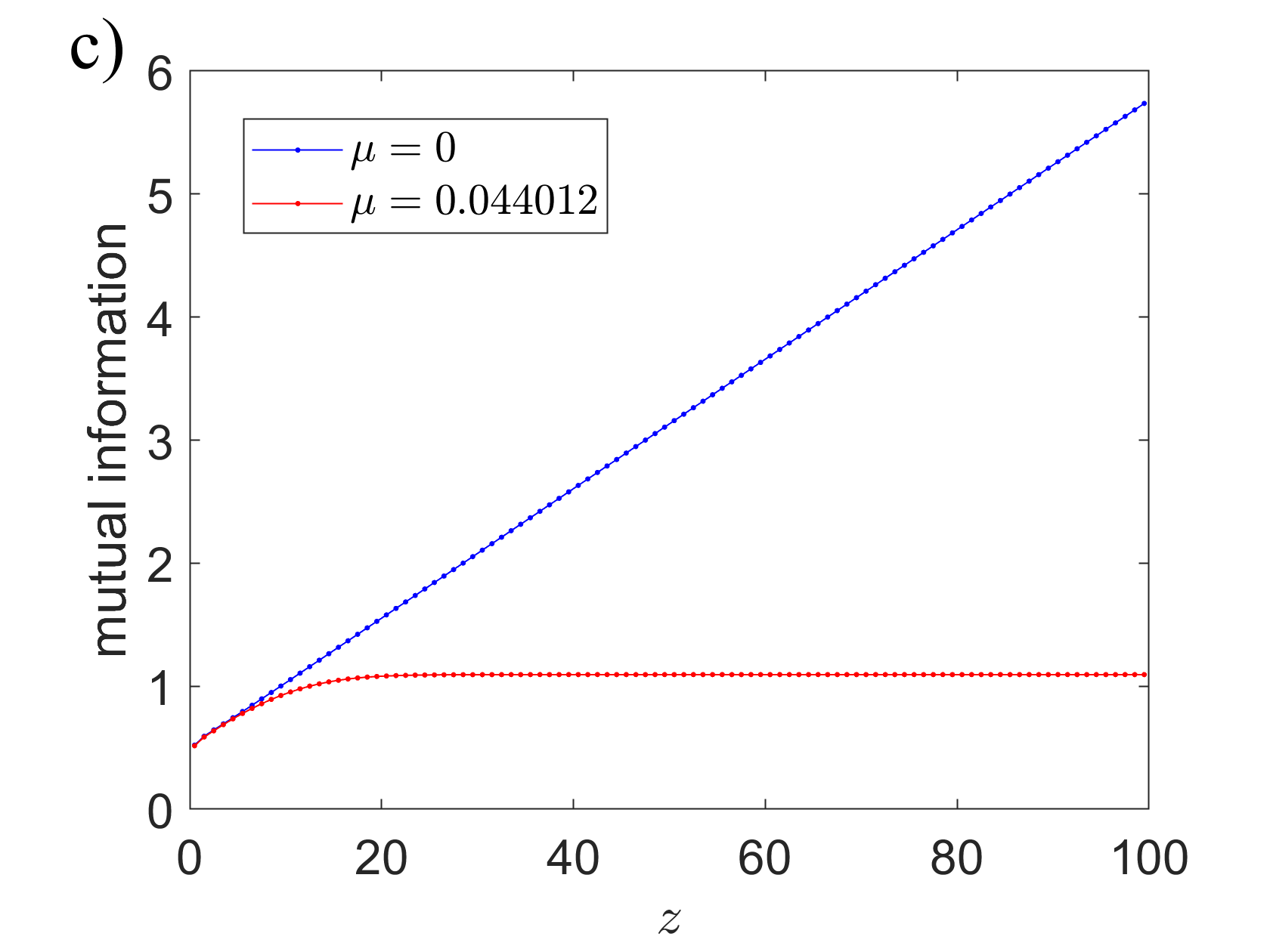}\includegraphics[width=0.48\textwidth]{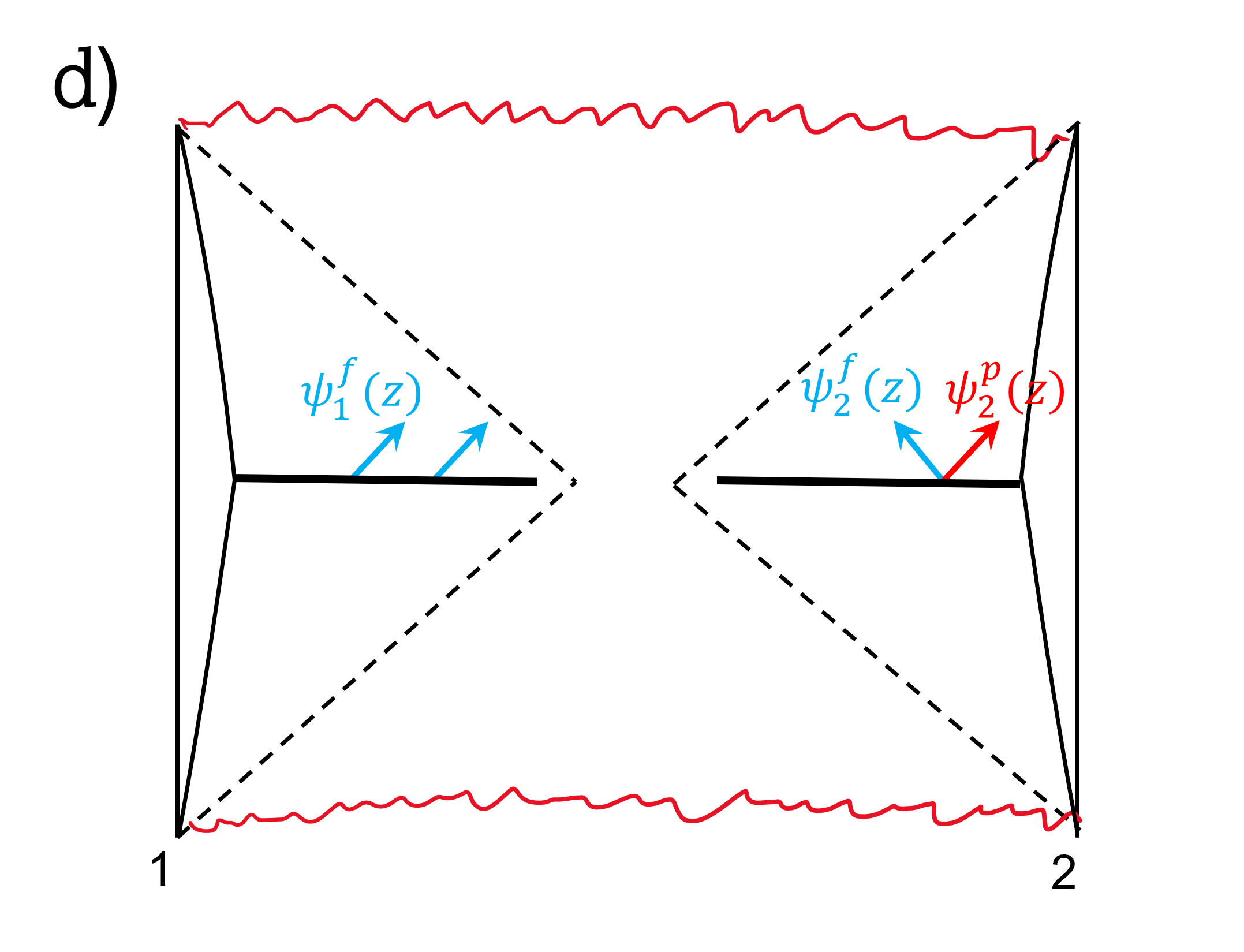}
    \caption{(a) The $\zeta(t)$ function for a $\delta$-function shockwave. (b) The two point function $C_{12}^{fp}(z,z)$ (correlation between blue arrow in wedge $1$ and red arrow in wedge $2$ in (d)) as a function of $z$, for the TFD state (blue) and the shockwave case (red). (c) The mutual information between two regions $[0,z]$ in the two wedges, for the TFD state (blue) and the shockwave case (red). (d) Illustration of the Penrose diagram corresponding to the shockwave geometry. }\label{fig:shockwave}
\end{figure}


The last example we would like to discuss is the eternal traversable wormhole\cite{maldacena2018eternal}, which is the geometry corresponding to the ground state of Hamiltonian~\eqref{eq: Ham_coupledSYK} with a constant $\mu=\frac{\hat{\mu}}q$. The two point function is given by\footnote{A subtle technical point is that the branch of $q$-th root in this equation must be chosen such that the two point function is continuous in $t-t'$.}
\begin{align}
    G_{11}(t,t')&=\left(-\frac{V_G^2}{\sin^2\left(\frac{V_G}2\left(t-t'\right)-i\delta\right)}\right)^{1/q}\label{eq:G11_ETW}\\
    G_{12}(t,t')&=i\left(\frac{V_G^2}{\cos^2\left(\frac{V_G}2\left(t-t'\right)-i\delta\right)}\right)^{1/q}\label{eq:G12_ETW}
\end{align}
with $V_G,~\delta$ determined by Eqs.~(\ref{eq:parameter_ETW1}) and (\ref{eq:parameter_ETW2}). Compared to the previously studied examples, a key difference is that in this geometry, the two point function at long time is oscillatory insteading of decaing, with a period $T=\frac{2\pi}{V_G}$. In the low temperature limit, this system is dual to a global AdS$_2$ geometry, with two reflective boundaries that are in causal contact with one another. A particle from boundary $1$ can propagate to boundary $2$ in time $T/2$, and return to the same boundary in time $T$. 

\begin{figure}
    \centering
    \includegraphics[width=0.2\textwidth]{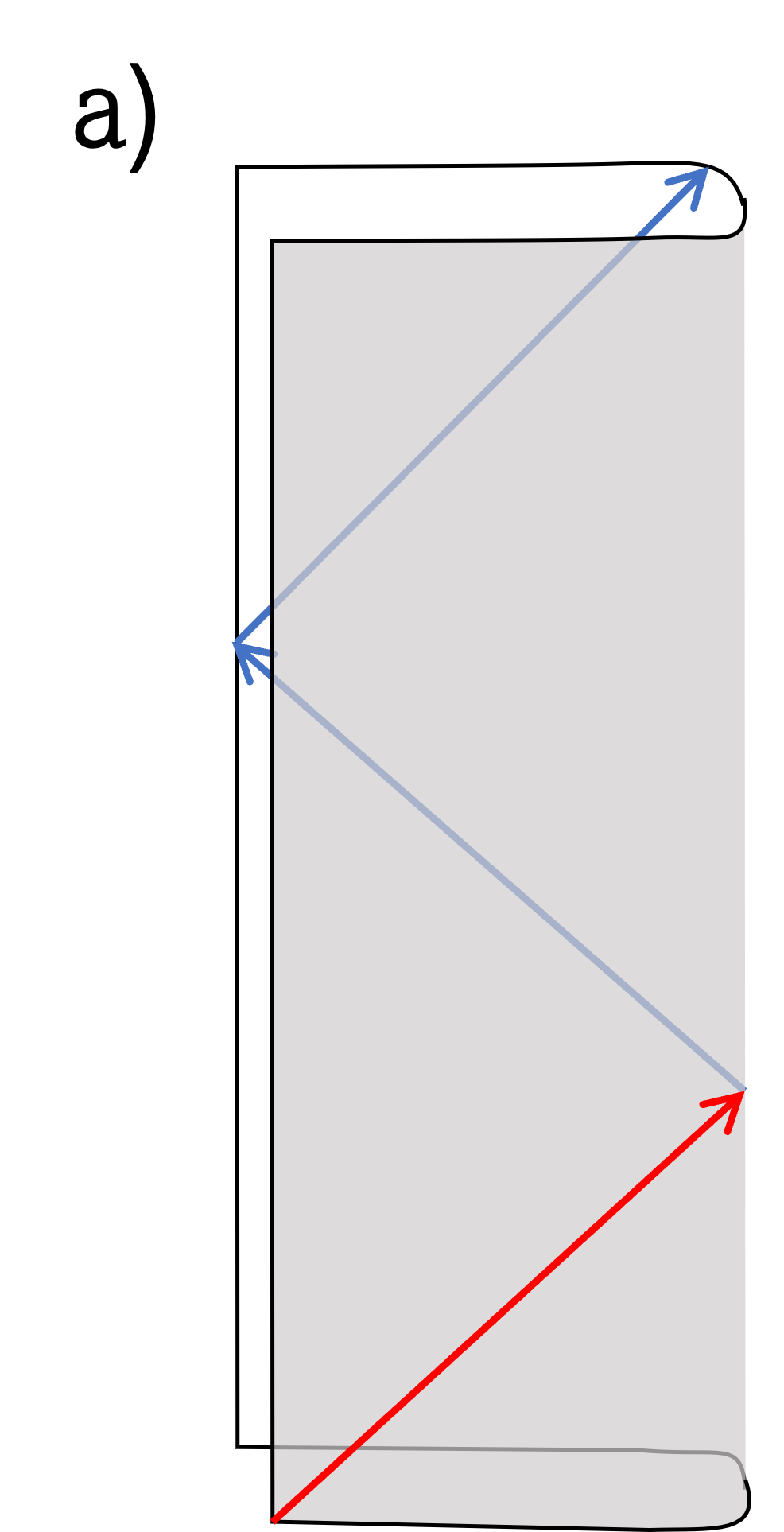}\includegraphics[width=0.23\textwidth]{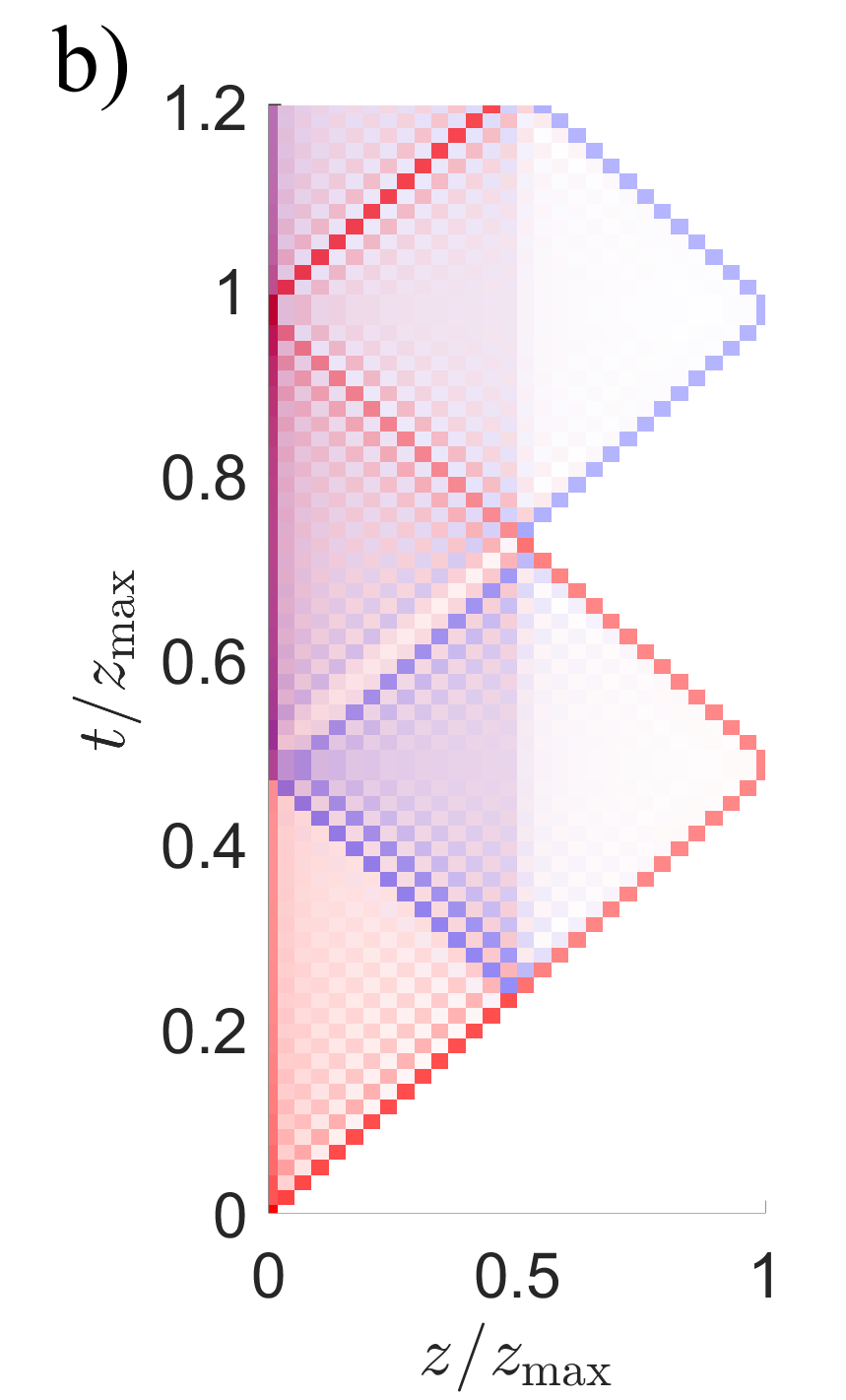}\includegraphics[width=0.5\textwidth]{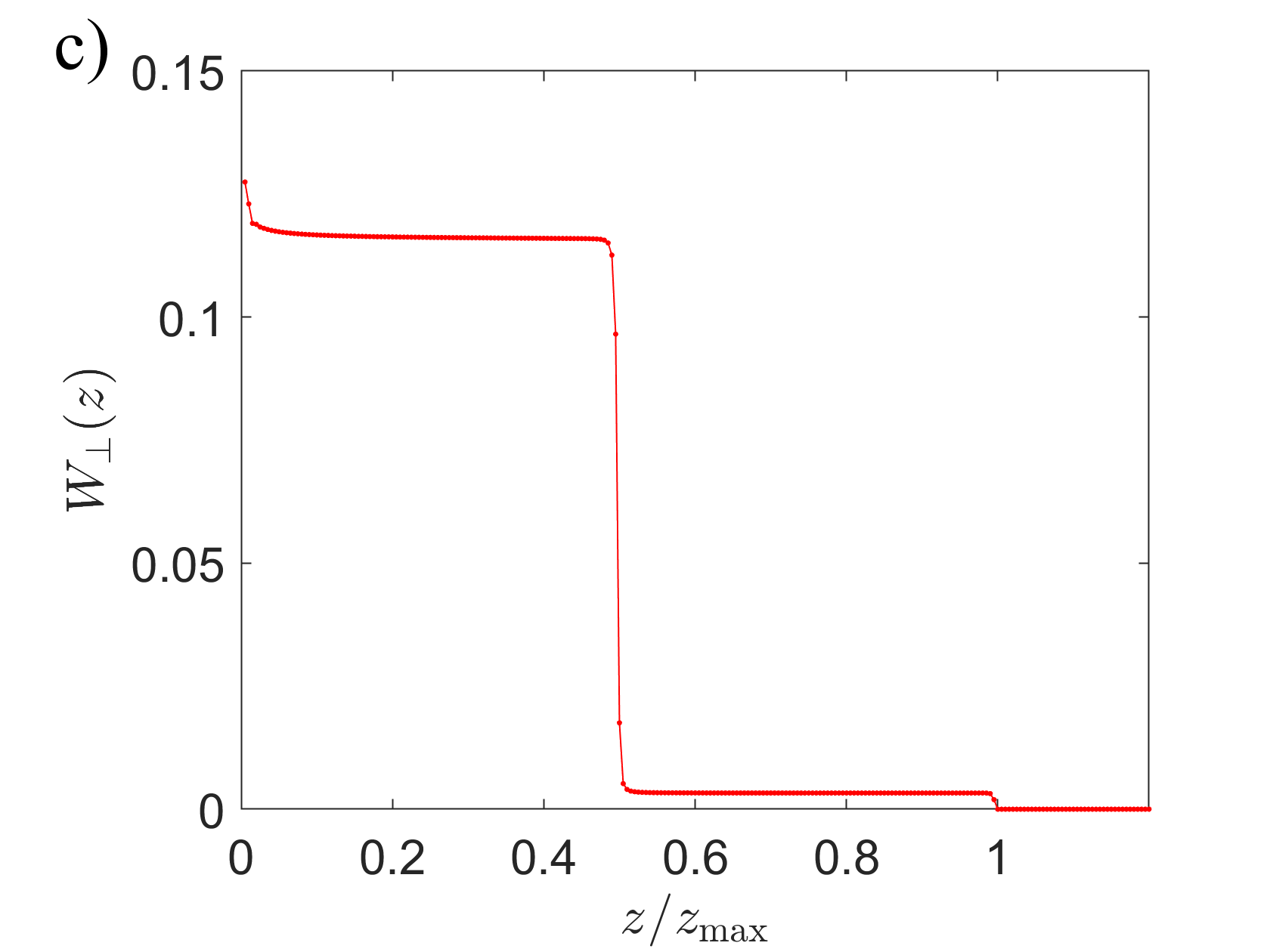}\\
    \centering
    \includegraphics[width=0.5\textwidth]{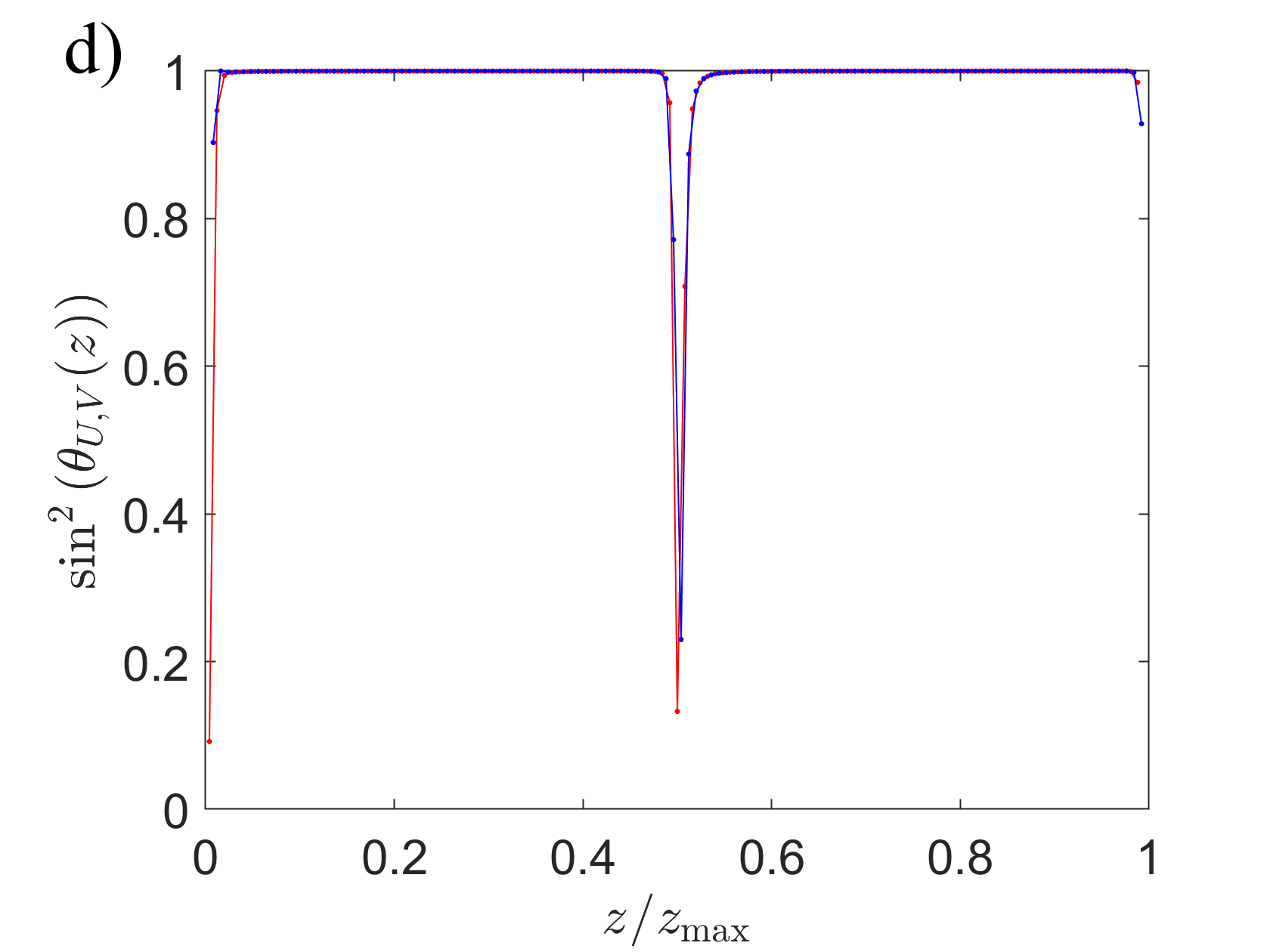}\includegraphics[width=0.5\textwidth]{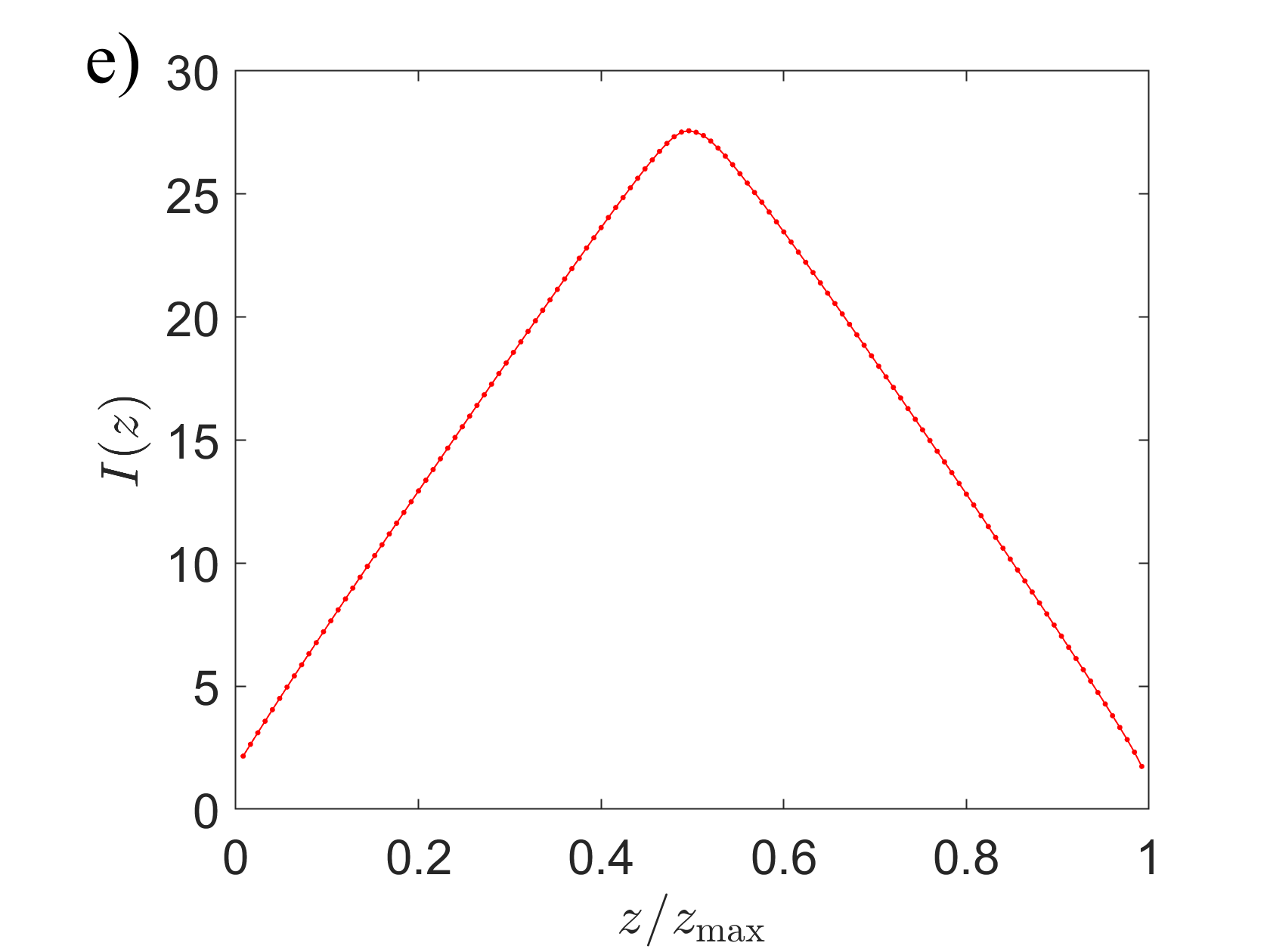}
    \caption{(a) Illustration of the global AdS$_2$ geometry in JT gravity, folded such that two boundaries $1,2$ are both on the left side. A particle starting from boundary $1$ reaches the other boundary at time $T/2$. (b) A pseudocolor plot that shows the propagator $\left\{\chi_1(0),\psi_{1,2}^{f,p}(z,t)\right\}$ as a function of $z,t$. Both left-movers and right-movers in the bulk are assigned the $(z,t)$ coordinate at the center of the corresponding bond. The red and blue color represent flavor $1$ and $2$ respectively. $z=0.5z_{\rm max}$ corresponds to the right end of the JT gravity picture, but in the quantum circuit it is a partial reflection surface. (c) The orthogonal component $W_\perp(t)$ of $\chi_s(t)$ defined with respect to the linear space of operators $\chi_s(t'),t\geq t'<t+2z$. (d) The transmission probability $\sin^2\theta_{U,V}$ for $U$ and $V$ gates. (e) The mutual information between the spatial regions $[0,z]$ (for both flavors $1$ and $2$) and $[z,z_{\rm max}]$. The calculations are done for $\mu=10^{-7}$.}\label{fig:ETW}
\end{figure}

Because of the fractional power $1/q$ in Eqs.~\eqref{eq:G11_ETW} and \eqref{eq:G12_ETW}, the two point function satisfies a twisted boundary condition $G_{ss'}(t+T,t')=G_{ss'}(t,t')e^{-i\frac{2\pi}q}$. As a consequence we obtain the following linear equation
\begin{align}
    A_{12}\left(t,t'+\frac{T}2\right)+\frac{1}{2\sin\frac{\pi}q}\left[A_{11}(t,t')-A_{11}(t,t'+T)\right]&=0\label{eq: linear eq for ETW}
\end{align}
A more detailed discussion is presented in Appendix \ref{app: ETW 2pt}. This linear equation implies that ${\rm det}\left(A\right)=0$ for any time interval with duration $\geq T$. The zero determinant implies that our bulk kernel construction, which requires $KAK^\T=\Id$, can only work until the duration of the time window reaches $T$, or correspondingly $z=z_{\rm max}\equiv \frac{T}2$. In other words, the bulk circuit has a finite width, as is illustrated in Fig. \ref{fig:ETW}(a). At $z=z_{\rm max}$, the right-mover is reflected back as left-mover, in the same way as at $z=0$. This example demonstrates that the topology of bulk spacetime can change depending on the behavior of boundary correlation functions.

In this model, each gate $U(z,t)$ or $V(z,t)$ is a $4\times 4$ matrix that captures "local" inter-site coupling and coupling between the two wedges. To measure the inter-site coupling described by the gate, we denote
\begin{align}
    U(z,t)=\left(\begin{array}{cc}U_{RL}&U_{RR}\\U_{LL}&U_{LR}\end{array}\right)
\end{align}
We consider the off-diagonal $2\times 2$ block $U_{RR}$, which measures the transition amplitude from right-mover to right-mover. We denote the singular values of $U_{RR}$ by $\sin\theta_{Us}(z)$ with $s=1,2$, and similarly for $V(z,t)$. $\theta_{U(V)s}$ are basis-independent measures of the gates. It turns out $\theta_{U(V)s}$ is degenerate in $s=1,2$, which is shown in Fig. \ref{fig:ETW}(d). (At $z=z_{\rm max}$, the gate angle $\theta_U=\theta_V=0$, which is not shown.) 

Interestingly, there is a almost-reflecting wall at $z=z_{\rm max}/2$, where $\theta_{U,V}$ is nearly $0$, but then returns to near $\frac{\pi}2$. This corresponds to the behavior of bulk fermion illustrated in Fig. \ref{fig:ETW}(b). This figure is a color plot of the anticommutator $\left\{\chi_1(0),\psi_{1,2}^{f,p}(z,t)\right\}$ between bulk fermions and a boundary fermion in system $1$ at $t=0$, which is the analog of the earlier figure \ref{fig: kernel}(b), where we now denote the extra flavor index $1,2$ by red and blue. We see that at $z=z_{\rm max}/2$ the bulk fermion from wedge $1$ will transition into wedge $2$ with a high probability. This is the analog of what happens in JT gravity (illustrated in Fig. \ref{fig:ETW}(a)). However, in JT gravity this reflection has probability $1$, but in our model we instead see a finite probability of transmission. The particle passing the surface at $z_{\rm max}/2$ will propagate more or less freely till $z_{\rm max}$, where it will be reflected the same wedge. $z_{\rm max}$ is an analog of "end-of-the-world brane" (EOTW brane) in holographic duality. 

To further understand the EOTW brane, we study the diagonal component (in $s=1,2$) of the kernel of a bulk fermion. The expansion of a future-directed bulk fermion may be written with the diagonal component separated out as:
\begin{align}
    \psi^f_s(z,t)=K(z,t|t-z)\chi_s(t-z)+\sum_{t-z<t'\leq t+z}K(z,t|t')\chi_s(t'),~s=1,2
\end{align}
The inverse of the diagonal component $W_\perp(z)=K^{-1}(z,t|t-z)$ measures the magnitude of the projection of operator $\chi_s(t-z)$ in the subspace $\mathbb{V}_\perp^f$ (see Result \ref{result: orthogonalization}) that is perpendicular to the subspace formed by $\chi_s(t'),~t'\in[t-z+1,t+z]$. Fig. \ref{fig:ETW}(c) captures the behavior of $W_\perp (z)$. We see that $W_\perp(z)$ has a finite almost constant value for $z<z_{\rm max}/2$, and drops to a much smaller but nonzero value for $z\in[z_{\rm max}/2,z_{\rm max}]$. Numerically, the two plateau values of $W_\perp(z)$ are both proportional to $1/q$ in the large $q$ limit. 

We can also study the entanglement entropy in the bulk state. In particular, we study the spatial entanglement between regions $[0,z)$ and $[z,z_{\rm max}]$. As expected, the entire system is in a pure state, and the mutual information $I(z)=S_{[0,z)}+S_{[z,z_{\rm max}]}=2S_{[0,z)}$ is shown in Fig. \ref{fig:ETW}(e). We see that the entanglement entropy is proportional to volume, in a Page-like fashion, but the slope is far from the maximal value ($\log 2$). 

An open question which we reserve for future work is whether there is a physical interpretation of this "doubling" phenomena we observed in term of a modified 2d gravity theory. Using the same technique employed in this subsection, we can study more complicated boundary systems. For example, we can start from a TFD state and turn on a fine-tuned coupling $\hat{\mu}=\hat{\mu}(\beta)$ for a finite time before turning it off. This corresponds to a traversable wormhole that is kept open for a finite time (related to the discussions in Ref.~\cite{gao2017traversable,maldacena2017diving}. Without presenting more details in this work, we will comment that the bulk geometry in this case will contain an end-of-the-world brane that falls into a black hole horizon.

\section{Discussion and Conclusions}\label{sec: discussion}


Let us briefly summarize what we have done. We have developed a general protocol for mapping a theory of (fermionic or bosonic) generalized free fields to canonically-free bulk fields purely from boundary data. For a (0+1)d boundary theory, the boundary two point functions along with a chosen discretization of time are sufficient to uniquely determine the bulk quantum state and dynamics with a corresponding spacetime discretization. The only assumption about the bulk is either the causal structure of the dynamics, or the existence of a linear boundary-to-bulk mapping. Our bulk construction takes the form of a Gaussian quantum circuit, which implements linear evolution of the bulk free fermions or bosons. The bulk two point functions evaluated at the boundary reproduce the boundary model's two point functions, as per required of a duality. Meanwhile, an arbitrary bulk operator in the causal wedge of a boundary region is mapped to a linear superposition of the boundary operators in this region. Our construction is agnostic to the particular type of interacting boundary dynamics, and does not require time translation symmetry or any other symmetry.

Notably, our bulk construction is only enabled when the boundary theory has \textit{interacting} GFFs. The set of single-particle (fermionic or bosonic) operators form an order $\OO(N)$-dimensional subspace of the entire space of operators with exponentially large dimension. Unlike a free theory, the subspace of single-particle operators for an interacting theory is {\it not} preserved by time evolution. The single-particle operators are orthogonal with respect to the anticommutator $\ev{\{.,.\}}_\rho$ for fermions or the commutator $\ev{[.,.]}_\rho$ for bosons, which enables the definition of emergent bulk fields with additional independent degrees of freedom. In other words, the emergent spatial dimension is a geometrization of operator scrambling in the boundary dynamics. Meanwhile, the fact that the dual bulk dynamics are free is a consequence of large-$N$ factorization of the boundary GFFs. 
Our prescription offers a unique and constructive way to determine the bulk free field dynamics \textit{without} making any assumptions about the bulk geometry. The bulk dynamics can then be used as a way to probe the bulk geometry. Our numerical results in large-$q$ SYK model uncover different bulk dynamics for different boundary theory, which generalizes the JT gravity dual in the low energy region. Our results clarified that phenomena such as black hole horizon, bulk negative curvature, end-of-the-world brane, etc. are well-defined at the level of two point functions even if the boundary is not asymptotically conformal invariant. 

We now remark on several important open questions related to this work. First, although our protocol is explicitly realized by discretizing time, the construction of the kernel is mathematically well-defined for continuous time system. The analog of equations \eqref{eq: bulk_ansatz} and \eqref{eq: bulk_alg_cond} can be written down in the continuous case as integral relations on $K$. Explicitly performing the construction in the continuum would enable one to directly determine the bulk dynamics and geometry, which is an avenue we are exploring in further work. 

Secondly, we would like to discuss the effect of considering finite $N$. A crucial aspect of our approach is the ability to perform orthogonalization of the boundary operators in the boundary model to generate the bulk operators. In principle, one can imagine performing an analogous orthogonalization for a finite $N$ theory. However, such a procedure would necessarily need to involve higher-point boundary correlators, and the emergent bulk operators would no longer be linearly related to the boundary operators. Rather, the emergent bulk description would describe interacting bulk fields, whose interactions are suppressed by $\frac{1}{N}$ corrections. For example, the boundary four point function is schematically given by
\begin{align}
    \ev{\chi\chi\chi\chi} = \sum_\text{contractions}\ev{\chi\chi}\ev{\chi\chi} + \frac{1}{N}\ev{\chi\chi\chi\chi}_c + \dots
\end{align}
where all the $\chi$'s generically have different time indices, and $\ev{\chi\chi\chi\chi}_c$ is the connected four point function. The strong hologrpahic principle says that the above expression should be equal to a four point function of the bulk operators evaluated at the boundary. However, the bulk operators will now be some nonlinear function\footnote{For a similar discussion in the usual HKLL protocol, see Ref.~\cite{harlow2018tasi}.} of the boundary operators,
\begin{align}
    \psi = K^{(1)}\chi + K^{(3)}\chi\chi\chi + \dots
\end{align}
Imposing bulk causality in the interacting case amounts to finding a set of coefficients $K^{(n)}$ such that spacelike-separated $\psi$ are still orthonormal. Given given knowledge of the boundary correlators, it is not immediately obvious how to do the orthogonalization in general, or whether this can be done to all orders in $\frac{1}{N}$. However, as causality imposes a rather strong constraint on the operator algebra, we expect that at least perturbatively, this is able to select out a bulk theory. A key question would then be the locality of the bulk theory. Generically, there is no reason to expect the interaction in the bulk will be local, but our approach may lead to a more quantitative probe of the (non)-locality of the bulk theory. It is interesting to relate this direction to the recent works in double-scaled SYK model~\cite{cotler2017black,berkooz2018chord,berkooz2019towards,lin2022bulk}, since the latter provides an example model in which Wick's theorem does not hold, but correlation functions can still be computed exactly.

The third open question is how to generalize our approach to higher dimensions. Given a general boundary model of GFFs in $(d+1)$-dimensions, our approach can be applied directly, whereby the extra boundary spatial coordinates are treated in the same way as the internal flavor indices in the coupled SYK model. If the boundary theory has spatial translation symmetry, we can carry out the bulk reconstruction for each momentum separately, and define bulk fields $\psi_{ia}(\vec{k},z,t)$. The problem is that this construction does not seem to guarantee locality in the tranverse directions. When the boundary is a generic model which could be non-relativistic, this may not be a problem. When the boundary is relativistic with a strict light cone, an open question is whether the emergent bulk also has strict light cone. This discussion is related to entanglement wedge reconstruction, which relates butterfly cones on the boundary to light cones in the bulk~\cite{mezei2017entanglement,qi2017butterfly}. 



We also draw a parallel between the bulk Hilbert space in our construction, and the construction of Hilbert spaces from operator algebras. Because the large-$N$ factorization is only present in a certain state $\rho$ (e.g. in the SYK model, this is the TFD), the bulk Hilbert space is constructed around this state. This is reminiscent of recent results investigating emergent bulk operator algebras in large-$N$ theories~\cite{leutheusser2022subalgebra}. Furthermore, we would like to make some more general connections between our approach and recent work regarding operator algebras and holography~\cite{leutheusser2021causal,leutheusser2022subalgebra}. From the point of view of operator algebra, our construction relies on the qualitative difference between the operator algebras of interacting GFFs and that of free fields or finite-$N$ systems. For finite-$N$ Majorana fermion theories, there is no other operator that anticommutes with all $\chi_i(t=0)$. The operators $\chi_i(t)$ and $\chi_i(0)$ are not independent from each other. Only in the large-$N$ limit do $\chi_i(t)$ and $\chi_i(0)$ become independent, allowing one to define the bulk fermion operators. This is consistent with the operator algebra discussed in Ref. \cite{leutheusser2021causal}. An interesting question is whether we can explicitly follow the approach in this work to obtain time evolution in the black hole interior. From the point of view of the bulk Gaussian quantum circuit, the interior dynamics does not seem to be constrained, since it occurs in the future of all events in the exterior regions. An key question is providing a physical explanation for selecting out the interior dynamics given by the modular operator in Ref. \cite{leutheusser2021causal}. Another question is how to better understand how bulk causality (or more generally, geometric locality) emerges from the operator algebra. In our construction, bulk causality is imposed, and this is sufficiently constraining to pick out an emergent bulk operator algebra. This appears to be consistent with Ref. \cite{leutheusser2021causal} which suggests that causality is not a generic feature of operator algebras, but rather emerges for a special boundary operator algebras. In this reference, the authors argued that sharp bulk light cones can only emerge when the boundary operator algebra is a Type III$_1$ von Neumann algebra and acts on a non-factorizable Hilbert space. While this is consistent with our construction of an bulk operator algebra in the large-$N$ limit when the boundary operators can be shown to be of Type III$_1$, it is not immediately obvious from the perspective of our construction that locality \textit{cannot} be guaranteed for boundary operator algebras of Type II, where the operators are infinite-dimensional but admit a suitable notion of "finite rank" projectors. In other words, it is not obvious whether the ability to execute our orthogonalization procedure necessitates a type III boundary algebra. 

Lastly, we discuss some general remarks and takeaways from our work. To date, holographic duality is only well established as duality between the low-energy Hilbert spaces of perturbative quantum gravity and a large-$N$ gauge theory with conformal symmetry. Prior work seems to suggest the importance of matching the conformal symmetry in the boundary theory to isometries in the gravitational bulk. However, our results show that one can distill an emergent bulk description even in the absence of conformal symmetry in the boundary theory. This suggests that holography can potentially be extended to more general models without a canonical bulk dual (e.g. the SYK model). The ability to conduct the orthogonalization procedure to produce bulk operators can give some insight into the minimal criteria necessary for a theory to contain an emergent local bulk description. Furthermore, because our explicit construction of the bulk naturally gives an isomorphism, one can potentially use our method to construct putative holographic boundary models for a given \textit{bulk} theory, such as matter on a de Sitter background. Finding a "boundary" dual theory of de Sitter space gravity is an active area of research~\cite{maldacena2021two,cotler2021emergent,susskind2021entanglement}, and our approach provides a new angle in investigating this problem. 

{\bf Acknowledgement.} We would like to thank Yingfei Gu, Hong Liu, Yuri Lensky, and Raghu Mahajan for helpful discussions. This work is supported by the National Science Foundation under grant No. 2111998, and the Simons Foundation. This work was partially finished when XLQ was visiting the Institute for Advanced Study, Tsinghua University (IASTU). XLQ would like to thank IASTU for hospitality.


\bibliographystyle{jhep}
\bibliography{SYK_bib}

\appendix
\newpage

\section{QR decomposition}\label{app: qr}
In our bulk construction, we would like to find solutions $K$ to the matrix equation $KAK^\T = \Id$ where $K$ is restricted to be upper or lower triangular for bulk operators on the past and future light sheets. These solutions will correspond to $K_f$ and $K_p$, respectively, as dictated by the ordering choice for the bulk fermions along $\Sigma_f$ and $\Sigma_p$. The ordering choice is exemplified by the following schematic: 
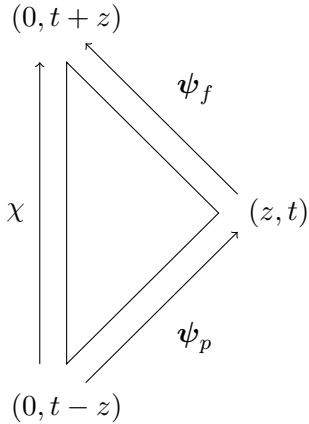
\begin{figure}[H]
    \centering
    \label{fig:ordering_convention}
    \begin{tikzpicture}[every node/.style={inner sep=0.5em}]
    \draw (0,0) node[anchor=north,inner sep=1em]{$(0,t-z)$}
      -- (0,4) node[anchor=south,inner sep=1em]{$(0,t+z)$}
      -- (2,2) node[anchor=west,inner sep=1em]{$(z,t)$}
      -- cycle;

    \draw [->] ({-1/(2*sqrt(2))},0) -- ({-1/(2*sqrt(2))},4) node[midway, left] {$\chi$};
    \draw [<-] (0.25,4.25) -- (2.25,2.25) node[midway,above right] {$\bm\psi_f$};
    \draw [->] (0.25,-0.25) -- (2.25,1.75) node[midway,below right] {$\bm\psi_p$};
    \end{tikzpicture}
    \caption{Schematic of the ordering choice along $\Sigma_f$ and $\Sigma_p$ of the vertex $(z,t)$, as well as the ordering along the boundary interval $I$ in its causal wedge.}\label{fig: ordering}
\end{figure}

Because $A$ is symmetric, it may be diagonalized as $A = VDV^{\T}$ where $D$ is a diagonal matrix and $V$ is orthogonal. Eigenvalues of $D$ are non-negative since $A$ is positive semi-definite. Assuming $A$ is nonsingular, we can define the symmetric square root $A^{-1/2}=VD^{-1/2}V^\T$. The equation $KAK^\T=\Id$ can therefore be written as $Q=KA^{1/2},~QQ^\T=\Id$, where $Q$ is orthogonal. Consequently, we have
\begin{align}
    A^{-1/2}=Q^{f(p)\T}K^{f(p)}\label{eq: QR and QL}
\end{align}
In the ordering convention of Fig. \ref{fig: ordering}, $K_f$ is a upper triangular matrix, so that Eq. (\ref{eq: QR and QL}) is a QR decomposition. Similarly $K_p$ is lower triangular, which is obtained by a QL decomposition of $A^{-1/2}$. (Equivalently, one can obtain $K_p$ by a QR decomposition after reversing the order of both the boundary time and the $\psi_p$ coordinate.) 

\begin{fact}
    For a non-singular matrix $M$, the upper triangular matrix in QR decomposition $M=QR$ is unique up to a sign for each \textit{row} of $R$.
\end{fact} 
\begin{proof}
    Let $M = QR = \Qt\Rt$. Then
    \begin{align}
        \Qt^\T Q = \Rt R^{-1}
    \end{align}
    Since the product of upper right triangular matrices is also upper right triangular, the RHS is upper right triangular. Furthermore, the LHS is orthogonal. The only matrices that are both orthogonal and upper right triangular are the diagonal "sign" matrices
    \begin{align}
        D_\pm \equiv \begin{pNiceMatrix}
            \pm 1 & 0 & \Ldots & & 0 \\
            0 & \pm 1 & \Ddots & & \Vdots\\
            \Vdots & \Ddots &  &  & \\
            & & & & 0 \\
            0 & \Ldots & & 0 & \pm 1
        \end{pNiceMatrix}
    \end{align}
    This can be shown by taking an upper right triangular matrix, and imposing orthonormality of the rows or columns. (e.g. start from the last row, impose normality, then impose orthogonality with the second to last row, impose normality on the second to last row, and so on). Hence, any two QR decompositions are related as
    \begin{align}
        \Qt\Rt = QD_{\pm}R
    \end{align}
    where $\Rt$ differs only by $R$ by a sign choice for each row.
\end{proof}
This implies the kernel is unique up to a sign choice for each row of the kernel.

\section{An alternate derivation of local bulk dynamics}\label{app: bulk_locality}
In this section, we present an alternative argument that leads to the bulk description by a local Gaussian quantum circuit. In Sec.~\ref{sec: framework}, we start by making an ansatz that the bulk theory is given by a certain Gaussian quantum circuit, and then prove that there exists a unique kernel (up to local orthogonal transformations) with compact support that is consistent with the ansatz, and determines the bulk theory from boundary correlation functions. Here we would like to offer a different line of reasoning. We instead start by assuming compact support of the kernel, and that the canonical anticommutation relations hold only within each time slice. From this, we derive that the bulk dynamics are unitary \textit{and} free, and can therefore be represented by a local Gaussian quantum circuit of the form given in Sec.~\ref{sec: framework}. 

As a starting point, consider a boundary model with a discretization of time $t=\tb\Delta t,~\tb\in \mathbb{Z}$, and the anticommutator matrix $A$ defined in Eq.~\eqref{eq: boundary spectral}. We define the bulk fermions $\psi_{iR,L}(\zb,\tb)$ as generators of the linear spaces $\mathbb{W}^{f,p}_{\zb,\tb}$ in Result \ref{result: orthogonalization}. As a reminder, such linear spaces are defined entirely in terms of $\mathbb{V}_{\left[\tb_1,\tb_2\right]}$ and the inner product defined by matrix $A$, which does not require the bulk dynamics to be unitary. With this definition, we can prove the following statements.

\begin{claim}[Orthogonality of spacelike-separated operators]\label{result: locality}
    The bulk fermions constructed using the orthogonalization procedure in Result. \ref{result: orthogonalization} anticommute when the following hold, with the convention that $\zb>\zb'$,
    \begin{align}
        \left\{\psi_{iR}(\zb,\tb),\psi_{jR,L}(\zb',\tb')\right\}&=0,~\text{if}~\zb-\zb'>\tb-\tb'\geq \zb'-\zb\\
        \left\{\psi_{iL}(\zb,\tb),\psi_{jR,L}(\zb',\tb')\right\}&=0,~\text{if}~\zb-\zb'\geq \tb-\tb'> \zb'-\zb
    \end{align}
\end{claim}

\begin{proof}
    Consider the right-movers $\psi_{iR}(\zb,\tb)$, which one can think about as further in the bulk interior than the point $(\zb',\tb')$. These operators span $\mathbb{W}^f_{\zb,\tb}$, which by construction is orthogonal to $\mathbb{V}_{\left[\tb-\zb+1,\tb+\zb\right]}$, denoted by
    \begin{align}
        \mathbb{W}^f_{\zb,\tb}\perp \mathbb{V}_{\left[\tb-\zb+1,\tb+\zb\right]}
    \end{align}
    $\psi_{jR}(\zb',\tb')$ spans space $\mathbb{W}^f_{\zb',\tb'}\subseteq \mathbb{V}_{\left[\tb'-\zb',\tb'+\zb'\right]}$. Similarly, $\psi_{jL}(\zb',\tb')$ spans $\mathbb{W}^p_{\zb',\tb'}\subseteq \mathbb{V}_{\left[\tb'-\zb',\tb'+\zb'\right]}$. When $\zb-\zb'>\tb-\tb'\geq \zb'-\zb$, we have $\tb'-\zb'>\tb-\zb,~\tb'+\zb'\leq \tb+\zb$, such that $\left[\tb'-\zb',\tb'+\zb'\right]\subseteq \left[\tb-\zb+1,\tb+\zb\right]$. Consequently \begin{align}
        \mathbb{W}^{f,p}_{\zb',\tb'}\subseteq \mathbb{V}_{\left[\tb'-\zb',\tb'+\zb'\right]}\subseteq \mathbb{V}_{\left[\tb-\zb+1,\tb+\zb\right]}
    \end{align}
    The above two equations lead to the conclusion 
    \begin{align}
        \mathbb{W}^f_{\zb,\tb}\perp \mathbb{W}^{f,p}_{\zb',\tb'}
    \end{align}
    Following the same reasoning, we can prove $\mathbb{W}^p_{\zb,\tb}\perp \mathbb{W}^{f,p}_{\zb',\tb'}$ for $\zb-\zb'\geq \tb-\tb'> \zb'-\zb$.
\end{proof}

Result \ref{result: locality} clarifies that the bulk has a sharp light cone, like a local quantum circuit. We can also use the linear space formalism to provide a direct proof that there is a unitary gate mapping $\psi_{iR}\left(\zb-\frac12,t-\frac12\right)$ and $\psi_{iL}(\zb,t)$ to $\psi_{iL}\left(\zb-\frac12,t+\frac12\right)$ and $\psi_{iR}(\zb,t)$, as shown below.

\begin{claim}[Unitarity of gates]\label{result: unitarity}
The linear space spanned by $\psi_{iR}\left(\zb-\frac12,\tb-\frac12\right)$ and $\psi_{iL}(\zb,\tb)$ is the same as that spanned by $\psi_{iL}\left(\zb-\frac12,\tb+\frac12\right)$ and $\psi_{iR}(\zb,\tb)$.
\begin{align}
    \mathbb{W}^f_{\zb-\frac12,\tb-\frac12}\oplus \mathbb{W}^p_{\zb,\tb}=\mathbb{W}^p_{\zb-\frac12,\tb+\frac12}\oplus \mathbb{W}^f_{\zb,\tb}\label{eq: equivalence}
    \end{align}
    Consequently, $\psi_{iR}\left(\zb-\frac12,\tb-\frac12\right), \psi_{iL}(\zb,\tb)$ and $\psi_{iL}\left(\zb-\frac12,\tb+\frac12\right), \psi_{iR}(\zb,\tb)$ are two sets of orthonormal bases for the same linear space, and thus must be related by a local unitary linear transformation, which can be taken to be orthogonal since the $\psi_{ia}(\zb,\tb)$ are real. Since this is true for all $\zb,\tb$, we have a local unitary Gaussian quantum circuit.
\end{claim}

\begin{proof}
By definition 
    \begin{align}
        \mathbb{W}^f_{\zb-\frac12,\tb-\frac12}\oplus \mathbb{V}_{\left[\tb-\zb+1,\tb+\zb-1\right]}&=\mathbb{V}_{\left[\tb-\zb,\tb+\zb-1\right]}\\
        \mathbb{W}^p_{\zb,\tb}\oplus \mathbb{V}_{\left[\tb-\zb,\tb+\zb-1\right]}&=\mathbb{V}_{\left[\tb-\zb,\tb+\zb\right]}
    \end{align}
    which implies
    \begin{align}
        \mathbb{W}^f_{\zb-\frac12,\tb-\frac12}\oplus \mathbb{W}^p_{\zb,\tb}\oplus \mathbb{V}_{\left[\tb-\zb+1,\tb+\zb-1\right]}=\mathbb{V}_{\left[\tb-\zb,\tb+\zb\right]}
    \end{align}
    In the same way we can prove
    \begin{align}
        \mathbb{W}^p_{\zb-\frac12,\tb+\frac12}\oplus \mathbb{W}^f_{\zb,\tb}\oplus \mathbb{V}_{\left[\tb-\zb+1,\tb+\zb-1\right]}=\mathbb{V}_{\left[\tb-\zb,\tb+\zb\right]}
    \end{align}
    This proves Eq.~\eqref{eq: equivalence}, since both sides of the equation are equal to the orthogonal complement of $\mathbb{V}_{\left[\tb-\zb+1,\tb+\zb-1\right]}$ in the bigger space $\mathbb{V}_{\left[\tb-\zb,\tb+\zb\right]}$. 
\end{proof}

\section{Bosonic generalized free fields}\label{app: bosonGFF}

In this appendix, we discuss how to carry our bulk construction when the given (0+1)d boundary model describes bosonic GFFs $b_i(t)$ with canonical commutation relations
\begin{align}
    \left[b_i(t),b_j^\dagger(t)\right]=\delta_{ij}
\end{align}
Let $\xi(t)$ be a vector with $2N$ components, 
\begin{align}
    \xi_\mu(t)=\left(\begin{array}{c}b_i(t)\\b_i^\dagger(t)\end{array}\right)
\end{align}
where $\mu = 1,\dots,2N$ denotes operator species or "flavor." Correlation functions of $b_i(t)$ are required to satisfy Wick's theorem. In parallel with the fermionic case, one can prove that the commutator of bosons at different times is proportional to an identity operator (in the large-$N$ limit in which Wick theorem applies, and in the Hilbert space of finite number of excitations). Let $\rho$ be the state for which Wick's theorem holds. Then, we have
\begin{align}
    {\rm tr}\left(\rho O_2\left[\xi_\mu(t),\xi_\nu(t')\right]O_1\right)=\left\langle \left[\xi_\mu(t),\xi_\nu(t')\right]\right\rangle\left\langle O_2O_1\right\rangle
\end{align}
for all $O_1,O_2$ that are finite superpositions of $b,b^\dagger$. We define the matrix
\begin{align}
    C_{\mu\nu}(t,t')=\left\langle\left[\xi_\mu(t),\xi_\nu(t')\right]\right\rangle
\end{align}
Note that $C_{\mu\nu}(t,t)$ evaluated at the same time is no longer an inner product, but a symplectic form:
\begin{align}
    C_{\mu\nu}(t,t)&=J_{\mu\nu}\\
    J&=\left(\begin{array}{cc}\bm 0 &\Id\\-\Id&\bm 0\end{array}\right)
\end{align}
with each entry of $J$ an $N\times N$ block. We would like to construct bulk bosonic fields out of a linear superposition of the $\xi_\mu$, just as in the fermionic case. We again make a linear ansatz
\begin{align}
    \phi_{\mu,L(R)}(\zb,\tb)=\sum_{\tb'\in[\tb-\zb,\tb+\zb]}K_{\mu L(R),\nu}(\zb,\tb|\tb')\xi_\nu(\tb')
\end{align}
and demand that the $\phi_{\mu,L(R)}(\zb,\tb)$ satisfy the correct canonical algebra on a Cauchy slice. The procedure is analogous to the fermionic case where we demand that the bulk operators satisfy the correct algebra on a Cauchy slice, except that the canonical algebra is associated to a symplectic form. In particular, for the future and past light cone, we can write the linear equation in matrix form as
\begin{align}
    \phi^{f(p)}&=K_{f(p)}\xi\\
    K_{f(p)}C K_{f(p)}^\T&=\mathbb{J}
\end{align}
where $\mathbb{J} = J\otimes \Id$ is identity in the spatial indices and symplectic in the flavor indices. The algorithm that determines $K_{f(p)}$ is an orthogonalization procedure that is similar to Result \ref{result: orthogonalization}, but with the inner product replaced by the corresponding symplectic form. Since $C$ is antisymmetric, it may be expressed as 
\begin{align}
    C=\Lambda \mathbb{J}\Lambda^\T
\end{align}
for some unitary $\Lambda$. Then we have
\begin{align}
    Q&\equiv K_{f(p)}\Lambda\\
    Q\mathbb{J}Q^\T&=\mathbb{J}
\end{align}
$Q$ is a symplectic matrix and $Q^{-1}K^{f(p)}=\Lambda^{-1}$ is the analog of the QR decomposition.

The discussion about gate unitary in Appendix \ref{app: bulk_locality} can be straightforwardly generalized to the case of bosonic fields, where the bulk bosons $\phi_{iR}\left(\zb-\frac12,\tb-\frac12\right)$,$\phi_{iL}\left(\zb,\tb\right)$ and $\phi_{iL}\left(\zb-\frac12,\tb+\frac12\right)$, $\phi_{iR}\left(\zb,\tb\right)$ are two different bases of the same linear space, which is the {\it symplectic complement} of $\mathbb{V}_{[\tb-\zb+1,\tb+\zb-1]}$ in the bigger space $\mathbb{V}_{[\tb-\zb,\tb+\zb]}$. Therefore the single-particle transformation of the gate $U(\zb,\tb)$ and $V(\zb,\tb)$ are symplectic matrices.

\section{Comparison between quantum circuits and continuous Majorana fermions}\label{app: circuit_geo}
In this appendix, we compare the Gaussian quantum circuit model with a continuous Majorana fermion theory on a curved (1+1)d background. The goal is to determine the continuous limit of the circuit model, and be able to determine its geometry.

In a two-dimensional curved space with metric 
\begin{align}
    ds^2=\Omega(t,z)^2\left(dt^2-dz^2\right)
\end{align}
the Majorana fermion action has the form
\begin{align}
    \mathcal{A}=\int dtdz\left[\psi_R\left(i\partial_t+i\partial_z\right)\psi_R+\psi_L\left(i\partial_t-i\partial_z\right)\psi_L-2im\Omega(t,z)\psi_L\psi_R\right]
\end{align}
where the conformal factor $\Omega(t,z)$ only appears in the mass term, since the massless part of the Majorana fermion action is conformally-invariant. The Hamiltonian corresponding to this action is
\begin{align}
    H(t)=\int dz\left[\psi_R\left(-i\partial_z\right)\psi_R-\psi_L^\T \left(-i\partial_z\right)\psi_L+2im\Omega(t,z)\psi_L\psi_R\right]
\end{align}
When we compare this continuous theory with the quantum circuit, it is natural to identify the right-mover and left-mover with those in the quantum circuit. If $m=0$, the right-mover and left-mover decouple, and the equation of motion gives $\psi_R(z,t)=\psi_R(z-t)$, $\psi_L(z,t)=\psi_L(z+t)$. This agrees with the quantum circuit time evolution when each gate is given by a swap
\begin{align}
    U=V=\left(\begin{array}{cc}0&1\\-1&0\end{array}\right)
\end{align}
up to a $-1$ sign that we have discussed in the main text. Thus we see that a finite $m$ corresponds to a gate with angle $\theta_{U,V}\neq \frac{\pi}2$. 

To make a more quantitative comparison, we want to compare the unitary evolution operator $e^{-i\Delta tH(t)}$ with two layers of unitary gates in the quantum circuit, consisting of one layer of $V$ gates followed by one layer of $U$ gates. This bilayer of gates is the unit of circuit time evolution. For example, if there is time translation symmetry, time evolution by $t=n\Delta t$ is $(\Ut\Vt)^{n}$ in the quantum circuit, while it is $e^{-in\Delta tH}$ in the continuous theory. As a reminder, $\Ut$ and $\Vt$ denote the many-body operators while $U$ and $V$ are single-particle transformation matrices.

\begin{figure}
    \centering
    \includegraphics[width=0.5\textwidth]{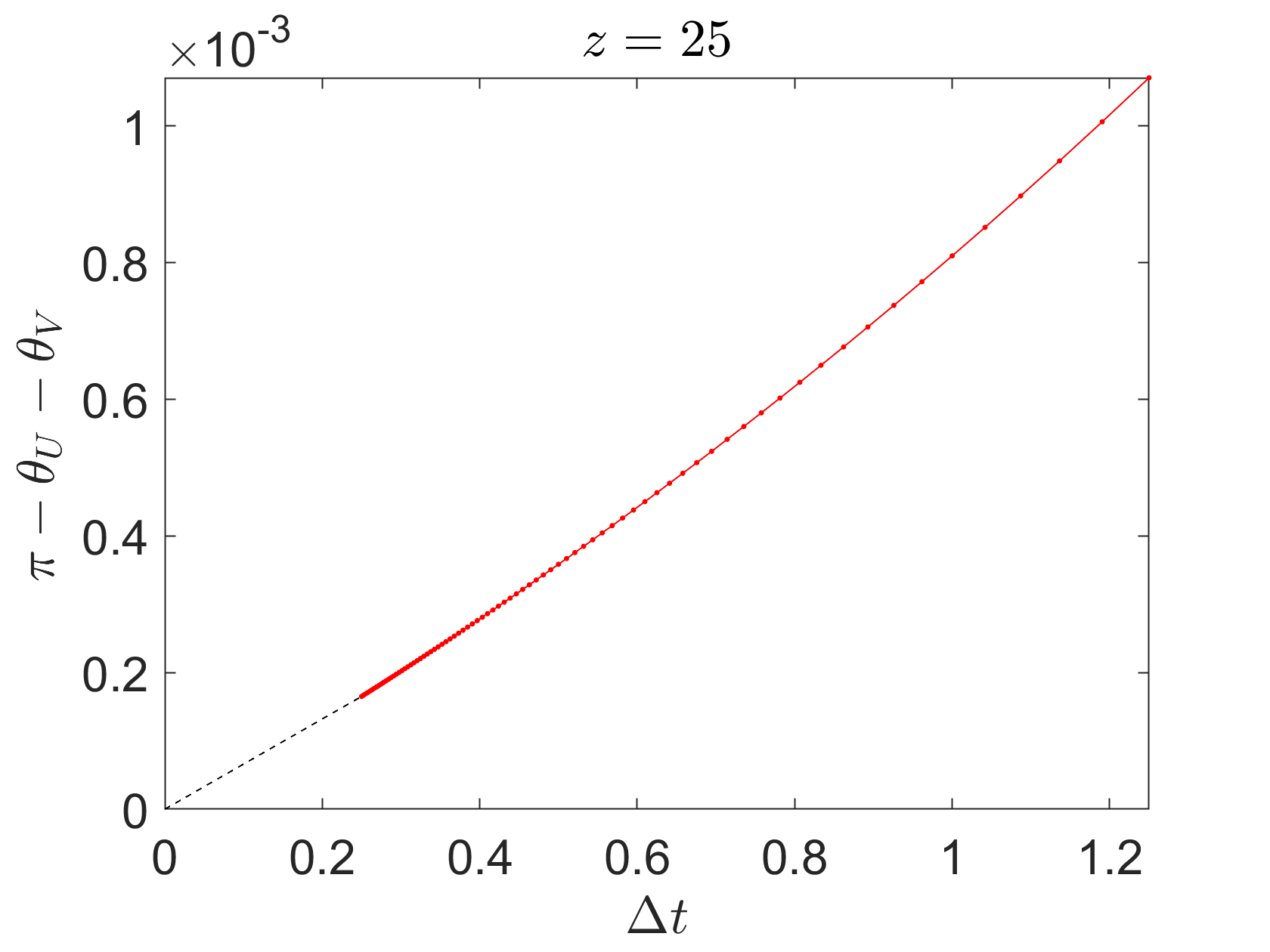}
    \caption{The dependence of $\pi-\theta_U-\theta_V$ as a function of $\Delta t$ for a fixed location $z$. The calculation is done for $q=10, \beta=30$.}\label{fig: theta vs dt}
\end{figure}

To determine the correspondence between the gates and the Hamiltonian dynamics in $\Delta t\rightarrow 0$ limit, we consider a translation invariant Hamiltonian $\Omega(t,z)=\Omega$, and compare it with a translation invariant circuit $U(z,t)=U,~V(z,t)=V$. This is different from our physical situation, but we expect the correspondence between $m$ and gate angles is local and does not depend on the spatial inhomogeneity. In the translation invariant case, we can go to momentum space and write the Hamiltonian as
\begin{align}
    H=\int dk\left(\psi_{R,-k}~\psi_{L,-k}\right)\left(\begin{array}{cc}k&-im\Omega\\im\Omega&k\end{array}\right)\left(\begin{array}{c}\psi_{R,k}\\\psi_{L,k}\end{array}\right)
\end{align}
where $\psi_{L(R),k}$ is the fourier transform of $\psi_{L(R)}(z)$. 
To compare with the quantum circuit, we need to compare the time evolution operator. The single-particle time evolution is given by
\begin{align}
    \left(\begin{array}{c}\psi_{R,k}(\Delta t)\\\psi_{L,k}(\Delta t)\end{array}\right)&=U_c(k,\Delta t) \left(\begin{array}{c}\psi_{R,k}(0)\\\psi_{L,k}(0)\end{array}\right)\\
    U_c(k,\Delta t)&=\exp\left[-i\Delta t\left(\begin{array}{cc}k&-im\\im&-k\end{array}\right)\right]\label{eq: u continuum}
\end{align}
Now we compare $U_c(k,\Delta t)$ of the continuous theory with the bi-layer of unitaries in a translation invariant quantum circuit. We label the bulk fermions in the quantum circuit by the integer sites ($V$ gates in Fig. \ref{fig: circuit_cuts}(a))\footnote{Note that this notation is different from the notation in the main body of the paper. We label $\psi_R(z=n,t=0)$ as $\psi_{R,n}$, and denote $\psi_L\left(z=n-\frac 12,t=\frac12\right)$ as $\psi_{L,n}$, for integer $n$.}. In this convention, the $V$ gates are on-site and $U$ gates couple the neariest neighbor sites. Labeling the \textit{outputs} of all  sites by integers $\psi_{nL},\psi_{nR}$ we have
\begin{align}
    \left(\begin{array}{c}\psi_{R,n}(\Delta t/2)\\\psi_{L,n}(\Delta t/2)\end{array}\right)&\equiv \Vt^\dagger\left(\begin{array}{c}\psi_{R,n}(0)\\\psi_{L,n}(0)\end{array}\right)\Vt=\left(\begin{array}{cc}\sin\theta_V&-\cos\theta_V\\-\cos\theta_V&-\sin\theta_V\end{array}\right)\left(\begin{array}{c}\psi_{R,n}(0)\\\psi_{L,n}(0)\end{array}\right)\\
    \left(\begin{array}{c}\psi_{R,n+1}(\Delta t)\\\psi_{L,n}(\Delta t)\end{array}\right)&\equiv\Ut^\dagger\left(\begin{array}{c}\psi_{R,n}(\Delta t/2)\\\psi_{L,n+1}(\Delta t/2)\end{array}\right)\Ut=\left(\begin{array}{cc}\sin\theta_U&\cos\theta_U\\\cos\theta_U&-\sin\theta_U\end{array}\right)\left(\begin{array}{c}\psi_{R,n}(\Delta t/2)\\\psi_{L,n+1}(\Delta t/2)\end{array}\right)
\end{align}
Carrying out the Fourier transform, the single-particle transformation in $k$-space after combining the $U$ and $V$ gates can be written as
\begin{align}\left(\begin{array}{c}\psi_{R,k}(\Delta t)\\\psi_{L,k}(\Delta t)\end{array}\right)&=U_d(k,\Delta t) \left(\begin{array}{c}\psi_{R,k}(0)\\\psi_{L,k}(0)\end{array}\right)\\
    U_d(k,\Delta t)&=\left(\begin{array}{cc}e^{-ik\Delta t}\sin\theta_U&\cos\theta_U\\\cos\theta_U&-e^{ik\Delta t}\sin\theta_U\end{array}\right)\left(\begin{array}{cc}\sin\theta_V&-\cos\theta_V\\-\cos\theta_V&-\sin\theta_V\end{array}\right)\label{eq: u circuit}
\end{align}
Note that we define $k\Delta t\in[0,2\pi)$ as the crystal momentum, so that $k$ directly corresponds to the physical momentum that we can compare with the continuous theory. We expect $U_d(\Delta t)\simeq U_c(\Delta t)$. To compare these two $SU(2)$ unitaries we expand $U_d$ for small $\Delta t$ and assume $\frac{\pi}2-\theta_{U,V}\propto \Delta t$. More explicitly we assume
\begin{align}
    \frac{\pi}2-\theta_{U,V}\simeq \lambda_{U,V}\Delta t
\end{align}
which is confirmed numerically~(\cref{fig: theta vs dt}). In this limit we can write
\begin{align}
    U_d(k,\Delta t)=e^{-i\frac{k\Delta t}2\sigma_z}e^{-i\sigma_y\lambda_U\Delta t}e^{-i\frac{k\Delta t}2\sigma_z}e^{-i\sigma_y\lambda_V\Delta t}
\end{align}
with $\sigma_y,\sigma_z$ Pauli matrices. To the linear order of $\Delta t$ we can neglect the commutator and obtain
\begin{align}
    U_d(k,\Delta t)\simeq \exp\left[-ik\Delta t\sigma_z-i\left(\lambda_U+\lambda_V\right)\Delta t\sigma_y\right]\label{eq: Ud expansion}
\end{align}
Comparing Eq. (\ref{eq: Ud expansion}) with Eq. (\ref{eq: u continuum}) we see that 
\begin{align}
    m\Omega=\lambda_U+\lambda_V=\lim_{\Delta t\rightarrow 0}\frac{\pi-\theta_U-\theta_V}{\Delta t}
\end{align}
In the region where $\theta_{U,V}$ is a smooth function of $z,t$, in the limit of $\Delta t\rightarrow 0$ we expect the same equation to hold locally:
\begin{align}
    m\Omega(z,t)=\lambda_U+\lambda_V=\lim_{\Delta t\rightarrow 0}\frac{\pi-\theta_U(z,t)-\theta_V(z,t)}{\Delta t}
\end{align}

\section{Zero determinant for the eternal traversable wormhole spectral function}\label{app: ETW 2pt}

In this appendix, we explain why the spectral function $A_{ab}(t,t')$ for the eternal traversable wormhole geometry has a zero determinant for $t'-t\geq T=\frac{2\pi}{V_G}$. The two point function in large-$q$ limit is defined in Eqs.~\eqref{eq:G11_ETW} and \eqref{eq:G12_ETW}. Due to the $1/q$ power, the two point functions are periodic in $qT$, as is shown in Fig. \ref{fig: ETW long time}. 
\begin{figure}
    \centering
    \includegraphics[width = 0.5\textwidth]{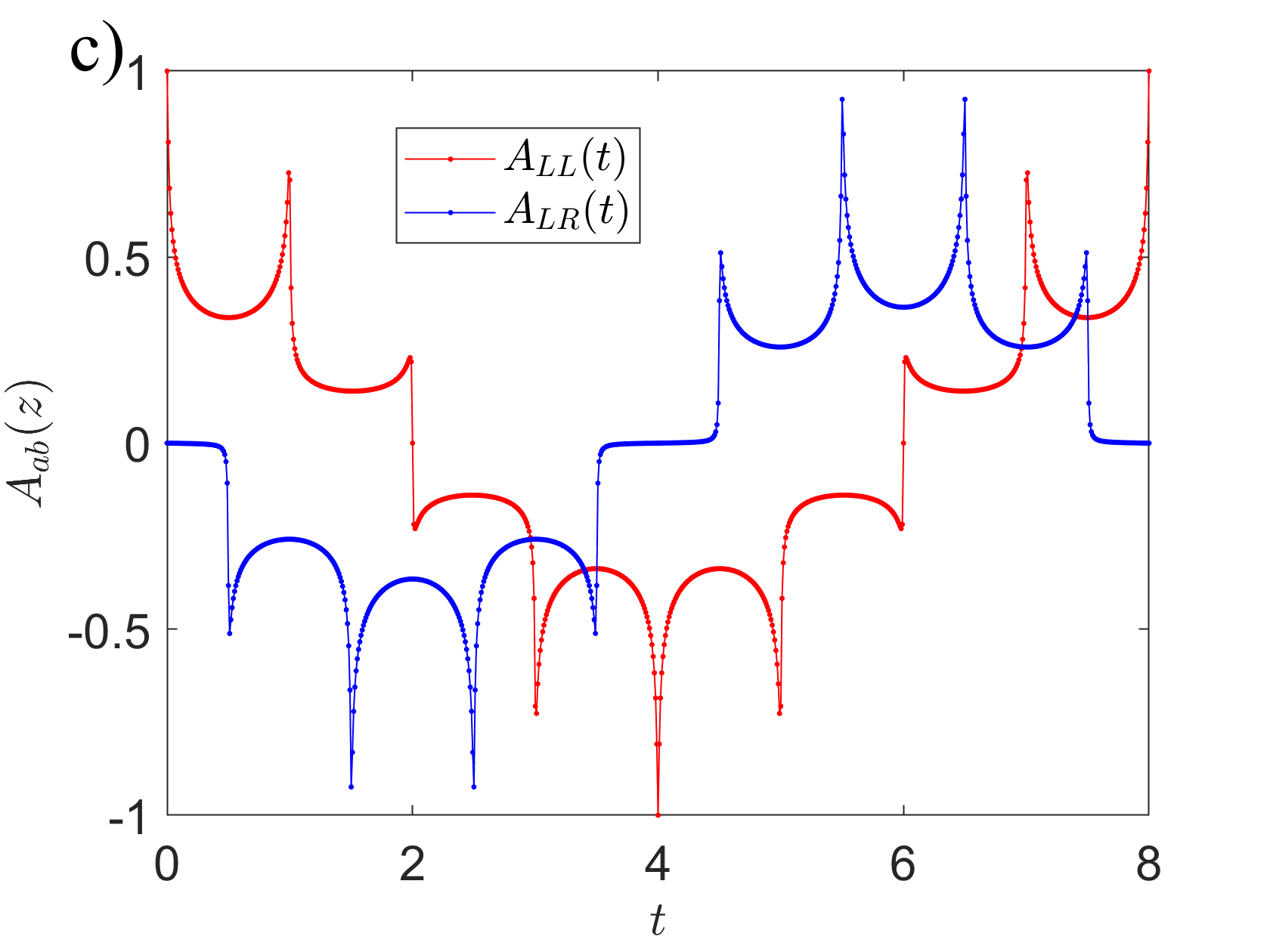}
    \caption{Long time behavior of the anticommutator matrix for eternal traversable wormhole geometry (Eqs.~\eqref{eq:G11_ETW} and \eqref{eq:G12_ETW}) for $q=8$.}\label{fig: ETW long time}
\end{figure}

To understand the behavior of this two point function, denote
\begin{align}
    e^{g(t)}\equiv -\frac{V^2}{\sin^2\left(\frac{V_G}2t-i\delta\right)}=\frac{2V^2}{\cos\left({V_G}t-2i\delta\right)-1}
\end{align}
The denominator can be expanded to
\begin{align}
    \cos\left({V_G}t-2i\delta\right)-1=\cos\left(V_Gt\right)\cosh\delta-1+i\sinh\delta\sin\left(V_Gt\right)\label{eq:gt_ETW}
\end{align}
which has a winding number $1$ around the origin of the complex plane when $t$ grows from $0$ to $T$. Consequently, the two point function $G_{11}(t)$ satisfies the boundary condition $G_{11}(t+T)=G_{11}(t)e^{-i\frac{2\pi}q}$. The branch of $1/q$ power is fixed by the boundary condition $G_{11}(0)=1$, which corresponds to $g(0)=0$. The function $g(t)$ is completely determined by Eq. (\ref{eq:gt_ETW}) together with the boundary condition $g(0)=0$ and the continuity condition in $t$. With this definition of $g(t)$, the diagonal and off-diagonal two point function can both be expressed by $g(t)$:
\begin{align}
    G_{11}(t,t')&=e^{g(t-t')/q},~G_{12}(t,t')=ie^{\left[g\left(t-t'+\frac T2\right)+i\pi\right]/q}
\end{align}
The shift of $i\pi$ is because $g\left(\frac{T}2\right)=-i\pi$, and ${\rm Re}G_{12}(0,0)=A_{12}(0,0)=0$. Consequently, for $t\in[0,T]$ we have the following two equations:
\begin{align}
    G_{11}(t,0)-G_{11}(t,T)&=e^{g(t)/q}-e^{\left[g(t)+2\pi i\right]/q}\\
    G_{12}\left(t,\frac{T}2\right)&=ie^{\left[g(t)+i\pi\right]/q}=-\frac{1}{2\sin\frac{\pi}q}\left[G_{11}(t,0)-G_{11}(t,T)\right]
\end{align}
Therefore we can take the real part and replace $t$ by $t-t'$ to obtain the linear equation
\begin{align}
    A_{12}\left(t,t'+\frac{T}2\right)+\frac{1}{2\sin\frac{\pi}q}\left[A_{11}(t,t')-A_{11}(t,t'+T)\right]&=0
\end{align}
If we consider a time interval with width $\geq T$, this equation is a linear equation satisfied by the matrix $A_{ab}(t,t')$, such that ${\rm det}A=0$.

\end{document}